\documentclass[11pt]{report}
\usepackage[utf8]{inputenc}
\usepackage{graphicx}
\graphicspath{ {images/} }
\usepackage{subfiles}
\usepackage{blindtext}
\usepackage[T1]{fontenc}

\usepackage{amsfonts,amssymb,amsmath,amsthm,longtable,array}
\usepackage{graphics,graphicx,color}
\usepackage{epstopdf,epsfig}
\newcommand{\ket}[1]{ | #1 \rangle}
\newcommand{\bra}[1]{ \langle #1  |}

\newcommand\oa[1]{\text{OA}\left( #1  \right) }
\newcommand\ame[1]{\text{AME}\left( #1  \right) }
\DeclareMathOperator{\tr}{tr}
\DeclareMathOperator{\id}{Id}

\newcommand\GF[1]{\text{GF}(#1)}

\DeclareMathOperator{\PT}{\Gamma}
\DeclareMathOperator{\R}{R}

\usepackage{braket}
\usepackage{subfig}
\usepackage{graphicx}
\usepackage{todonotes}
\usepackage[shortlabels]{enumitem}

\newtheorem{lemma}{Lemma}
\newtheorem{observation}{Observation}
\newtheorem{proposition}{Proposition}
\newtheorem{theorem}{Theorem}
\newtheorem{corollary}{Corollary}
\newtheorem{conjecture}{Conjecture}

\newtheorem{question}{Question}

\theoremstyle{definition}
\newtheorem{definition}{Definition}
\newtheorem{example}{Example}
\newtheorem{remark}{Remark}

\usepackage{lipsum}

\newcommand{\Id}[0]{\text{Id}}
\usepackage{braket}
\usepackage{floatrow} 

\usepackage{cancel}

\newcommand{\blue}[1]{\textcolor{blue}{#1}}

\newcommand{\bs}[1]{\textbf{#1}}
\newcommand{\ea}[1]{\begin{align}#1\end{align}}
\newcommand{\eq}[1]{\begin{equation}#1\end{equation}}

\newcommand{\ma}[1]{\mathcal{#1}}

\newcommand{\SL}[0]{\text{SL}}
\newcommand{\SLC}[0]{\text{SL(}2,\mathbb{C}\text{)}}
\newcommand{\SLIP}[0]{\text{SLIP}}

\usepackage{centernot}
\usepackage{mathtools}
\usepackage{ stmaryrd }

\usepackage{xparse}
\pgfdeclarelayer{myback}
\pgfsetlayers{myback,background,main}

\NewDocumentCommand{\highlight}{O{blue!40} m m}{%
\draw[mycolor=#1] (#2.north west)rectangle (#3.south east);
}

\NewDocumentCommand{\fhighlight}{O{blue!40} m m}{%
\draw[myfillcolor=#1] (#2.north west)rectangle (#3.south east);
}

\PassOptionsToPackage{breaklinks}{hyperref}
\usepackage[colorlinks,allcolors=black]{hyperref}

\usepackage{cleveref}
\crefname{proposition}{Proposition}{Propositions}
\crefname{definition}{Definition}{Definitions}
\crefname{lemma}{Lemma}{Lemmas}
\crefname{figure}{Figure}{Figures}
\crefname{theorem}{Theorem}{Theorems}
\crefname{corollary}{Corollary}{Corollary}
\crefname{conjecture}{Conjecture}{Conjectures}
\crefname{section}{Section}{Sections}
\crefname{appendix}{Appendix}{Appendixes}
\crefname{observation}{Observation}{Observation}
\crefname{question}{Question}{Questions}
\crefname{remark}{Remark}{Remark}
\crefname{example}{Example}{Examples}
\crefname{equation}{Eq.}{Eqs.}
\crefname{table}{Table}{Tables}
\crefname{chapter}{Chapter}{Chapters}
\makeatother

\usepackage{pifont}
\def\club{\ding{168}}
\def\diamond{\color{red}\ding{169}}

\def\spade{\ding{171}}

 \usepackage{tikz}
\usetikzlibrary{fit}
\tikzset{%
  highlight/.style={rectangle,rounded corners,fill=red!15,draw,fill opacity=0.5,thick,inner sep=0pt}
}
\newcommand{\tikzmark}[2]{\tikz[overlay,remember picture,baseline=(#1.base)] \node (#1) {#2};}

\tikzset{%
  highlight/.style={rectangle,rounded corners,fill=red!15,draw,fill opacity=0.5,thick,inner sep=0pt}
}
\newcommand{\Highlight}[1][submatrix]{%
    \tikz[overlay,remember picture]{
    \node[highlight,fit=(left.north west) (right.south east)] (#1) {};}
}
\newcommand{\tikzmarkk}[2]{\tikz[overlay,remember picture,baseline=(#1.base)] \node (#1) {#2};}

\tikzset{%
  highlightt/.style={rectangle,rounded corners,fill=blue!15,draw,fill opacity=0.3,inner sep=-0.2pt}
}
\tikzset{%
  highlighttt/.style={rectangle,rounded corners,fill=red!15,draw,fill opacity=0.3,inner sep=-0.2pt}
}
\newcommand{\Highlightt}[1][submatrix]{%
    \tikz[overlay,remember picture]{
    \node[highlightt,fit=(up.north west) (down.south east)] (#1) {};}
}
\newcommand{\Highlighttt}[1][submatrix]{%
    \tikz[overlay,remember picture]{
    \node[highlighttt,fit=(up.north west) (down.south east)] (#1) {};}
}

\usetikzlibrary[positioning]
\usetikzlibrary{patterns}
\newcommand{\daywidth}{1.3 cm}

\usetikzlibrary{shapes.arrows}
\usepackage{centernot}
\usepackage{mathtools}
\usepackage{ stmaryrd }
\usepackage{xparse}
\usetikzlibrary{matrix,backgrounds}
\pgfdeclarelayer{myback}
\pgfsetlayers{myback,background,main}
\tikzset{mycolor/.style = {line width=1bp,color=#1}}%
\tikzset{myfillcolor/.style = {draw,fill=#1}}%

\DeclareMathOperator{\OA}{OA}

\usepackage{pdfpages} 

\begin{document}
 
\begin{titlepage}
	\centering
	{\scshape\large\bfseries Jagiellonian University, Krak\'{o}w\par}
	\vspace{0.1cm}
	{\scshape\large\bfseries Faculty of Physics, Astronomy and Applied Informatics\par} 
	\vspace{0.9cm}
	\includegraphics[scale=0.1]{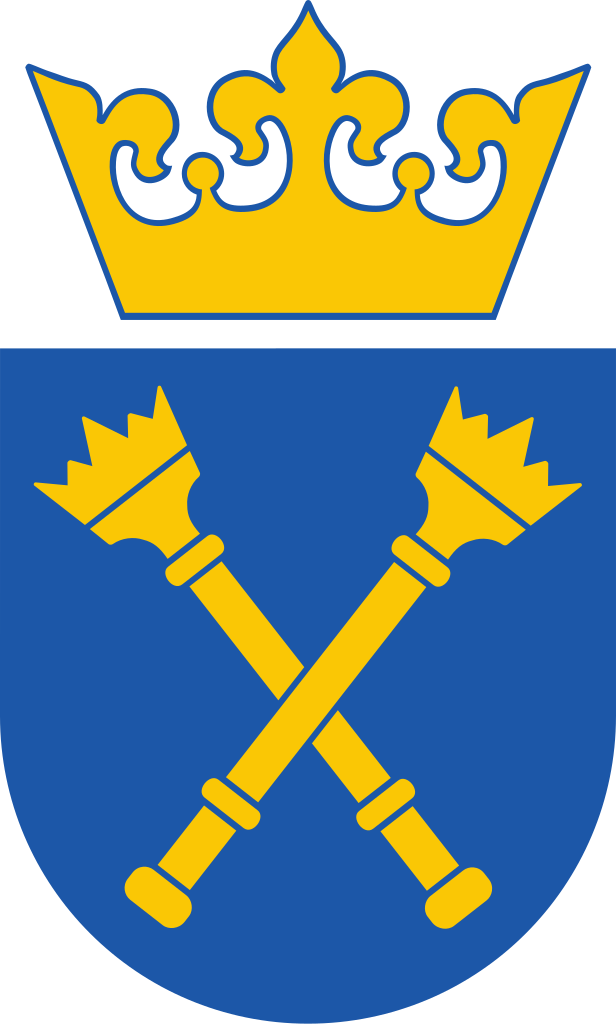}
	\par\vspace{0.9cm}
	
	{\Large\scshape\bfseries PhD dissertation \par}
	\vspace{0.1cm}	
	{\scshape\large\bfseries Physical sciences\par}
	\vspace{0.1cm}	
	{\scshape\large\bfseries Physics\par}
	\vspace{1.6cm}

	{\Large\bfseries Symmetry and Classification of Multipartite Entangled States\par}
	\vspace{1.3cm}
	
	{\huge Adam Burchardt\par}
	\vspace{1.6cm}
	
	{\bfseries Supervisor\par}
	\vspace{0.1cm}
	{\Large Prof. dr hab. \textsc{Karol \.{Z}yczkowski}} \par
	\vspace{1.9cm}
	

	{\large Krak\'{o}w, September 2021}
\end{titlepage}


\newpage
\begin{titlepage}
	\begin{flushright}
	\mbox{}
     \vfill
	{\Large \textit{cioci Basi}\par}
	\end{flushright}
\end{titlepage}
\begin{titlepage}
	\begin{flushright}
	\mbox{}
     \vfill
	{\Large \textit{Separable states are all alike; every entangled state is entangled in its own way.}\par}
	\vspace{1cm}
	{\Large — Lew Tolstoy, Anna Karenina (Paraphrasis)\par}
	\end{flushright}
	\end{titlepage}
\begin{titlepage}
$\quad$
\vskip 7cm

\indent
I am immensely grateful to all people who have been in my life for the last three years. 
First of all, to my supervisor, Karol {\.Z}yczkowski, who introduced me to the world of quantum information, which I liked so much. Thank you for the scientific problems you presented to me, for the advice you gave me, and for the time you devoted to me. 

Thank you to all colleagues and collaborators with whom I have had the opportunity to work, talk, and chat over the past years:
Zahra Raissi, Gon\c{c}alo Quinta, Rui Andr\'{e}, Jak{\'{o}}b Czartowski, Marcin Rudzi{\'{n}}ski, Grzegorz Rajchel-Mieldzio{\'{c}}, Wojtek Bruzda, Suhail Ahmad Rather, Arul Lakshminarayan, Baichu Yu, Pooja Jayachandran, Valerio Scarani, Yu Cai, Nicolas Brunner, Pawe\l{} Mazurek, M\'at\'e Farkas, Jerzy Paczos, Marcin Wierz\-bi{\'{n}}\-ski, Waldemar K{\l}obus, Mahasweta~Pandit, Tam\'as V\'ertesi, Wies{\l}aw Las\-kow\-ski, Adrian Ko{\l}odziejski, Felix Huber, Markus Grassl, Dardo Goyeneche, M\={a}ris Ozols, Timo Simnacher, Antonio Ac\'{i}n, Albert Rico Andr\'{e}s, Gerard Angles Munne, Kamil Korzekwa, Oliver Reardon-Smith, Roberto Salazar, Alexssandre de Oliveira, Micha{\l} Eckstein, Simon Burton,  John Martin, David Lyons, and Fereshte Shahbeigi, among many others. 
I express my special thanks to two of my colleagues: Konrad Szyma{\'{n}}ski and Stanis{\l}aw Czach{\'{o}}rski, who greatly helped me in the last days of preparing this dissertation. 

I would like to thank Agata Hadasz, Agnieszka Hakenszmidt and Agnieszka Golak for all their help during my stay in Krakow. 

All of this would not have been possible without the help of my family, friends, and beloved.
\end{titlepage}

\chapter*{Abstract}
One of the key manifestations of quantum mechanics is the phenomenon of quantum entanglement. While the entanglement of bipartite systems is already well understood, our knowledge of entanglement in multipartite systems is still limited. This dissertation covers various aspects of the quantification of entanglement in multipartite states and the role of symmetry in such systems. Firstly, we establish a connection between the classification of multipartite entanglement and knot theory and investigate the family of states that are resistant to particle loss. Furthermore, we construct several examples of such states using the Majorana representation as well as some combinatorial methods. Secondly, we introduce classes of highly-symmetric but not fully-symmetric states and investigate their entanglement properties. Thirdly, we study the well-established class of Absolutely Maximally Entangled (AME) quantum states. On one hand, we provide construction of new states belonging to this family, for instance an AME state of 4 subsystems with six levels each, on the other, we tackle the problem of equivalence of such states. Finally, we present a novel approach for the general problem of verification of the equivalence between any pair of arbitrary quantum states based on a single polynomial entanglement measure.


\newpage
\subsection*{This PhD Dissertation is based on the following publications and preprints available online}
\begin{enumerate}
\item[{[A]}]
G. M. Quinta, R. Andr\'{e}, \textbf{A. Burchardt}, and K. \. Zyczkowski
\textit{Cut-resistant links and multipartite entanglement resistant to particle loss}, 
\href{https://journals.aps.org/pra/abstract/10.1103/PhysRevA.100.062329}{Phys. Rev. A 100, 062329 (2019)}. 
\item[{[B]}]
\textbf{A. Burchardt}, Z. Raissi, 
\textit{Stochastic local operations with classical communication of absolutely maximally entangled states}, 
\href{https://journals.aps.org/pra/abstract/10.1103/PhysRevA.102.022413}{Phys. Rev. A 102,022413 (2020)}.
\item[{[C]}]
\textbf{A. Burchardt}, J. Czartowski, and K. \. Zyczkowski, 
\textit{Entanglement in highly symmetric multipartite quantum states},
\href{https://journals.aps.org/pra/abstract/10.1103/PhysRevA.104.022426}{Phys. Rev. A 104, 022426 (2021)}.
\item[{[D]}]
S. Rather${}^\ast$, \textbf{A. Burchardt}${}^\ast$, W. Bruzda, G. Rajchel-Mieldzio{\'c}, A. Lakshminarayan, K. {\.Z}yczkowski, 
\textit{Thirty-six entangled officers of Euler: Quantum solution to a classically impossible problem}, 
\href{https://journals.aps.org/prl/abstract/10.1103/PhysRevLett.128.080507}{Phys. Rev. Lett. 128, 080507 (2022)}. 
\\${}^\ast$Contributed equally
\item[{[E]}]
\textbf{A. Burchardt}, G. M. Quinta, R. Andr\'{e}, 
\textit{Entanglement Classification via Single Entanglement Measure}, 
\href{https://arxiv.org/abs/2106.00850}{ArXiv: 2106.00850 (2021)}.
\end{enumerate}

\newpage
\subsection*{List of Abbreviations}

\begin{table}[h!]
\begin{tabular}{ll}
AME & Absolutely Maximally Entangled    \\
LM &  Local Monomial  \\
LU &  Local Unitary  \\
MOLS & Mutually Orthogonal Latin Square     \\
OA & Orthogonal Array     \\
OLS & Orthogonal Latin Square     \\
QECC & Quantum Error Correction Code \\
SL & Special linear \\
SLIP &  SL-symmetric Invariant Polynomial  \\
SLOCC &  Stochastic Local Operations with Classical Communication  \\
\end{tabular}
\end{table}

\setcounter{tocdepth}{1}
\tableofcontents
 
\chapter{Introduction}

In 1935, Albert Einstein, Boris Podolsky, and Nathan Rosen studied some counterintuitive predictions of quantum mechanics about strongly correlated systems, known today as Einstein–Podolsky–Rosen paradox (EPR paradox) \cite{PhysRev.47.777}. 
The EPR paradox aroused great interest among physicists. 
One of the concerned was Erwin Schrödinger, who in a later letter to Albert Einstein coined the word \textit{entanglement} to describe this phenomenon. 
The phenomenon of quantum entanglement occurs when a collection of distinct particles interact in such a way that the resulting quantum state cannot be described independently for each particle, exactly as it is observed in the EPR paradox. The first experiment that verified entanglement was successfully corroborated by Chien-Shiung Wu and Irving Shaknov in 1949 \cite{PhysRev.77.136}. This result specifically proved the quantum correlations of a pair of photons. Since then, quantum entanglement has been demonstrated experimentally in many other systems, as with photons \cite{physRevLett.18.575}, neutrino,\cite{PhysRevLett.117.050402}, electrons \cite{verification13}, large molecules as buckyballs \cite{Arndt1999,Nairz03}, and even small diamonds \cite{Lee1253}. The EPR paradox and notion of entanglement became a milestone for the development of quantum mechanics and our understanding of the world.  Moreover, quantum entanglement has become the heart of a completely new and dynamically developing field of science lying in the intersection of Quantum Physics and Information Theory: Quantum Information Theory.

\subparagraph*{Qubits.}
A central notion in Quantum Information Theory is \textit{qubit}, an abstract term, distilled from various concrete physical realizations. Arguably, the simplest quantum-mechanical system is a two-level (or two-state) system. Such a system might be physically realized as the spin of the electron, which in a given reference frame is either up or down, or as the polarization of a single photon, with distinguished vertical and horizontal polarization.  In an abstract way, a pure qubit state is a coherent superposition of the aforementioned distinguished basis states $\ket{0},\ket{1}$, i.e. a linear combination 
\begin{equation}
\label{eqQubit}
\ket{\psi}=\alpha \ket{0}+ \beta \ket{1}, 
\quad\quad 
|\alpha|^2+|\beta |^2 =1.
\end{equation}
Such an object is nowerdays refered as \textit{qubit} - quantum binary unit \cite{PhysRevA.51.2738}, a quantum analogue of a classical bit. In a classical system, a bit is always in precisely one of two states, either $0$ or $1$. However, quantum mechanics allows the qubit to be in a coherent superposition of both states simultaneously, a property that is fundamental to quantum mechanics and quantum computing.


\subparagraph*{Two-qubit system.} 
The most common variant of the EPR paradox (formulated by Bohm \cite{Bohm1,PhysRev.108.1070,RevModPhys.81.1727}) was expressed in terms of the quantum mechanical formulation of spin. In the more recent notations, it might be explained as a particular state of a two-qubit system, i.e. a linear combination of four vectors $\ket{00},\ket{01},\ket{10},\ket{11}$ representing directions of both spins. So-called EPR state might be thus written as
\begin{equation}
\ket{\text{EPR}}=\dfrac{1}{\sqrt{2}} (\ket{00} +\ket{11}).
\end{equation}
Such state is strongly correlated in any reference frame, in a way that classical correlations do not allow. 

While studying quantum entanglement, it is useful to abstract from certain local properties of a state which do not affect global entanglement and quantum correlations. An example of this type of abstraction is a Local Unitary (LU) equivalence of two given states. Having say that, it is known that any two-qubit state is LU-equivalent to the following system
\begin{equation}
\ket{\psi} = \sqrt{p} \ket{00}+\sqrt{1-p} \ket{11},
\label{EPRp}
\end{equation}
for some parameter $p \in [0,\frac{1}{2} ]$. In particular, for $p=0$ the state is separable, while for $p=\frac{1}{2}$ the state is maximally entangled, and coincides with the famous EPR state. In such a way, we obtained a satisfactory quantification of entanglement in two-qubit systems. Indeed, the exact amount of entanglement is measured by two related values $(p, 1-p)$, which are also known as Schmidt coefficients \cite{RanksReducedMatrices}. The closer the values $p$ and $1-p$ are to each other, the more entangled the state is. 

Furthermore, the number of non-vanishing values among $p, 1-p$ is known as Schmidt rank, which is equal to $1$ for separable states and equal to $2$ for entangled states. This results in the coarse-grained classification of entanglement in two-qubit systems. Indeed, with respect to the Schmidt rank two classes of states are distinguished: separable and entangled. This straightforward division into two classes of bipartite states coincides with the division under another class of local operations: Stochastic Local Operations with Classical Communications (SLOCC) \cite{ThreeQub}. SLOCC operations include not only local unitary rotations but also additions of ancillas (i.e. enlarging the Hilbert space), measurements, and throwing away parts of the system, each performed locally on a given subsystem \cite{ThreeQub,PhysRevA.63.012307}. As it was shown, mathematically SLOCC operations are represented via \emph{invertiable} operators \cite{ThreeQub}. SLOCC operators cannot generate entanglement between subsystems, however, they might enhance or strengthen the existing entanglement with some non-vanishing probability of success.  This is reflected in the fact that there are only two states which are not equivalent with respect to SLOCC transformations: separable $\ket{00}$ state and entangled $\ket{\text{EPR}}$ state. 

An important feature of both types of local operations: LU and SLOCC are that each of them provides an equivalence relation on the state-space, and hence divides it into equivalence classes \cite{LUupto5qubits}. Any two states from one class are interconvertible by an adequate local operator, while such a transformation cannot be provided for states from different classes. In that way, quantum states which exhibit the same (for LU) or similar (for SLOCC) entanglement properties are grouped together, which gives the solution for the problem of \emph{entanglement classification} of two-qubit systems. 

As we already discussed, the problem of entanglement classification of two-qubit states is well-understood. Indeed, there are infinitely many LU-equivalence classes, indexed by one real parameter $p\in [ 0 ,\frac{1}{2}]$, see \cref{EPRp}. On the other hand, there is only one SLOCC-equivalence class of genuinely entangled states, which might be represented by the EPR state.

In 1997 Scott Hill and William K. Wootters introduced the notion of \textit{Concurrence} $C$, an entanglement measure, which for any two-qubit pure state state 
$
\ket{\psi} = c_{00} \ket{00}+c_{01}  \ket{01} +c_{10} \ket{10}+c_{11}  \ket{11} 
$ 
reads 
\begin{equation}
\label{conLag}
C (\ket{\psi} ):= 2|c_{00}c_{11}-c_{01}c_{10} | \in [0,1]
\end{equation}
as an absolute value of the degree two polynomial in the state-coefficients \cite{PhysRevLett.78.5022}. 
As it was observed, the value of concurrence is not only independent of LU operation performed on any qubit, but also on any local Special Linear operation (SL), i.e. local invertible matrices with determinant one. Any invertible operation might be presented as SL operation up to the global constant. In that way, local SL operations are relevant to SLOCC operations on qubits. The notion of concurrence was an important step towards a modern approach to the problem of entanglement quantification via SL-\textit{polynomial invariant} ($\SLIP$), measures. $\SLIP$ measures are at the same time polynomials in the state coefficients and invariant under SLOCC operations \cite{Eltschka_2014}, which significantly facilitates the problem of quantifying entanglement resources.

\subparagraph*{Three-qubit system.}
In the late 1980s Daniel Greenberger, Michael Horne and Anton Zeilinger studied three-qubit entangled state the tripartite generalisation of the EPR pair \cite{greenberger2007going}:
\[
\ket{\text{GHZ}}=\dfrac{1}{\sqrt{2}} (\ket{000}+\ket{111} ).
\]
Bouwmeester et al. firstly presented the experimental observation of the state above as a polarization entanglement for three spatially separated photons \cite{PhysRevLett.82.1345}. Since then, many other experimental realisations of GHZ state were performed. From the theoretical point of view, there are two important features of GHZ state. On one hand it is \textit{maximally entangled}, i.e each of its reductions to the one-particle subsystem is maximally mixed, which is often a desired property. On the other hand, however, state GHZ is not robust against particle loss. Indeed, all its reductions to the two-particle subsystem form an unentangled mixed state, in other words its two-particle correlations are of a classical nature. 

There is yet another particular three-partite entangled state
\[
\ket{\text{W}}=\dfrac{1}{\sqrt{3}} (\ket{001}+\ket{010}+\ket{100}  ),
\]
introduced by in 2000 by Wolfgang Dür, Guifre Vidal, and Ignacio Cirac in their seminal paper ``Three qubits can be entangled in two inequivalent ways'' \cite{ThreeQub}. Contrary to the GHZ state, the W state is robust against the particle loss, indeed, each of its two-particle subsystems exhibits quantum correlations. On the other hand, the one-particle subsystems are not maximally-mixed, as it was for GHZ state.

As it was shown, the aforementioned GHZ and W states are the only two distinct SLOCC-non-equivalent states of genuinely entangled three-qubit states \cite{ThreeQub}. Furthermore, there are infinitely many LU-equivalence classes of genuinely entangled three qubit states, parametrized by three real parameters \cite{ThreeQub,threeQubits}. 

In 2000 Valerie Coffman, Joydip Kundu, and William K. Wootters introduced the first entanglement measure related to the $3$-body quantum correlations in the system: the \emph{three-tangle} $\tau^{(3)}$ \cite{DistributedEntanglement}. Three-tangle distinguishes between GHZ and W states and achieves two extreme values $\tau^{(3)}(\text{GHZ}) =1$ and $\tau^{(3)}(\text{W}) =0$ on both states. While the GHZ state exhibits maximal $3$-body and vanishing $2$-body quantum correlations, the quantum correlations in the W state are reversed. As it was shown, the three-tangle is a SLIP measure, hence invariant under SLOCC operations. In that way, the three-tangle provides a satisfactory method for distinguishing between different SLOCC classes of genuinely entangled three-qubit states.

\begin{figure}[h!]
  \centering
  \includegraphics[width=0.4\textwidth]{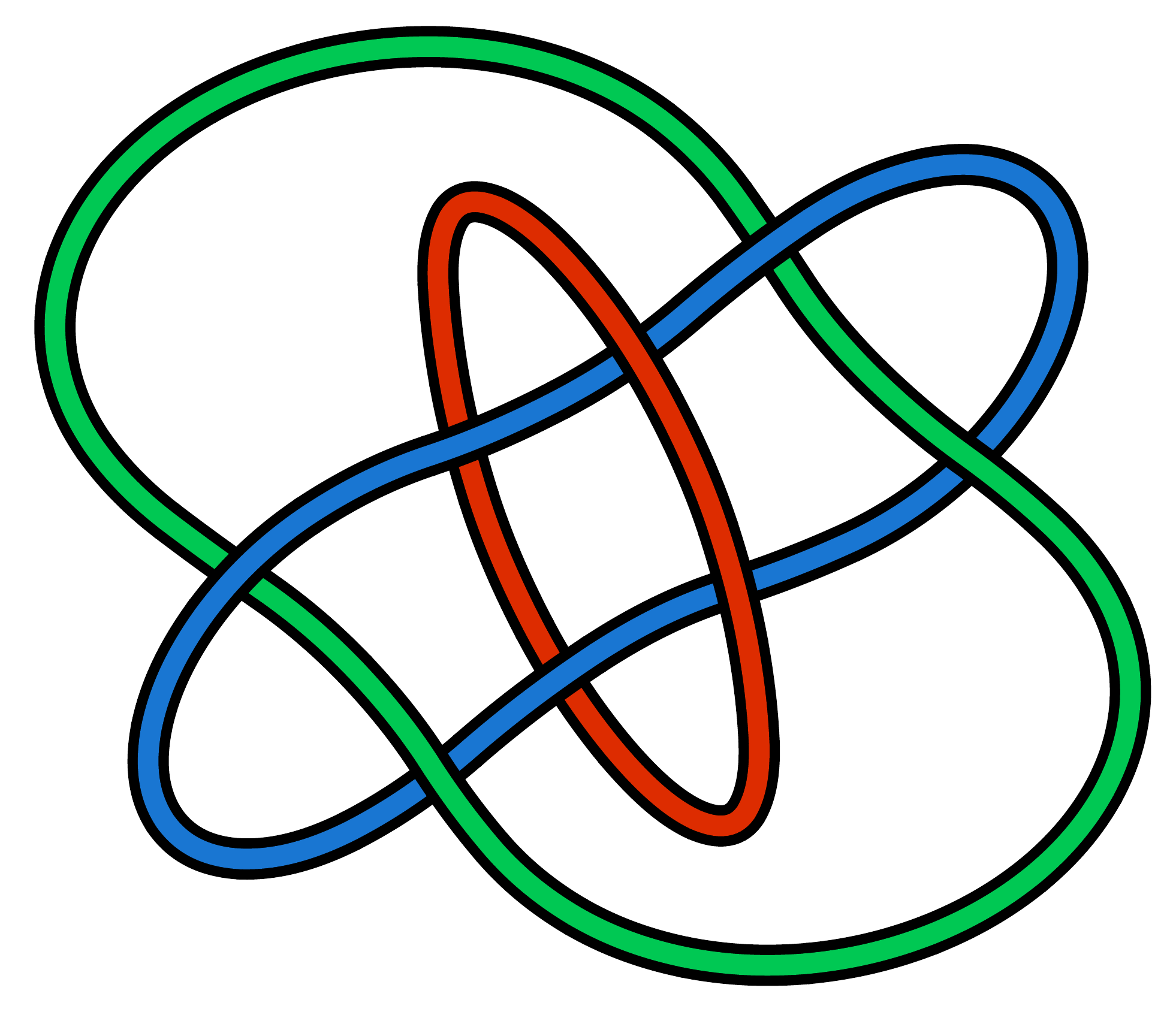}
  \hspace{1cm}
    \includegraphics[width=0.4\textwidth]{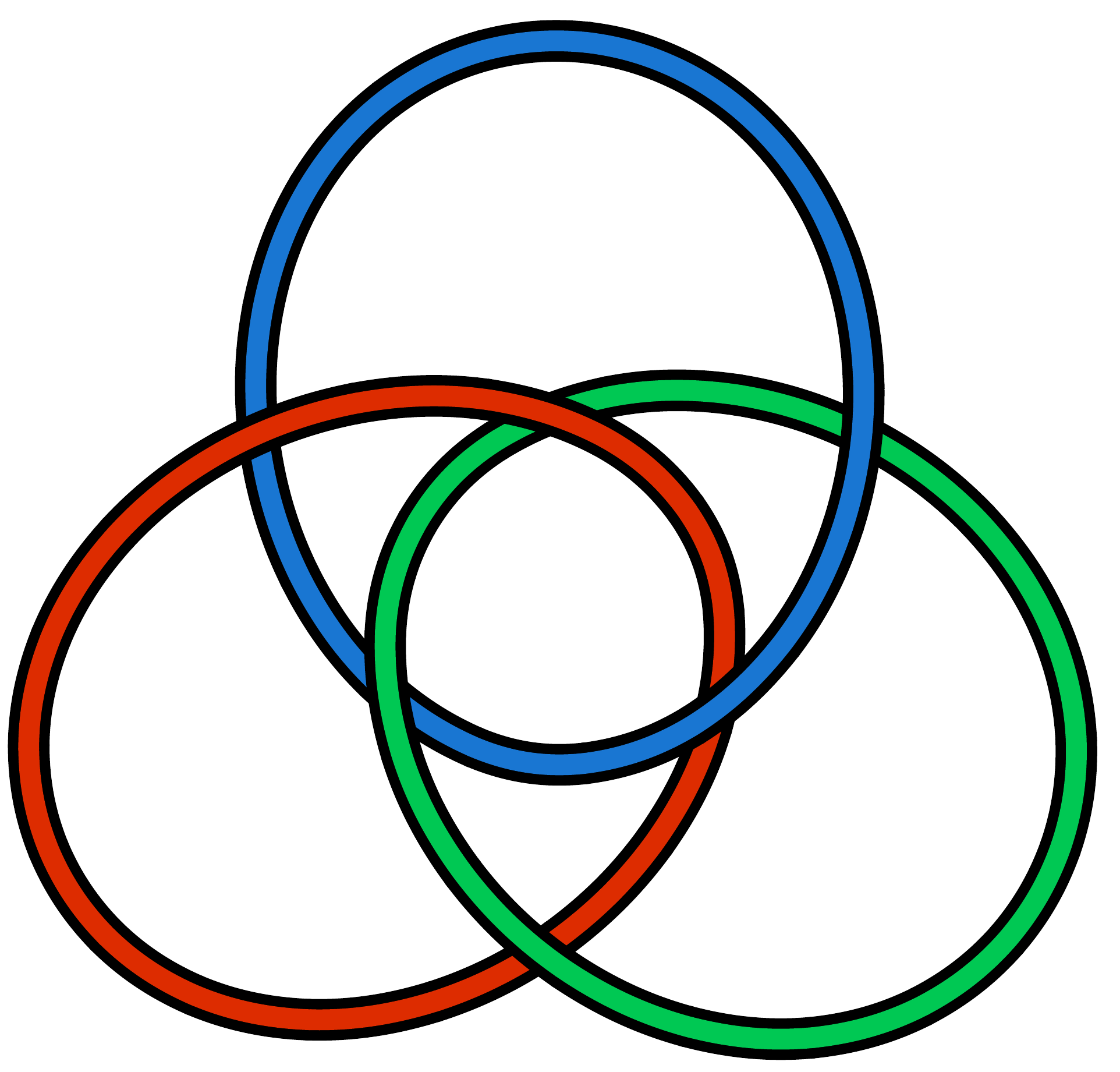}
      \caption{Two distinct three-partite qubit states: GHZ and W represented in the form of links. On the left \textit{Borromean rings}, characterized by the fact that if one ring is cut, the remaining two become disconnected. This resembles entanglement in GHZ state, which is not robust against particle-loss. On the right another link of three knots, for which cut of any component does not disconects the link itself. This is relevant to the robustness of the W state.}
      \label{34}
\end{figure}

\subparagraph*{Four-qubit system.}
If the discussion of entangled in three-qubit systems did not convince the reader that entanglement in the multipartite system has a rich and complex form, we shall proceed to the four qubit states. As we shall see, a system of four qubits is already large enough to reveal two substantial problems concerning multipartite-entanglement: the existence of maximally entangled states and the classification of entanglement. In the following, I discuss briefly both problems.

Atsushi Higuchi and Anthony Sudbery in their seminal paper ``How entangled can two couples get?'' investigate entanglement in system of four qubits concerning bipartitions of the system \cite{HIGUCHI2000213}. In particular they showed that there is no four-partite qubit system, which is maximally entangled with respect to any of bipartitions for two-partite systems: $12|34$, $13|24$, and $14|23$. Furthermore, they showed that the following state
\begin{equation}
\ket{M_4} =   \dfrac{1}{\sqrt{6}} \Big( \ket{0011} + \ket{1100} +  e^{\frac{2 \pi i}{  3}} (\ket{1010}+ \ket{0101}) + e^{\frac{4 \pi i}{ 3}} (\ket{1001} + \ket{0110}) \Big),
\end{equation}
maximize the average entropy of entanglement for such bipartitions. As it was later observed, state $\ket{M_4} $, as well its conjugate transpose exhibits intriguing symmetric properties \cite{Lyons_2011,CenciCurtLyons}. Contrary to tripartite states $\ket{\text{GHZ}}$ and $\ket{\text{W}}$, state $\ket{M_4}$ is not fully permutation-invariant. Nevertheless state $\ket{M_4}$ is invariant under any even permutation of qubits. Such states for which any element of an \textit{alternating subgroup} of the permutation group, $A_4<S_4$, leaves the state invariant, shall be called $A_4$-symmetric.

Contrary to the case of three-qubits, there are infinitely many SLOCC classes of four-qubit states \cite{threeQubits}. Despite this obvious obstacle, the problem of classification of entanglement in four-qubit states arouses great interest \cite{FourQubits,Djokovic4qubits,Discussion4qubit,4QubitsAllClasses,PolInv4qubits,SLOCCallDim,SpeeKraus,FourQubits8}. The four-qubit states were successfully divided into nine families, most of which contain an infinite number of SLOCC-classes \cite{FourQubits,Djokovic4qubits,SpeeKraus}. In particular, the so called $G_{abcd}$ family with the most degrees of freedom is represented by the following highly-symmetric states
\begin{align*}
\ket{G_{abcd}}&=\frac{a+d}{2} \big(\ket{0000}+\ket{1111} \big)+\frac{a-d}{2} \big(\ket{0011}+\ket{1100} \big)
\\
&+\frac{b+c}{2}  \big(\ket{0101}+\ket{1010} \big)+\frac{b-c}{2} \big(\ket{0110}+\ket{1001} \big).
\end{align*}
Furthermore, non-SLOCC-equivalent four qubit states can be effectively distinguished by three \textit{independent} $\SLIP$ measures \cite{PolInv4qubits,PolynomialInvariantsofFourQubits}.

\subparagraph*{General $N$-qubit system.}
Quantum states with non-equivalent entanglement represent distinct resources and hence may be useful for different protocols. The idea of clustering states into classes exhibiting different qualities under quantum information processing tasks resulted in their classification under SLOCC. Such a classification was successfully presented for two, three and four qubits \cite{ThreeQub,FourQubits,SLOCCallDim,4QubitsAllClasses}. However, the full classification of larger systems is completely unkown. Even the much simpler problem of detecting if two $N$-qubit states ($N>4$) are SLOCC-equivalent is, in general, quite demanding \cite{PolInv4qubits,Zhang_2016,KempfNessToEntanglement,
BurchardtRaissi20,sowik2019link}.

It was shown that for multipartite $N\geq 4$ systems there
exist infinitely many SLOCC classes \cite{ThreeQub}. In particular, the set of equivalence classes under SLOCC operations depends at least on $2(2^N-1)-6N$ parameters, which grows exponentially with the number of qubits \cite{ThreeQub}. Even larger number of LU-equivalence classes was also investigated \cite{PhysRevLett.104.020504,RaG,PhysRevLett.108.060501}.

So far, great effort has been made to reduce this verification problem to the problem of computing entanglement measures, which by definition take the same value for the equivalent states \cite{PolynomialInvariantsofFourQubits,GT09,_okovi__2009, PolInv4qubits,Li_2013, Szalay_2012,JAFARIZADEH2019707,compass1,Compass2}. Nevertheless, starting from four-partite states, one needs to compute values of at least three independent measures, to decide with certainty on local equivalence. Furthermore, the number of independent measures grows exponentially with the number of qubits \cite{ThreeQub}, making it intractable to use this approach to discriminate locally equivalent states \cite{Love07}. In fact, this rapid increasing difficulty to verify the entanglement equivalence between two states is common to any procedure \cite{Love07}.


\subparagraph*{Qutrits and beyond.} 
So-far, in our study of entanglement we focused on the multipartite qubit systems, consisting of two-level subsystems. 
These theoretical considerations go hand in hand with experiments, since manipulation and relative control over several qubits have already become a standard task \cite{PhysRevLett.106.130506,Riofr_o_2017}. Even though two-level subsystems are the most common multipartite states, they are certainly not the most general from both a theoretical and experimental point of view. Recently, much attention has been paid to the \textit{qutrites}. Qutrits are realized by a 3-level quantum system, that may be in a superposition of three mutually orthogonal quantum states \cite{Nisbet_Jones_2013}. Similarly to \cref{eqQubit}, a qutrit state might be written in the form
\begin{equation*}
\ket{\psi}=\alpha \ket{0}+ \beta \ket{1}+ \gamma \ket{2}, 
\quad\quad 
|\alpha|^2+|\beta |^2 +|\gamma |^2=1.
\end{equation*}
Although precise manipulation of such states turned out to be a demanding task, there are some successful indirect methods for it \cite{PhysRevLett.100.060504}. Furthermore, there is an ongoing development of quantum computers using qutrits and qubits with multiple states \cite{Nisbet_Jones_2013}. 

Overall, there is no reason to restrict the number of levels of each subsystem in the general considerations of multipartite quantum states \cite{10.3389/fphy.2020.589504}. There are several advantages of using larger dimensional subsystems in quantum computations, including the increasing the variety of quantum gates available, applications to adiabatic quantum computing devices \cite{Amin2013AdiabaticQO,Zobov2012ImplementationOA,Peng_2008}, or to topological quantum systems \cite{Cui_2015,Bocharov16}.

In this thesis, beside of multipartite qubit systems, I consider homogeneous multipartite qudit systems, i.e. systems of $N$ particles in which each particle has the same number of levels $d$ ($d=2$ for qubits, $d=3$ for qutrits, etc.). 
The heterogenous multipartite systems are also being investigated \cite{Huber_2018}, and some effort has been made to classify them \cite{PhysRevA.67.012108,Miyake_2004,Hyperdeterminant2}.

\subparagraph*{Absolutely Maximally Entangled states.} 
Besides the difficult problem of the entanglement classification in the $ N $ qudit system, we prompt a simpler question about which states represent maximum entanglement in such a system. This question is ambiguous, there is no unique way of generalizing the maximally entangled state of two-qubit into any other system. In particular, according to different entanglement measures for multipartite states (like the tangle, the Schmidt measure, the localizable entanglement, or geometric measure of entanglement), the states with the largest entanglement do not overlap in general \cite{LUupto5qubits}. One of the possible interpretations of states with largest entanglement, however, was signalized by the work of Higuchi and Sudbery \cite{HIGUCHI2000213}. It resulted in the well-established notion of \textit{Absolutely Maximally Entangled} (AME) states \cite{HelwigAME}. AME states are those multipartite quantum states that carry absolute maximum entanglement for all possible partitions. 
AME states are being applied in several branches of quantum information theory: in quantum secret sharing protocols \cite{HelwigAME}, in parallel open-destination teleportation \cite{helwig2013absolutely}, in holographic quantum error correcting codes \cite{Pastawski2015HolographicQE}, among many others. 
Different families of AME states have been introduced \cite{Rains1999NonbinaryQC,Helwig2013AbsolutelyME} and the problem of their existence is being investigated \cite{PhysRevA.69.052330,HIGUCHI2000213,Felix72}. 
It has been demonstrated that the simplest class of AME states, namely AME states with the \emph{minimal support}, is in one-to-one correspondence with the classical error correction codes \cite{AME-QECC-Zahra} and combinatorial designs known as \emph{Orthogonal Arrays} \cite{ComDesi}. 
Henceforward, the two-way interaction with combinatorial designs and quantum error correction codes is observed \cite{PhysRevA.69.052330,DiK}. 
AME states are special cases of \emph{k-uniform} states characterized by the property that all of their reductions to k parties are maximally mixed \cite{kUNI}.

\subparagraph*{Structure and aims of the thesis.} 
As we have demonstrated so far, the quantification and classification of entanglement for multipartite states is an involving long-distance project. Characterization of different classes of entanglement in multipartite quantum systems remains a major issue relevant for various quantum information tasks and is interesting from the point of view of foundations of quantum theory. In this dissertation, we will show the progress in several directions of these complex issues. The main goals and results of the thesis are presented in the six next chapters of the dissertation. We briefly present the scope of the proceeding chapters and indicate the main results below. If not specified differently, the author's contributions to the work covered by this chapter were significant.

\subparagraph*{Chapter 2}constitutes a link between the classification of multipartite entanglement and knot theory. We introduce the notion of \textit{$m$-resistant states}, states which remains entangled after losing an arbitrary subset of $m$ particles, but becomes separable after losing any number of particles larger than $m$. I construct several families of $N$-particle states with the desired property of entanglement resistance to particle loss.

\subparagraph*{Chapter 3}discusses highly-symmetric states and their properties. Inspired by remarkable entanglement properties of a state $\ket{M_4}$, I introduce the notion of $G$-symmetric states of $N$-qubits, for any subgroup of the permutation group symmetric $G<S_N$. I present two methods for constructing such states: first based on the group and second on graph theory. I propose quantum circuits efficiently generating such states, and Hamiltonians with the ground states related to them. In addition, these states were experimentally simulated on available quantum computers: IBM – Santiago, Vigo, and Athens. 

\subparagraph*{Chapter 4}investigats a class of AME states and $k$-uniform states. I briefly recall correspondence between AME states and classical combinatorial designs focusing attention on the different linear structures of classical designs. 

\subparagraph*{Chapter 5}settles a long-standing problem in multipartite quantum entanglement: of whether absolutely maximally entangled states exist for all 4 party states, with local dimensions greater than 2. We settle this positively and I provide an explicit analytical example for local dimension 6, which was the sticking point. Presented work is connected to the famous Euler problem of the non-existence of orthogonal Latin squares in dimension 6, involving 36 officers from 6 different regiments and ranks. This result shows that if officers of different ranks and regiments can be entangled, the classically impossible problem has a quantum solution. I discuss how presented solution opens doors to the new area dubbed as quantum combinatorics. The aforementioned issue concerning the existence of the state presented in this thesis appears on at least two open-problem lists of quantum information.

\subparagraph*{Chapter 6}concerns the local equivalence of AME states. All reduced density matrices of AME states are maximally mixed. Therefore, the classical method for verification of local equivalence, which is comparison of Schmidt rank and coefficients, fails. I present general techniques for verifying either two AME states are locally equivalent. I falsify the conjecture that for a given multipartite quantum system all AME states are locally equivalent. I also show that the existence of AME states with minimal support of 6 or more particles results in the existence of infinitely many such non-equivalent states. As an immediate consequence, I show that not all AME states belong to the class of stabilizer states, which was supposed by several experts in the field. Moreover, I present AME states which are not locally equivalent to the existing AME states with minimal support.

\subparagraph*{Chapter 7}tackles a particularly discernible problem of discrimination and classification of multipartite entanglement. So far, states were compared by computing values of independent entanglement measures. Nevertheless, the number of independent measures grows exponentially with the number of qubits, making it intractable to use this approach to discriminate locally equivalent states. I provide a complementary approach and use a single entanglement measure for verification if generic multipartite states are locally equivalent. This approach can be used independently on the number of qubits and bypasses the exponential difficulty of standard procedures. In essence, I investigate the roots of an entanglement measure and show that the roots of locally equivalent states must be related by a Möbius transformation, which is straightforward to verify. Moreover, I apply the presented approach on $4$-qubit states and show that roots of the 3-tangle measure of the related reduced $3$-qubit states, might be used for the more demanding problem of classifying entanglement in $4$-qubit states. This gives hope to apply this method to classify entanglement in larger ($N>4$) systems, for which the classification is completely unknown. In addition, I have developed a novel method to obtain the normal form of a $4$-qubit state which bypasses the possibly infinite iterative procedure.

\subparagraph*{Chapter 8}summarizes results obtained in the dissertation and outlines the open problems and perspectives for further research.


\chapter{Quantum states and topological links}
\label{chap1}
In this chapter, I introduce the notion of \textit{$m$-resistant states}. We say that an entangled quantum state of $N$ subsystems is $m$-resistant if and only if it remains entangled after losing an arbitrary subset of $m$ particles, and becomes separable after losing any number of particles larger than $m$. Furthermore, I establish an analogy to the problem of designing a topological link consisting of $N$ rings characterized by the fact that by cutting any $(m + 1)$ of them, the remaining link becomes unknotted. This analogy allows us to exhibit several $N$-particle states with the desired property of entanglement resistance to particle loss. An extension of some parts of this chapter, in which the author’s contribution was not substantial can be found in the joint work \cite{PhysRevA.100.062329}. 

\section{Motivation}

An effort to study multipartite entanglement using the knot theory was proposed by Aravind, who related some $N$-partite quantum states with a link composed of $N$ closed rings \cite{Ar97}. The $|\textrm{GHZ}\rangle$ state, which becomes separable after any measurement in the computational basis performed on its subsystem, was an initial inspiration. In such a way $|\textrm{GHZ}\rangle$ state can be associated with a particular configuration of three rings, called \textit{Borromean rings}, see \cref{34}. Indeed, Borromean link is characterized by the fact that if one ring is cut, the remaining two become disconnected. Usage of other tools borrowed from knot theory to analyze multipartite entanglement was further advocated in \cite{KL02,KM19}.

The aforementioned analogy is, however, not basis independent. On the one hand, measuring the $|\textrm{GHZ}\rangle$ state in the computational basis collapses always in a separable state, on the other hand, measurement in the rotated basis $\ket{+}, \ket{-}$ basis results into an entangled state. This motivated Sugita \cite{AS07} to revise the above analogy by proposing the partial trace of a subsystem as an alternative interpretation for cutting related ring. In that way, the physical process corresponding to cutting the ring becomes basis independent \cite{MSS13}. This analogy between quantum entanglement and linked rings, was later explored \cite{Quinta_2018}, and will be used here. Partial trace over given subsystems can be physically interpreted as not registering them by a measurement device or as the loss of related particles. Studying quantum entanglement resistant to a particle loss thus corresponds to the question, whether a given reduced density matrix represents an entangled state \cite{Neven,OB18,KNM02,KZM02}.

\section{Links and quantum states}
\label{MresistStates}

Recently, Neven et al. introduced the notion of $m$-resistant quantum $N$-qubit states which entanglement of the reduced state of $N-m$ subsystems is fragile with respect
to loss of any additional subsystem \cite{Neven}.

\begin{definition}
\label{Aarhus1}
An entangled state $\ket{\psi}$ of $N$ parties is
called $m$-resistant if:
\begin{itemize}
\item $\ket{\psi}$ remains entangled as any $m$ of its $N$ subsystems are traced away;
\item $\ket{\psi}$ becomes separable if a partial trace is performed over
  an arbitrary set of $m+1$ subsystems.
  \end{itemize}
\end{definition}

The first non-trivial examples of $m$-resistance arise for $N=3$ qubits, which is simple enough to study on intuitive grounds. The $m=0$ resistance is exemplified by the GHZ state
\eq{
\ket{\rm GHZ} = \frac{1}{\sqrt{2}}\left(\ket{000} + \ket{111}\right) ,
\label{GHZ3}
}
since tracing out any party of GHZ state results in a separable state. The corresponding link has the properties of the well-know Borromean link, depicted in \cref{34}. In particular, if any ring in Borromean link is cut, the remaining two rings become separated. Such a connection between Borromean link and the GHZ state was already noticed by Aravind \cite{Ar97}. For arbitrary number $N$ of rings, the natural generalization of Borromean link is called Brunnian link \cite{Brunn}. Brunnian link is characterized by the same property that cut of any of its rings disconnect the link, see \cref{N4k01connected}. On the other hand, the three-qubit state \cite{ThreeQub},
\eq{
\ket{\rm W} = \frac{1}{\sqrt{3}}\left( \ket{100} + \ket{010} + \ket{001} \right)
}
possesses property of being $m=1$ resistant. Indeed, a partial trace over any of its subsystems will result in an entangled mixed state. A corresponding link which mimics this property is represented on \cref{34}. It was first considered as
a representation of the W state in~\cite{Ar97}.

\section{In search for m-resistant states of N-qubit system}
\label{search}

In \cite{PhysRevA.100.062329}, we presented the general form of the symmetric mixed state of $N$ qubits with the property of being $m$ resistant. This form was obtained by relating to a given $m$-resistant link to the associated quantum state. Unlike the case of mixed states, similar identification between links and pure states is more intricate, if at all possible. It seems that the problem of finding pure $m$-resistant states for a given $m$ and the arbitrary number $N$ of qubits is not trivial. Obviously, the desired property is fully symmetric with respect to the exchange of subsystems. Therefore a search for such states among symmetric states seems to be a reasonable approach. Below, we present the partial solution to this problem, which is based on the Majorana representation~\cite{Aulbach_2010,Martin10}. 


The \textit{stellar representation} of Majorana \cite{M_kel__2010,Aulbach_2010,bengtsson_zyczkowski_2006} provides an alternative intuition on the geometry of symmetric states. Any permutation-invariant state of $N$ qubits might be wtitten in the following form
\eq{\label{MajRep}
\ket{\psi} = \frac{1}{\sqrt{K}} \sum_{\sigma \in \textrm{S}_N} \ket{\eta_1}_{i_1} \ldots \ket{\eta_N}_{i_N}
}
where $K$ is a suitable normalization constant, the sum runs over all permutations $\sigma \in \textrm{S}_N$ of particles indices $i_1, i_2,\ldots,i_N$, and
\eq{
\ket{\eta_j}_{X} = \cos\left(\frac{\theta_j}{2}\right)\ket{0}_X + e^{i \phi_j} \sin\left(\frac{\theta_j}{2}\right) \ket{1}_X\,.
}
The pairs of angles $(\theta_j,\phi_j)$ represent a point on the sphere, and are called Majorana stars. Thus, one may define a fully symmetric $N$-qubit state by fixing $N$ points on the sphere. For this reason, the name stellar representation is also common, as each point represents a star in the sky, while a group of stars forms a constellation.

Among many advantages, the Majorana representation ascribes some geometrical intuition of entanglement of a symmetric state. If one chooses all of the stars at a single point the corresponding state is separable. For example, $N$ degenerated stars on the North pole
represent the separable state $\ket{0\cdots0}$. Nevertheless, as the degeneracy is lifted, most of the non-trivial constellation
of stars corresponds to an entangled state of $N$ qubits.
Furthermore, the distance between the stars is related to the degree of entanglement \cite{ZS01,PhysRevA.85.032314},
although the precise criterium to quantify entanglement is not uniquely defined. 

We introduce the following family of states:
\eq{
\label{GenMPansatz}
\ket{\psi^{N}_{m}} = \frac{1}{\sqrt{1+\binom{N}{m}}}\left(\sqrt{\binom{N}{m}}\ket{0}^{\otimes N} - (-1)^{N+m} \ket{D^{N}_{m}}\right)
}
where $\ket{D^{N}_{m}}$ are so-called \textit{Dicke states} \cite{Dicke54}:
\begin{equation}
\ket{\text{D}_N^m} \propto 
\sum_{\sigma \in \mathcal{S}_N}  \sigma \Big( \ket{ \underbrace{1\;\cdots \;1}_{k \text{ times}} \; \underbrace{0 \;\cdots \; 0}_{N-m \text{ times}}} \Big) ,
\end{equation}
with the summation going over all permutations $\sigma$. This family of states was motivated by invatigation of $N=3,4$ qubit states. Indeed, for $N=3,4$ and any non-trivial value of $m$, the state $\ket{\psi^{N}_{m}}$ provides example of $m$-resistant state. In particular in the case of $N=3$ qubits presented constellations give us $m=0$, and $m=1$ resistant states respectively:
\eq{\label{N3m0State}
\ket{\psi^{3}_{0}} = \frac{1}{\sqrt{2}}\left(\ket{000}+\ket{111}\right),
}
\eq{
\ket{\psi^{3}_{1}}= \frac{1}{\sqrt{12}}\left(3\ket{000}+\ket{011}+\ket{101}+\ket{110}\right)\,.
}
It is straightforward to show that the constellations presented on \cref{StarsN3} give the desired answer for both non-trivial values of $m$. Similarly, states $\ket{\psi^{4}_{m}}$ are $m=0,1,2$ resistant states of four qubits for respective values of $m$. Related constellations are presented on \cref{StarsN4}. Note that $\ket{\psi^{N}_{m}}$ may be directly related to the distribution of $m$ stars on the North pole, with all others evenly distributed along the equator.

 \begin{figure}[ht!]
  \subfloat[0-resistant state \label{N3m0Star}]{\includegraphics[width=.25\columnwidth]{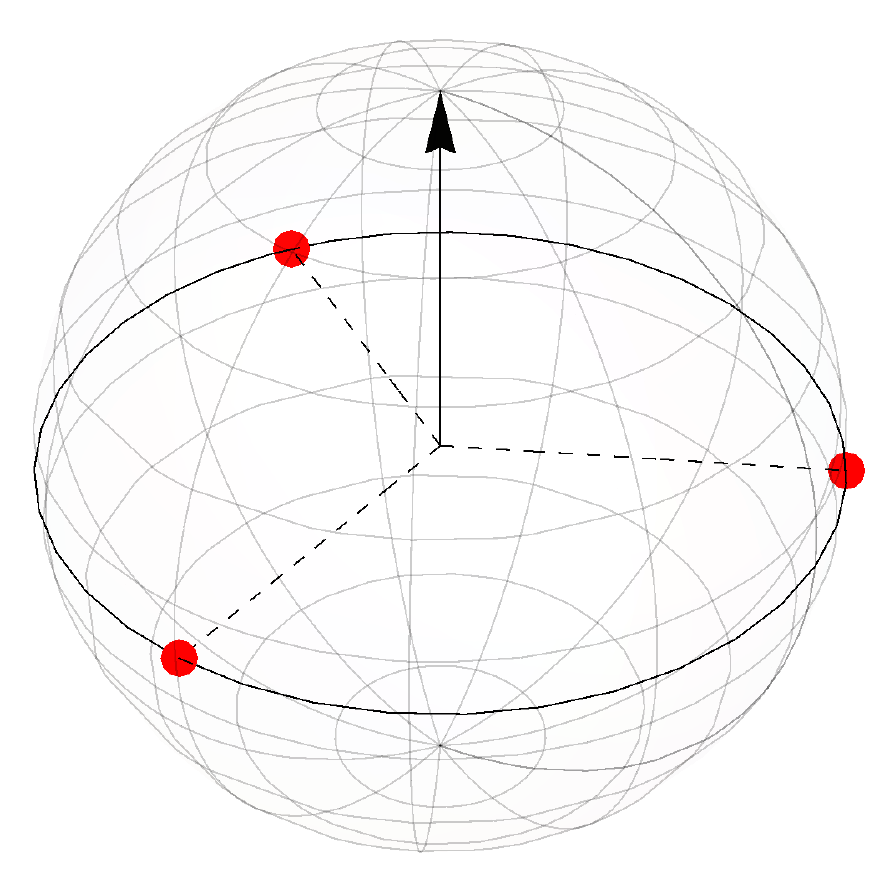}}
  \hspace{30mm}
  \subfloat[1-resistant state\label{N3m1Star}]{\includegraphics[width=.25\columnwidth]{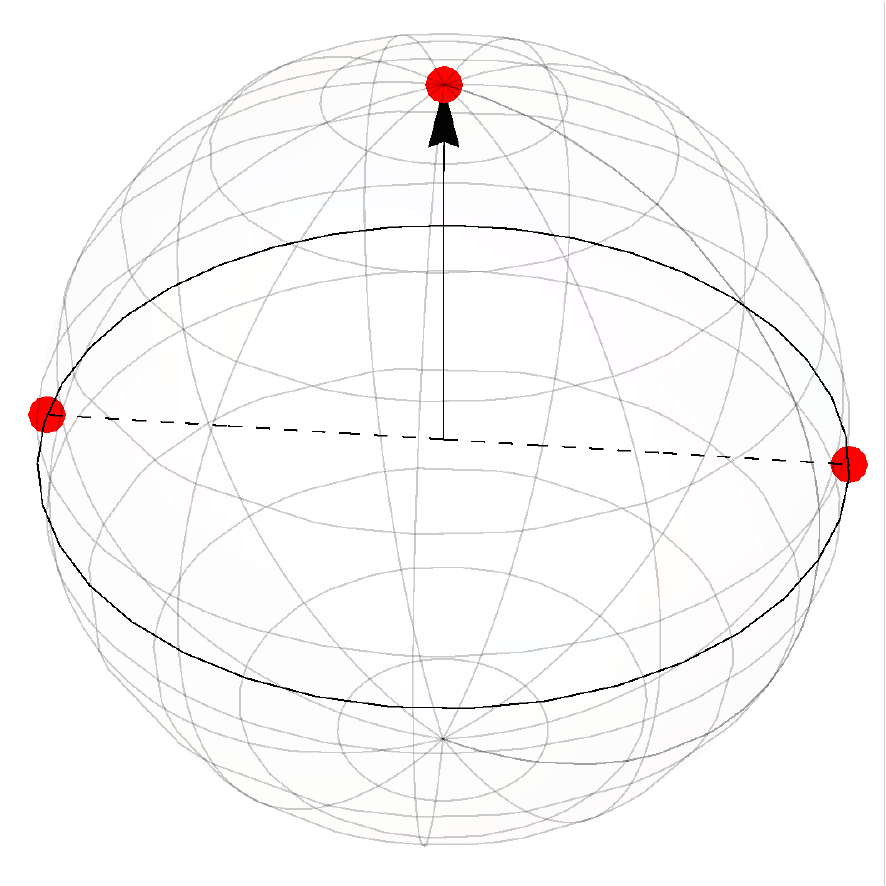}}
  \caption{Constellations defining $m$-resistant states of 3 qubits. The middle arrow serves as a reference pointing to the North pole. The constellation in (a) corresponds to the state $|\textrm{GHZ}\rangle$ and is illustrated by the Borromean link, see \cref{34}, while (b) presents an exemplary 1-resistant state of three qubits.}
  \label{StarsN3}
\end{figure}

\begin{figure}[ht!]
  \subfloat[$0$-resistant state\label{N4m0Star}]{\includegraphics[width=.25\columnwidth]{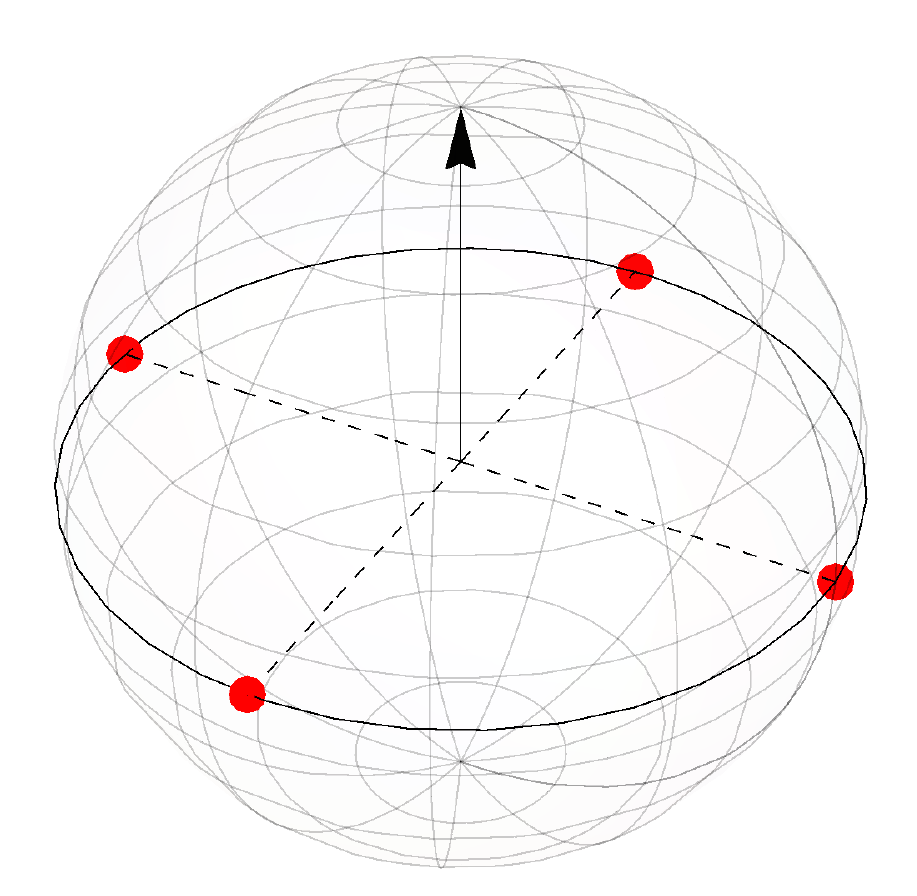}}
  \hfill
  \subfloat[$1$-resistant state\label{N4m1Star}]{\includegraphics[width=.25\columnwidth]{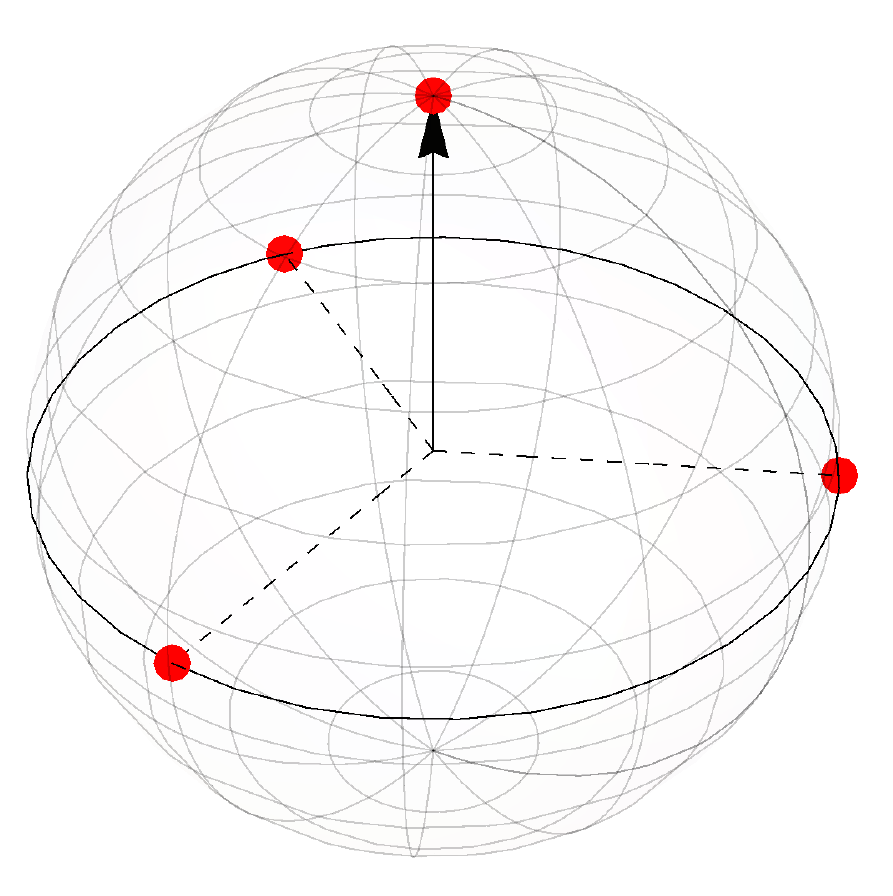}}
  \hfill
  \subfloat[$2$-resistant state\label{N4m2Star}]{\includegraphics[width=.25\columnwidth]{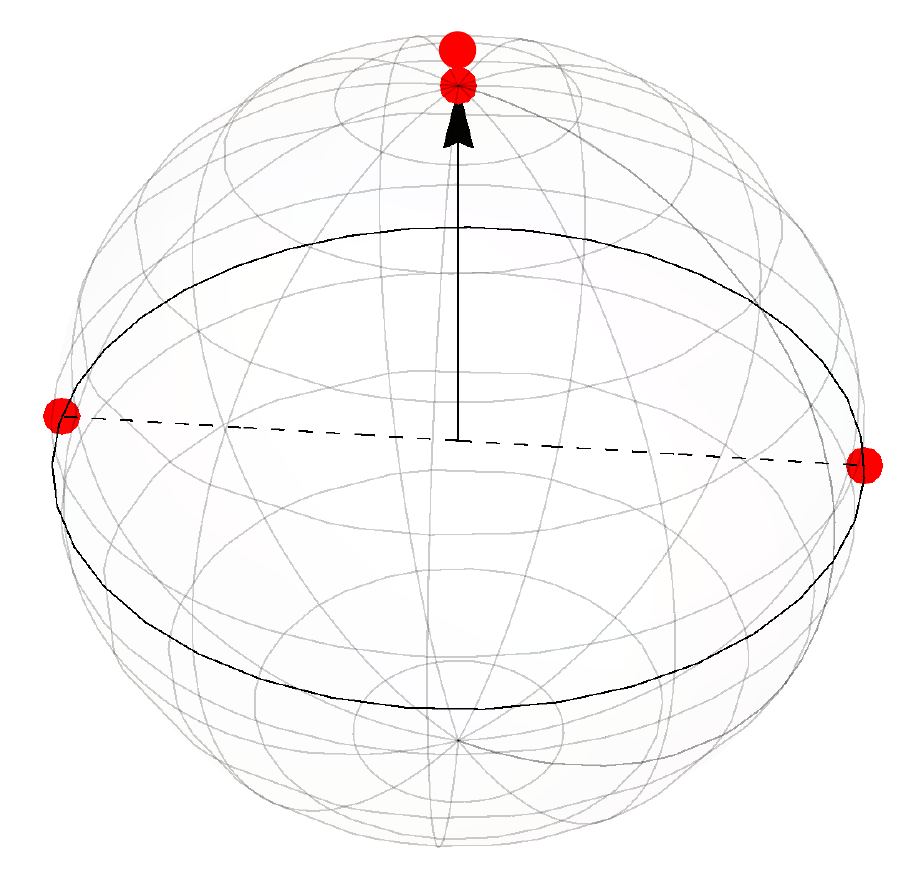}}
  \caption{Constellations defining $m$-resistant states of 4 qubits for $m=0,1,2$. The middle arrow serves as a reference pointing to the North pole. Presented constellations are related to three links presented on \cref{N4k01connected}.}
  \label{StarsN4}
\end{figure}

\begin{figure}[ht!]
\centering
\includegraphics[width=0.25\textwidth]{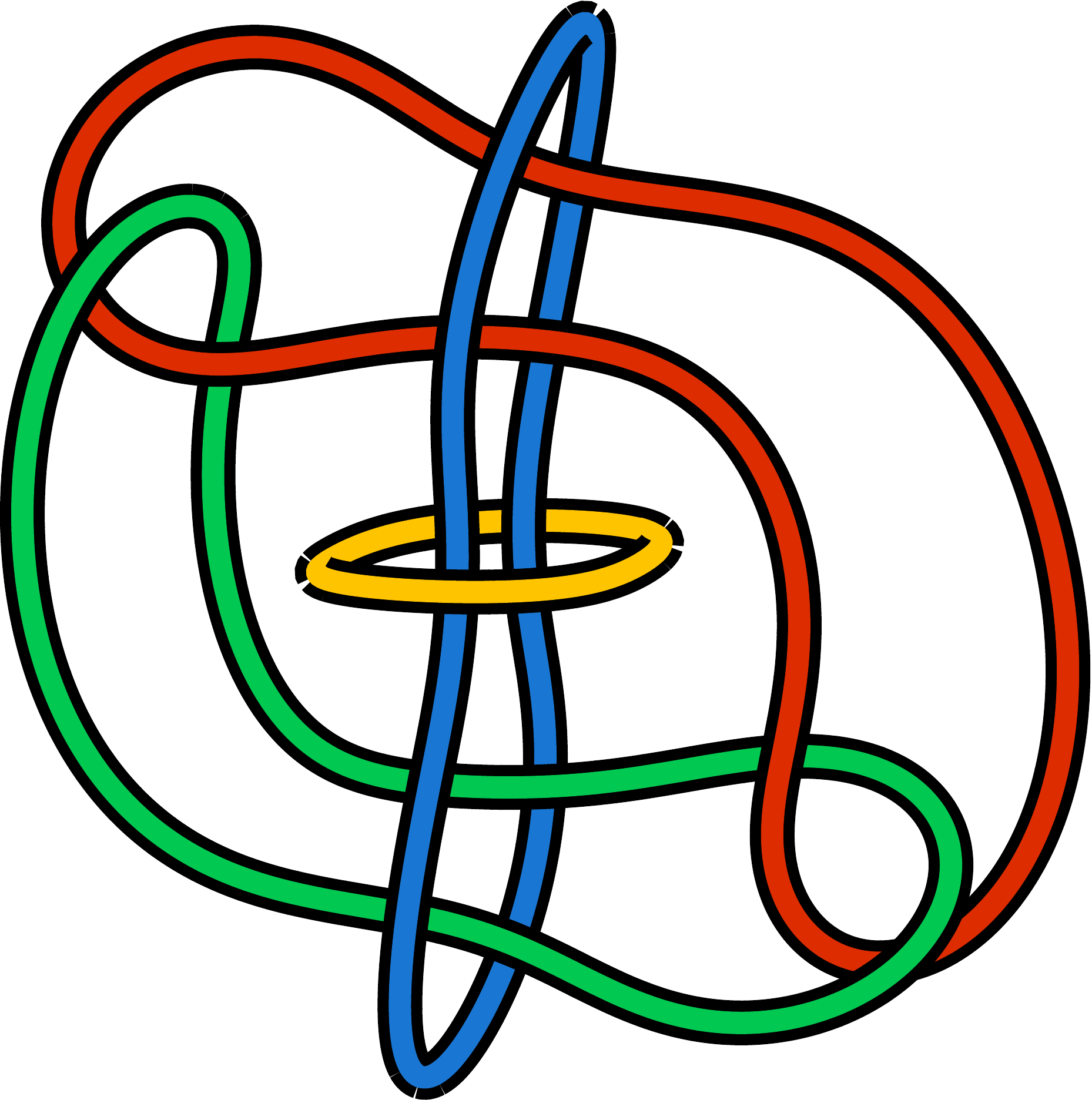}
\hspace{1cm}
\includegraphics[width=0.25\textwidth]{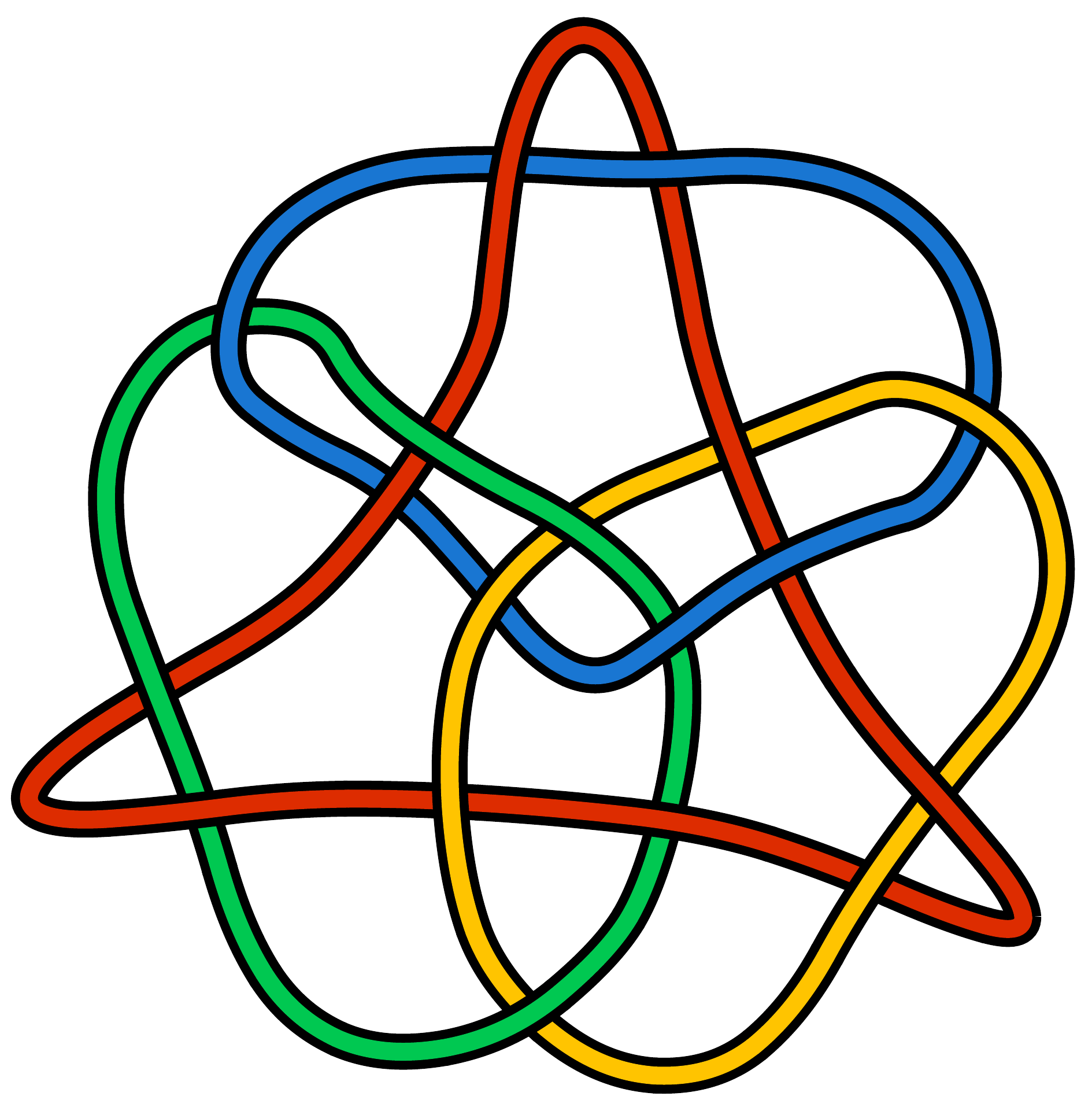}
\hspace{1cm}
\includegraphics[width=0.25\textwidth]{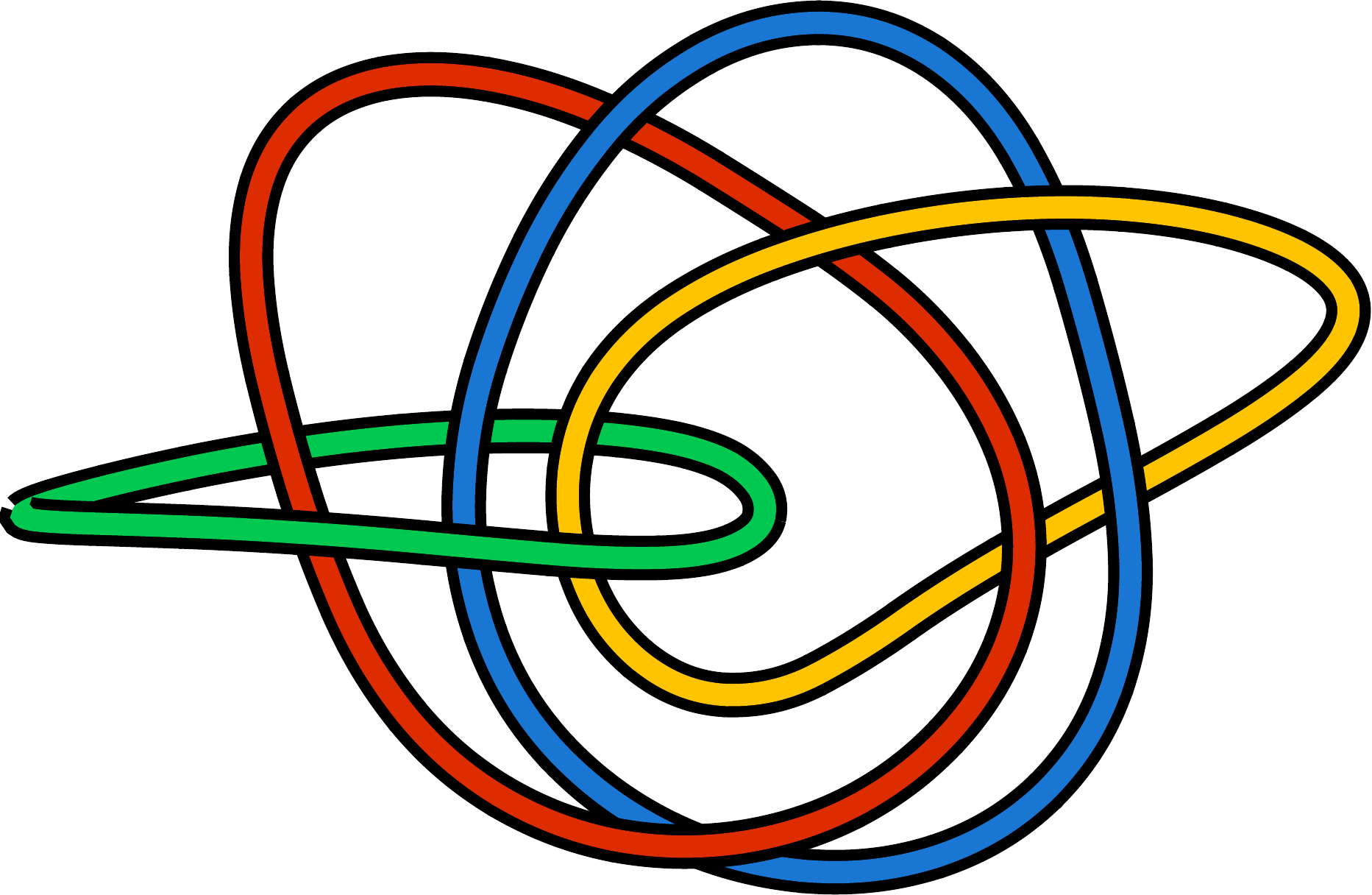}
\caption{On the left, an example of a 0-resistant link of 4 rings, which is a $4$-component Brunnian link. Notice that after cutting a single ring all three remaining rings become disconnected. 
In the center, an example of a 1-resistant link of 4 rings. After cutting any single ring one arrives at the Borromean link.
On the right, an example of a 2-resistant link of 4 rings. After cutting any two rings the remaining two are still connected.}
      \label{N4k01connected}
\end{figure}

The agreement between the above construction and $m$-resistant states breaks down for $N=5$ parties. Indeed, the state $\ket{\psi^{5}_{1}} = $ is in fact $2$-resistant. In all other cases, i.e. $m=0,2$ and $3$, related pure states are $0,2$, and $3$-resistant respectively. One may naturally inquire whether another symmetric state of five qubits might be in fact $1$-resistant. We have tested all combinations where two stars are lifted by a general latitude, as well as four stars, and it is possible to show that no $1$-resistant state exists for these families of states. Obviously, this does not prove that a pure $1$-resistant state of 5 qubits does not exist, although it is tantalizing to conjecture so. In general, we have verified the following.

\begin{proposition}
\label{Aarhus4}
For any number of particles $N\geq 3$, the family of states \cref{GenMPansatz}:
\eq{
\ket{\psi^{N}_{m}} = \frac{1}{\sqrt{1+\binom{N}{m}}}\left(\sqrt{\binom{N}{m}}\ket{0}^{\otimes N} - (-1)^{N+m} \ket{D^{N}_{m}}\right)
}
provides examples of $m=0,(N-3)$, and $(N-2)$-resistant states respectively.
\end{proposition}

\begin{proof}
Observe that $\ket{\psi^{N}_{0}}$ is zero resistant. Indeed, by (\ref{GenMPansatz}), we have
\eq{
\ket{\psi^{N}_{0}} = \frac{1}{2}\left(\ket{0}^{\otimes N} - (-1)^{N} \ket{1}^{\otimes N}\right)\,.
}
This state is a generalization of the GHZ state for $N$ qubits, which, after any partial trace, returns a density matrix of the form
\eq{
{\textrm{tr}_{k}} \left[\hat{\rho}^{N}_{0}\right] = \frac{1}{2}\left(\ket{0}\bra{0}^{\otimes (N-k)} + \ket{1}\bra{1}^{\otimes (N-k)}\right) ,
}
where ${\textrm{tr}_{k}}$ denotes the trace over any
set consisting of $k$ qubits.
This density matrix is separable for any $k>0$, and thus the state $\ket{\psi^{N}_{0}}$ is $0$-resistant for any $N$.

Regarding the state $\ket{\psi^{N}_{N-2}}$, after simple algebraic operations, one finds that any $N-2$ partial traces will result in the density matrix proportional to
\eq{
{\textrm{tr}_{N-2}} \left[ \hat{\rho}^{N}_{N-2} \right] \propto
\begin{pmatrix}
\alpha_{22} + (\alpha_{02})^2 & 0 & 0 & -\alpha_{02} \\
0 & \alpha_{21} & \alpha_{21} & 0 \\
0 & \alpha_{21} & \alpha_{21} & 0 \\
-\alpha_{02} & 0 & 0 & 1
\end{pmatrix} ,
}
where we define $\alpha_{ij} \equiv \binom{N-i}{j}$.
The partially transposed matrix has one of its eigenvalues equal to
\eq{\lambda_1 =
-N^2+3N-4
}
which is negative for any $N$. Thus, by the
 positive partial transpose (PPT) test,
 we confirm that each 2-qubit reduced density matrix is always entangled, and the state $\ket{\psi^{N}_{N-2}}$ is $(N-2)$-resistant for any $N$.

Finally, focusing on the state $\ket{\psi^{N}_{N-3}}$, we
perform the partial trace over any set of  $N-2$ qubits and obtain a two-qubit state
\eq{
{\textrm{tr}_{N-2}} \left[ \hat{\rho}^{N}_{N-3} \right] \propto
\begin{pmatrix}
\alpha_{23} + (\alpha_{03})^2 & 0 & 0 & 0 \\
0 & \alpha_{22} & \alpha_{22} & 0 \\
0 & \alpha_{22} & \alpha_{22} & 0 \\
0 & 0 & 0 & \alpha_{21}
\end{pmatrix} .
}
The partially transposed matrix of size four has all positive eigenvalues for $N>1$.
In this case the PPT test
 guarantees that the resulting state is always separable. We must then look into the reduced density matrix with one less partial trace, in order to check if it is entangled. Such a 3-qubit density matrix is proportional to
\ea{
& {\textrm{tr}_{N-3}} \left[ \hat{\rho}^{N}_{N-3} \right] \propto  
 \begin{pmatrix}
\alpha_{33} + (\alpha_{03})^2 & 0 & 0 & 0 & 0 & 0 & 0 & \alpha_{03}\\
0 & \alpha_{32} & \alpha_{32} & 0 & \alpha_{32} & 0 & 0 & 0 \\
0 & \alpha_{32} & \alpha_{32} & 0 & \alpha_{32} & 0 & 0 & 0 \\
0 & 0 & 0 & \alpha_{31} & 0 & \alpha_{31} & \alpha_{31} & 0 \\
0 & \alpha_{32} & \alpha_{32} & 0 & \alpha_{32} & 0 & 0 & 0 \\
0 & 0 & 0 & \alpha_{31} & 0 & \alpha_{31} & \alpha_{31} & 0 \\
0 & 0 & 0 & \alpha_{31} & 0 & \alpha_{31} & \alpha_{31} & 0 \\
\alpha_{03} & 0 & 0 & 0 & 0 & 0 & 0 & 1 \\
\end{pmatrix}\,.
}
By partially transposing any qubit, we will obtain a matrix which possesses an eigenvalue equal to
\eq{\lambda_1 =
\left(13-7N+N^2 - \sqrt{193-202N + 79N^2 -14N^3 + N^4}\right)
}
which is negative for all $N$. PPT criterion implies that all the subsystems are entangled. This proves that the state $\ket{\psi^{N}_{N-3}}$ is $(N-3)$-resistant.
\end{proof}

As a final remark to the section, we note that the constellations for all $0$-resitant states follow the same rule, which is a regular $N$-sided poligon. This result has been previously found in \cite{Neven}.

\section{In search for m-resistant states of N-qudit system}
\label{Aarhus2}


As we mentioned in \cref{search}, we did not find a general construction of an $m$-resistant pure qubit state of $N$ qubit system. Hence, we have expanded our search for non-symmetric states with subsystems of a larger local dimension $d\ge 3$. We present the general formula for a $m$-resistant $N$-qudit pure state for any $N \geq 2m$. Our construction is based on a particular family of combinatorial designs and their connection to quantum states. In particular, we use the notion of \textit{orthogonal arrays} (\textit{OA}), and the established link~\cite{DiK} between multipartite quantum states and OA. This connection will be expanded and further use in \cref{chapter3} for designing another family of quantum states.

Orthogonal arrays \cite{OA} are combinatorial arrangements, tables with entries satisfying given orthogonal properties. A close connection between OA and codes, entangled states, error-correcting codes, uniform states has been established \cite{OA}. Therefore,
investigation of the connections between OA and resistant states seems to be a
natural approach.
Firstly, we briefly present the concept of OA, secondly, we demonstrate relations between OA and resistant states.

An orthogonal array $\oa{r,N,d,k}$ is a table composed by $r$ rows, $N$ columns with entries taken from $0,\ldots,d-1$ in such a way that each subset of $k$ columns contains all possible combination of symbols with the same amount of repetitions. The number of such repetitions is called \textit{the index} of the OA and denoted by $\lambda$. One may observe, that the index of OA is related to the other parameters:
\eq{
\lambda =\dfrac{r}{d^k} .
}
\cref{OA111} presents an example of an OA. A pure quantum state consisting of $r$ terms might be associated with $\oa{r,N,d,k}$, simply by reading all rows of OA
\cite{DiK,OA}.
The state  of $N$ qudits
associated with the orthogonal array $\oa{r,N,d,k}$
will be denoted as
by $\ket{\phi_{(N,k )}}_d$.

\begin{figure}[ht!]
\centering
\[
 \begin{array}{*5{c}}
    \tikzmark{left}{0} &
0&0&0&0 \\
0&1&1&1&1\\
0&2&2&2&2\\
0&3&3&3&3\\
1&0&1&2&3\\
1&1&0&3&2\\
1&2&3&0&1\\
1&3&2&1&0\\
2&0&2&3&1\\
2&1&3&2&0\\
2&2&0&1&3\\
2&3&1&0&2\\
3&0&3&1&2\\
3&1&2&0&3\\
3&2&1&3&0\\
3&    \tikzmark{right}{3}&0&2&1
\Highlight[first]
  \end{array}
  \qquad
  \begin{array}{*5{c}}
0&
\tikzmarkk{up}{0}&0&0&0 \\
0&1&1&1&1\\
0&2&2&2&2\\
0&3&3&3&3\\
1&0&1&2&3\\
1&1&0&3&2\\
1&2&3&0&1\\
1&3&2&1&0\\
2&0&2&3&1\\
2&1&3&2&0\\
2&2&0&1&3\\
2&3&1&0&2\\
3&0&3&1&2\\
3&1&2&0&3\\
3&2&1&3&0\\
3&  3&\tikzmarkk{down}{0}&2&1
\Highlightt[first]
  \end{array}
      \qquad
 \begin{array}{*3{c}}
 \ket{(5,1)_4}&\propto&\ket{00000} \\
 &+&\ket{01111} \\
&+&\ket{02222}\\
&+&\ket{03333}\\
&+&\ket{10123} \\
&+&\ket{11032}\\
&+&\ket{12301}\\
&+&\ket{13210}\\
&+&\ket{20231} \\
&+&\ket{21320} \\
&+&\ket{22013}\\
&+&\ket{23102}\\
&+&\ket{30312} \\
&+&\ket{31203} \\
&+&\ket{32130}\\
&+&\ket{33021}
  \end{array}
\]
\Highlightt[second]
\captionsetup{justification=justified,singlelinecheck=false}
\caption{Orthogonal array of unity index $\oa{4^2,5,4,2}$ obtained from the Reed-Solomon code of length $5$ over Galois field GF($4$). Each subset consisting of two columns contains all possible combination of symbols. Here, two such subsets are highlighted. The relevant quantum state is obtained by forming a superposition
of states corresponding to consecutive rows of the array  -- see expression
on the right-hand side.}
\label{OA111}
\centering
\end{figure}

The crucial quantity for our purpose, related to OA, is its index. It preserves the following information: how many repetitions of any sequence $i_1,\ldots,i_k$ there are for each subsystem of $k$ rows. For $\lambda =1$, any sequence appears only once, and such an array is called \textit{index unity array}. We emphasize their remarkable role in the search for resistant quantum states.

\begin{proposition}
\label{prop111}
For any orthogonal array of index unity $\OA \left( d^k,N,d,k \right) $, where $N \geq 2k$, the relevant quantum state $\ket{(N,m)_d}$ is $k-1$-resistant. 
\end{proposition}


For the proof of above statement, we refer to \cite{PhysRevA.100.062329}. From the OA presented in \cref{OA111}, we obtain, for example, the following
five-ququart, $1$-resistant state:
\begin{align}
\label{erty}
\ket{(5,1)_4}\propto&\ket{00000} +\ket{01111} +\ket{02222}+\ket{03333}+ \nonumber \\
&\ket{10123} +\ket{11032} +\ket{12301}+\ket{13210}+ \nonumber \\
&\ket{20231} +\ket{21320} +\ket{22013}+\ket{23102}+ \nonumber \\
&\ket{30312} +\ket{31203} +\ket{32130}+\ket{33021}.
\end{align}

Bush provided the general method for constructing OAs of index unity \cite{Bush}. 

\begin{theorem}[Bush, 53']
If $d$ is a prime power, i.e $d=p^n$ for some prime number $p$ and natural number $n$, then we can construct the array $\OA \left( d^k, d+1,d,k\right) $.
\end{theorem}

Note that combining Bush's result with \cref{prop111} provides the existence of $m$-resistant $N$-qudit states for $N \geq 2(m+1)$.

\begin{proposition}
\label{Aarhus5}
For any $N \geq 2(m+1)$ there exists the $N$-qudit state which is $m$-resistant. The local dimension $d$ is the smallest prime power larger than $N-1$.
\end{proposition}

\cref{erty}presents an example of a $m$-resistant qudit state obtained by reading consequtive roves of OA. A more interested reader might easily reproduce other resistant states with the help of available OA tables \cite{OAlib}. We organize all obtained results in \cref{table}. They encourage us to pose the following conjecture.

\begin{conjecture}
\label{Aarhus6}
For any $N$ and $m$, there exists an $m$-resistant $N$-qudit state in some local dimension $d$.
\end{conjecture}

\noindent
Furthermore, we leave as a list of related open problems.
\begin{enumerate}
\item Investigate whether there exist pure states of $N$ qubits
with the $m$-resistance property, for any $m=0,1,\dots N-2$.
\item If such a state do not exists, for each $N$ find the minimal local dimension $d$ such that there exist $m$-resistant states of $N$ qudits.
\item For any class of $m$-resistant states of $N$ qubits find a state for which its average entanglement after partial trace over any set of $m$ parties is the largest,  if measured with respect to a given measure of entanglement.
\end{enumerate}

\begin{table}
\begin{tikzpicture}[x=\daywidth, y=-1cm, node distance=0 cm,outer sep = 0pt]
\tikzstyle{day}=[draw, rectangle,  minimum height=1cm, minimum width=\daywidth, fill=yellow!15,anchor=south west]
\tikzstyle{day2}=[draw, rectangle,  minimum height=1cm, minimum width=1.5 cm, fill=yellow!20,anchor=south east]
\tikzstyle{hour}=[draw, rectangle, minimum height=1 cm, minimum width=1.5 cm, fill=yellow!30,anchor=north east]
\tikzstyle{1hour}=[draw, rectangle, minimum height=1 cm, minimum width=\daywidth, fill=yellow!30,anchor=north west]
\tikzstyle{Planche0}=[1hour,fill=red!40]
\tikzstyle{Planche1}=[1hour,fill=blue!50]
\tikzstyle{Planche2}=[1hour,fill=blue!35]
\tikzstyle{Planche3}=[1hour,fill=blue!20]
\tikzstyle{Planche4}=[1hour,fill=green!30]
\tikzstyle{Planche5}=[1hour,fill=blue!20]
\tikzstyle{Planche6}=[1hour,fill=blue!10]
\tikzstyle{Planche7}=[1hour,fill=blue!5]
\tikzstyle{Planche8}=[1hour,fill=blue!2]
\tikzstyle{Ang2h}=[2hours,fill=green!20]
\tikzstyle{Phys2h}=[2hours,fill=red!20]
\tikzstyle{Math2h}=[2hours,fill=blue!20]
\tikzstyle{TIPE2h}=[2hours,fill=blue!10]
\tikzstyle{TP2h}=[2hours, pattern=north east lines, pattern color=magenta]
\tikzstyle{G3h}=[3hours, pattern=north west lines, pattern color=magenta!60!white]
\tikzstyle{Planche}=[1hour,fill=white]
\node[day] (lundi) at (1,8) {4};
\node[day] (mardi) [right = of lundi] {5};
\node[day] (mercredi) [right = of mardi] {6};
\node[day] (jeudi) [right = of mercredi] {7};
\node[day] (vendredi) [right = of jeudi] {8};
\node[hour] (8-9) at (1,8) {0};
\node[hour] (9-10) [below = of 8-9] {1};
\node[hour] (10-11) [below= of 9-10] {2};
\node[hour] (11-12) [below = of 10-11] {3};
\node[hour] (12-13) [below  = of 11-12] {4};
\node[hour] (13-14) [below = of 12-13] {5};
\node[hour] (14-15) [below = of 13-14] {6};
\node[Planche1] at (1,8)  {$\ket{\psi^{4}_{0}}$}; 
\node[Planche2] at (1,9) {$\ket{\psi^{4}_{1}} $};
\node[Planche3] at (1,10) {$\ket{\psi^{4}_{2}} $};

\node[Planche1] at (2,8) {$\ket{\psi^{5}_{0}} $};
\node[Planche4] at (2,9) {$\ket{\phi^{}_{1,5}}_4 $};
\node[Planche2] at (2,10) {$\ket{\psi^{5}_{2}} $};
\node[Planche3] at (2,11) {$\ket{\psi^{5}_{3}} $};

\node[Planche1] at (3,8) {$\ket{\psi^{6}_{0}} $};
\node[Planche4] at (3,9) {$\ket{\phi^{}_{1,6}}_5 $};
\node[Planche0] at (3,10) {AME};
\node[Planche2] at (3,11) {$\ket{\psi^{6}_{3}} $};
\node[Planche3] at (3,12) {$\ket{\psi^{6}_{4}} $};

\node[Planche1] at (4,8) {$\ket{\psi^{7}_{0}} $};
\node[Planche4] at (4,9) {$\ket{\phi^{}_{1,7}}_7 $};
\node[Planche4] at (4,10) {$\ket{\phi^{}_{2,7}}_7 $};
\node[Planche2] at (4,12) {$\ket{\psi^{7}_{4}} $};
\node[Planche3] at (4,13) {$\ket{\psi^{7}_{5}} $};

\node[Planche1] at (5,8) {$\ket{\psi^{8}_{0}} $};
\node[Planche4] at (5,9) {$\ket{\phi^{}_{1,8}}_7 $};
\node[Planche4] at (5,10) {$\ket{\phi^{}_{2,8}}_7 $};
\node[Planche4] at (5,11) {$\ket{\phi^{}_{3,8}}_7 $};
\node[Planche2] at (5,13) {$\ket{\psi^{8}_{5}} $};
\node[Planche3] at (5,14) {$\ket{\psi^{8}_{6}} $};


\node[day2] at (1,8) {};
\draw (1,8)--(-0.15, 7);
\node at (0.65,7.4) {$N$};
\node at (0.1,7.6) {$m$};
\end{tikzpicture}
\captionsetup{justification=justified,singlelinecheck=false}
\caption{\label{table} Collection of $m$-resistant states of $N$ parties obtained so far. Three families: $\ket{\psi^{N}_{0}} $, $\ket{\psi^{N}_{N-3}} $ and $\ket{\psi^{N}_{N-2}} $ of $m$-resistant $N$-qubit states discussed in \cref{search} are presented on the differently shaded blue background. 
Furthermore, $m$-resistant qudit states $\ket{\phi^{}_{m,N}}_d $ constructed by virtue of orthogonal arrays are demonstrated on the green background. The number in subscript is relevant to the local dimension $d$ of a state. Observe that the local dimension $d$ of qudit states is usually (but not always) equal to $N-1$. 
Finally, a six-qubit state AME(6,2) provides an example of a $2$-resistant state.
}
\end{table}


\section{Asymptotic case}
\label{Aarhus3}

Consider a system consisting of $N$ subsystems with $d$ levels each and assume that $N$ is large. It is well known that a generic pure state of such a system is typically strongly entangled \cite{ZS02,HLW06}. Therefore the partial trace of the corresponding $N$-party  projector over any choice of subsystems traced away becomes mixed.  Typically, the more parties $k$ are traced out, the more mixed is the state describing the remaining $N-k$  subsystems.

If the number $k$ of subsystems removed is equal to $N/2$ the reduced state typically has a full rank so it can belong to the separable ball around the maximally mixed state \cite{ZHSL98, GB02}. It was indeed observed \cite{KNM02,KZM02,HLW06} that for $k\approx N/2$ a transition from entangled reduced states to the reductions with positive partial transpose takes place. Furthermore, if $k \approx 3N/5$  the remaining subsystem consisting of $2N/5$ particles is typically separable \cite{ASY12,ASY14} and this transition becomes sharp if $N \to \infty$. Thus one can expect that a generic pure state of a $N$--system state is $m$-resistant with $m$ typically of the order of $3N/5$ independently of the local dimension $d$.

\section{Conclusions}

In this chapter, I presented an analogy between quantum states and topological links, in which the operation of losing a subsystem is related to neglecting an associated ring. Furthermore, I introduced the notion of $m$-resistant states, which entanglement is resistant for loss of any $m$ particles but fragile for loss of any larger subsystem (\cref{Aarhus1}). I presented two methods of constructing $m$-resistant states: \cref{search} is based on the Majorana representation of symmetric states, while \cref{Aarhus2} on the combinatorial notion of Orthogonal Arrays. Related results are summarized in \cref{Aarhus4} and \cref{prop111,Aarhus5} respectively. \cref{table} presents both aforementioned families of states at glance. Finally, \cref{Aarhus3} discusses the typical resistance property of large systems.

\chapter{Highly symmetric states and groups}
\label{chap2}
In this chapter, I present an approach for constructing highly-entangled quantum states which is based on the group theory. Introduced states resemble the Dicke states, whereas the interactions occur only between specific subsystems related by the action of the chosen group. The states constructed by this technique exhibit desired symmetry properties and form a natural resource for less-symmetric quantum informational tasks. Furthermore, I~introduce a larger class of genuinely entangled states based on an arbitrary network structure, graphically represented as a (hyper)graph, with excitations appearing only in particular subsystems represented by (hyper)edges. I investigate the entanglement properties of both families of states and show an interesting phenomenon: most of the entanglement is concentrated between nodes of distance two and is absent between immediate neighbours. In addition I present two different methods of constructing introduced states: firstly, I propose quantum circuits whose complexity is comparable with the complexity of quantum circuits proposed for Dicke states, secondly, I present such states as a ground states of Hamiltonians with $3$-body interactions. I~demonstrate the viability of the provided constructions by simulating one of the considered states on available quantum computers: IBM - Santiago and Athens. An extension of some parts, in which the author's part was not substantial, as well as proofs of presented results, can be found in the joint paper \cite{burchardt2021entanglement}.

\section{Motivation}
	
Among different types of entangled states, permutation invariant states attracted a lot of attention in both continuous \cite{PhysRevA.71.032349,PhysRevLett.99.150501} and discrete \cite{PhysRevLett.98.060501,Bergmann_2013} variable systems. A remarkable example of such states is due to  Dicke \cite{Dicke54}, {\sl Dicke state} with $k$ excitations in a system of $N$-qubits is defined as \cite{PhysRevA.67.022112},
\begin{equation}
\label{Dicke}
\ket{\text{D}_N^k} \propto 
\sum_{\sigma \in \mathcal{S}_N}  \sigma \Big( \ket{ \underbrace{1\;\cdots \;1}_{k \text{ times}} \; \underbrace{0 \;\cdots \; 0}_{N-k \text{ times}}} \Big) ,
\end{equation}
where the summation runs over all permutations in the symmetric group $\mathcal{S}_N$. On one hand, symmetry of the Dicke states simplifies their theoretical \cite{Liu_2019} and experimental \cite{L_cke_2014,Wieczorek_2009} detection, furthermore facilitates the tasks of quantum tomography \cite{Mazza_2015}. On the other hand, the entanglement of Dicke states turned out to be maximally persistent and robust for the particle loss \cite{Dicke54,PhysRevLett.86.910}, such states provide inherent resources in numerous quantum information contexts, including quantum secret sharing protocols \cite{PhysRevLett.98.020503}, open destination teleportation \cite{PhysRevA.59.156}, quantum metrology \cite{RevModPhys.90.035005}, and decoherence free quantum communication \cite{PhysRevLett.92.107901}. 
	
So far most of the scientific attention was focused on fully symmetric tasks, for example, parallel teleportation \cite{helwig2013absolutely}, or symmetric quantum secret sharing protocols \cite{HelwigAME}. In various realistic situations, however, such a full symmetry between collaborating systems is not possible, required or even desirable. As an example, it was shown that there is no four-qubit state which is maximally entangled with respect to all possible symmetric partitions \cite{HIGUCHI2000213}. Such a state would allow for the parallel teleportation protocol of two qubits between any two two-qubit subsystems. Nevertheless, the following state:
\[
\dfrac{1}{2} \Big( \ket{1100}+\ket{0110}+\ket{0011}+\ket{1001} \Big)
\]
allows for the teleportation of a single qubit  to an arbitrary subsystem, and additionally for the parallel teleportation across the partition $13|24$. As we already discussed in Introduction, the following state
\begin{equation}
\ket{M_4} =   \dfrac{1}{\sqrt{6}} \Big( \ket{0011} + \ket{1100} +  e^{\frac{2 \pi i}{  3}} (\ket{1010}+ \ket{0101}) + e^{\frac{4 \pi i}{ 3}} (\ket{1001} + \ket{0110}) \Big),
\end{equation}
turned out to maximize the average entropy of entanglement over all bipartitions of four-qubit state. As it was later observed, state $\ket{M_4} $ is not fully permutational invariant, but invariant under any even permutation of qubits \cite{Lyons_2011,CenciCurtLyons}.

On one hand, it is especially reasonable to share resources in a not fully symmetric way in variants of quantum secrets sharing schemes, allowing only some parties for cooperation. On the other hand, various molecules in Nature (like benzene) stand out with remarkable symmetries, the general investigations of entanglement in highly symmetric systems may shed some light on the nature of correlation in relevant chemical molecules. In recent years, correlations and the entanglement contained in chemical bounds was investigated \cite{Szalay_2017,Szalay_2015,ding2020concept}, and a special attention was dedicated to highly symmetrical molecules \cite{Szalay_2015}. Although for most molecules the total correlation between orbitals seems to be classical, the general significance of entanglement in chemical bonds seems to be present \cite{ding2020concept}. 

\section{Group of symmetry of a quantum state}
\label{sec3a}
	A pure state $\ket{\psi}$ is called \textit{symmetric} if it is invariant under any permutation of its subsystems, i.e. $\sigma \ket{\psi} = \ket{\psi}$ for any element $\sigma$ from the permutation group $ \mathcal{S}_N$, where $N$ denotes number of subsystems. 	This might be generalized to a very natural definition of the group of symmetry of a quantum state.
	
\begin{definition}
\label{Aarhus9}
		A state $\ket{\psi}$ of $N$ subsystems is called $H$-symmetric, where $H$ is a subgroup of the permutation group, $H< \mathcal{S}_N$, iff it is permutation invariant for any $\sigma \in H$, and only for such permutations.
	\end{definition}
	
We begin with several examples of states with restricted group of symmetry. Firstly, consider a system consisting of three qubits. The celebrated states \cite{ThreeQub}:
\begin{align*}
\ket{\text{GHZ}} &\propto \ket{000} +\ket{111} , \\
\ket{\text{W}} &\propto \ket{001} +\ket{010} +  \ket{100} .
\end{align*}
are fully symmetric, see \cref{fig1ch2}. We also may construct three--qubit quantum states with other types of symmetry. For example
\[
|\chi_1\rangle \propto \ket{001}+\ket{010}
\]
is $\mathcal{S}_2$-symmetric state. This state is, however, bi-separable with respect to the partition $A|BC$. Nevertheless, a similar example $|\chi\rangle$ can be found among genuinely entangled states:
\[
|\chi_2\rangle \propto \ket{001}+\ket{010}+2\ket{100}+ 2\ket{111}.
\]
Furthermore the following state
\[
|\chi_3\rangle \propto \ket{001}+ e^{\frac{2 \pi i}{3}} \ket{010}+ e^{\frac{4 \pi i}{3}} \ket{100}
\]
exhibits the alternating $\mathcal{A}_3$-symmetry, symmetry of all even permutations of three qubits. Note that the above state is entangled and locally equivalent to the celebrated $\ket{W}$ state. In that way, we constructed states with all possible types of discrete symmetries in a three-qubit setting.

Observe that there exists an easy recipe for the construction of an $H$-symmetric state, where $H<S_N$ on $N$ qudits. 

\begin{proposition}
\label{Aarhus8}
Consider a subgroup $H < \mathcal{S}_N$. The following state $\ket{\psi} \in \mathcal{H}^{\otimes N}_N$ of the local dimension $N$:
\begin{equation*}
\ket{\psi} \propto \sum_{\sigma \in H} \ket{ \sigma_0 (0) \cdots \sigma_{N-1} (N-1)} 
\end{equation*}
is $H$ symmetric.
\end{proposition}
	
\begin{proof}
Clearly, that the group $H$ stabilizes $\ket{\psi}$. Suppose there is a larger group $H'$, which contains $H$, i.e. $H<H'$, and at the same time stabilize $\ket{\psi}$. Take any element $\sigma \in H'\setminus H$. Observe, that there is no term $\ket{ \sigma_0 (0) \cdots \sigma_{N-1} (N-1)}$ in $\ket{\psi} $, hence $H' $ cannot stabilize $\ket{\psi} $.
\end{proof}
	
\noindent
As an example, the three qutrits state
\begin{equation}
\ket{\psi} \propto \ket{012} +\ket{201} +\ket{120}
\end{equation}
is $\mathcal{A}_3$-symmetric. Indeed, one can see, that the cyclic permutation of the last three qutrits does not change the  state. This encourages us to pose the following question.
	
\begin{question}
\label{Aarhus7}
Consider any group of symmetry $H< \mathcal{S}_N$. What is the minimal local dimension $d$ for which there exists a $H$-symmetric state $\ket{\psi} \in \mathcal{H}^{\otimes N}_d$?
\end{question}
	
Consider now a natural basis for symmetric states. Any symmetric state $\ket{\psi}$ is a superposition of Dicke states \cite{Dicke54,PhysRevA.67.022112}. A similar decomposition for a $H$-symmetric state is also possible. We proceed with the following definition. 
	
\begin{definition}
\label{Dickelike}
For a given subgroup $H <\mathcal{S}_N$, we define the $N$-qubit \textit{Dicke-like} state with $k$ excitations is the following way:
\begin{equation}
\label{Dicke-like}
\ket{\text{D}_N^k}_H \propto  
\sum_{\sigma \in H}  \sigma \Big( \ket{ \underbrace{1\;\cdots \;1}_{k \text{ times}} \; \underbrace{0 \;\cdots \; 0}_{N-k \text{ times}}} \Big)
\end{equation}
where the summation runs over all permutations from the group $H<\mathcal{S}_N$. 
\end{definition}
	
As we shall see, the normalization constant in Eq.~(\ref{Dicke-like}) highly depends on the structure of the subgroup $H$ and the number of excitations, and is difficult to present in a consistent way. Note that the group symmetry of the Dicke-like state $\ket{\text{D}_N^k}_H$ is not necessarily given by $H$. The general analysis of the group of symmetries is tightly connected with the \textit{partially ordered set} (poset) of all subgroups of $\mathcal{S}_N$, which has rather complicated structure \cite{Poset}.

Any symmetric state $\ket{\psi}$ can be written in the computational basis as:
\begin{equation}
\label{eq1}
\ket{\psi} \propto \sum_{\sigma \in \mathcal{S}_N}
\ket{\phi_{\sigma (1)}}\cdots \ket{\phi_{\sigma (N)}},
\end{equation}
where the sum is taken over all permutations $\sigma \in \mathcal{S}_N$ in the symmetric group. In such a way, symmetric states have an effective representation, called \textit{Stellar representation} \cite{M_kel__2010}, as $N$ points (stars) on the Bloch sphere, corresponding to vectors  $\ket{\phi_i},\ldots , \ket{\phi_N}$, as we already discussed in \cref{chap1}. Stellar representation turned out to play a role by classification of entanglement in symmetric quantum states \cite{M_kel__2010,CrossRatioNqubits,PhysRevA.85.032314}. Furthermore, special symmetry conditions imposed on the stars, are related with highly entanglement properties of resulted states \cite{Martin10,Martin15,Martin17}. 

I introduce a natural generalization of this approach -- the \textit{generalized stellar representation}, which is suitable for quantum states exhibiting modes symmetries, i.e. for which summation in \cref{eq1}  runs over a subgroup $H$ of the symmetric group $\mathcal{S}_N$. This might possibly restrict the group of symmetries in the resulting state. As for symmetric states, consider $N$ points on the Bloch sphere: $\ket{\phi_1},\ldots , \ket{\phi_N},$ and the following product
\begin{equation}
\label{eq2}
\ket{\psi} \propto
\sum_{\sigma \in H} 
\ket{\phi_{\sigma (1)}}\cdots \ket{\phi_{\sigma (N)}},
\end{equation}
\noindent
where the sum runs over all permutations $\sigma$ from the group $H$. We might represent the state $\ket{\psi}$ as a collection of $N$ points on the Bloch sphere relevant to vectors $\ket{\phi_i},\ldots, \ket{\phi_N}$ with indicated action of the group $H$, see \cref{fig1ch2}. 
	
Notice that a given constellation of `stars' at the Bloch sphere  with the selected symmetry group $H$ do not represent uniquely the quantum state. The important information is carried in how the group $H$ is contained in $\mathcal{S}_N$, which mathematically might be expressed by immersion $H\hookrightarrow \mathcal{S}_N$ of the group $H$ into the symmetric group $\mathcal{S}_N$.

\begin{figure}[ht!]
\begin{center}
\begin{minipage}{0.25\linewidth}
\subfloat{\includegraphics[scale=1]{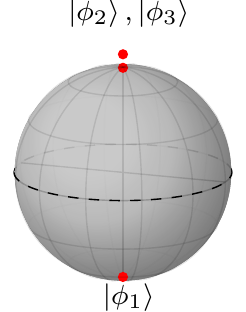}}
\end{minipage}
\hspace{1cm}
\begin{minipage}{0.6\linewidth}
\subfloat{\includegraphics[scale=0.9]{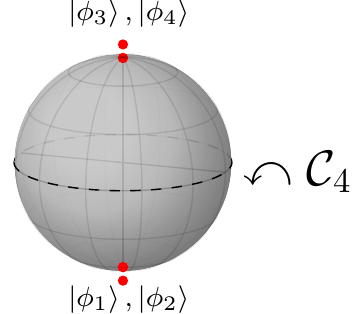}}
\subfloat{\includegraphics[scale=0.9]{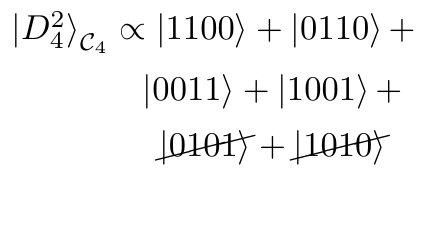}}
\end{minipage}
\end{center}
\caption{On the left, stellar representation of three qubits $\ket{\text{W}}$ state, by two stars at the North pole and a single star at the South pole. Stellar representation of a Dicke state $\ket{D_N^k}$ consists of $k$ stars at the South Pole and $N-k$ stars at the North Pole in general \cite{PhysRevA.85.032314}. Similar holds true for the Dicke-like states. On the right, $\ket{D_4^2}_{\mathcal{C}_4}$ defined in \cref{eeerrr} is represented by four stars on the Bloch sphere, with indicated action of the Cyclic group ${\mathcal{C}_4}$.
}
\label{fig1ch2}
\end{figure}
	
\section{Symmetric states related to (hyper)graphs}
\label{hyper}

In this section, I propose a scheme to associate to a given graph with $N$ vertices a single pure quantum state of an $N$-party system. Such a representation reflects not only the symmetry, but also the structure of a quantum circuit under which presented families of quantum states can be constructed. Let us emphasize that states constructed in that way are completely different form, the so-called \textit{graph states}, known also as \textit{cluster states} \cite{PhysRevLett.86.910,Graph1,Graph2,helwig2013absolutely}. 
	
A \emph{graph} $G$ is a pair $(V,E)$ where $V$ is a finite set, and $E$ is a collection of two-element subsets of $V$. We refer to elements of $V$ as \textit{vertices}, and elements of $E$ as \textit{edges} respectively. A \textit{hypergraph} is a natural generalisation of a graph, in which edges are arbitrary (not necessarily 2-element) subsets of $V$. A hypergraph is called \textit{uniform} if its all edges consist of the same number of elements equal to $k$, we refer to such an object as $k$-hypergraph. In particular, a 2-hypergraph is simply a graph.
	We shall denote the number of vertices by $N$. Moreover, we assume that vertices of the (hyper)graph are labeled by numbers $1,\ldots,N$.
	
\begin{definition}
\label{Aarhus10}
With a given (hyper)graph $G=(V,E)$, we associate a quantum state of $N$ qubits in the following way: 
\begin{equation}
\ket{G} := \dfrac{1}{\sqrt{|E|}}  \sum_{e \in E} \ket{\psi_e}
\end{equation} 
where $\ket{\psi_e}$ is a tensor product of $\ket{1}$ on positions labelled by indices which form the (hyper)edge $e$ and $\ket{0}$ on other positions. We shall refer to such states as \textit{excitation--states}. 
\cref{ex1p} illustrates this definition. 
\end{definition}

\begin{figure}[ht!]
\begin{center}
\includegraphics[scale=1]{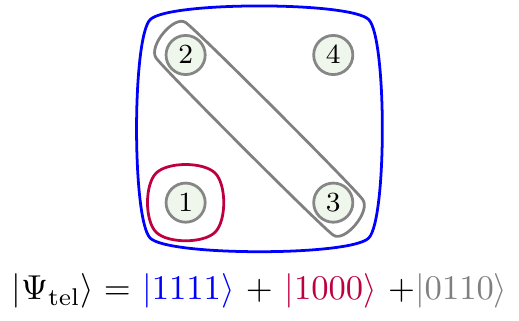}
\caption{So-called telescope state $\ket{\Psi_\text{tel}}$ \cite{ProblemTangle} represented as an excitation-state. Related  hypergraph is not uniform. Telescope state is simply obtained by reading over all edges. Note that as the empty set cannot form the hyperedge in the hypergraph, the ket $\ket{0000}$ cannot appear in the superposition which forms the excitation state.}
\label{ex1p}
\end{center}
\end{figure}

There is the following relation between the \textit{automorphisms group} of a complete $k$-hypergraph and the  group of symmetry of related state.

\begin{observation}
Consider an excitation-state $\ket{G}$ related to the hypergraph $G$. The group of symmetries of a state $\ket{G}$ is an automorphisms group of the related hypergraph $G$.
\end{observation}

\section{Entanglement properties}
	\label{cccc}

In this section, I present results concerning entanglement properties of general excitation states and Dike-like states. In particular, \cref{August1,cor5} provides separability criterion, while \cref{condition,A,B} provides formulas for exact value of a concurrence for regular graphs. As we shall see, entanglement properties of excitation states reflect the structure of the (hyper)graph. 

We begin by recalling the notion of \textit{Concurrence}, denoted by $C$ \cite{PhysRevLett.78.5022}. Concurrence is an entanglement measure, for any two-qubit mixed state $\rho$ its concurrence reads,
\begin{equation}
\label{concurrence}
C (\rho ):= \text{max} \{\lambda_1 -\lambda_2 -\lambda_3 -\lambda_4 ,0\},
\end{equation}
where $ \lambda _{1},\ldots,\lambda _{4}$ denote square 
roots of the eigenvalues of a Hermitian matrix ${\sqrt {\rho }}{\tilde {\rho }}{\sqrt {\rho }}$
ordered  decreasingly,
where
\[
{\displaystyle {\tilde {\rho }}=(\sigma _{y}\otimes \sigma _{y})\rho ^{*}(\sigma _{y}\otimes \sigma _{y})}
\]
is, the so-called, spin-flipped form of $\rho $. Note that the above definicion coincides with \cref{conLag} for all pure two-qubit states.

The \textit{generalized concurrence} $C_{v |\text{rest}}$ measures the entanglement between the subsystem $v$ and the rest of the system. For pure states, this quantity is determined by the relevant reduced density matrix \cite{PhysRevLett.78.5022}, $C_{v |\text{rest}} =2 \; \sqrt{\text{det} \rho_v} $, hence its values belong to the range $[0,1]$.
	
The distribution of bipartite entanglement, measured by the concurrence satisfies so-called monogamy inequality \cite{DistributedEntanglement,Osborne_2006}:
\begin{equation}
\label{eq5}
C_{v_1  |v_2 \ldots v_N }^2 \geq C_{v_1 v_2 }^2 + \ldots C_{v_1   v_N}^2 ,
\end{equation}
where $v_2 \ldots v_N$ are vertices relevant to subsystems.

In further considerations, we restrict our analysis to the $k$-uniform hypergraphs, and assume our graphs to be connected. Disconnected graphs are relevant for the tensor product of two excitation-states, and hence might be analyzed separately. We use the following notation. By \textit{distance} $d$ between vertices $v_0$ and $w$, we understand the minimal number of vertices $v_1 , \ldots v_d$, such that there exist (hyper)edges $e_1, \ldots, e_d$:
\[
v_{i-1}, v_i \in e_i
\]
and $v_d =w$. The \textit{degree} $d_v$ of the vertex $v$ is the number of edges on which $v$ is incident. The \textit{joint neighborhood} $n_{vw}$ of two vertices is the number of sets $W$ such that both: $W \cup v$ and $W \cup w$ constitute an edge. Finally, the \textit{section} $s_{vw}$ is the number of edges on which $v$ and $w$ are incident. Notice that for graphs, the joint neighborhood is simply the number of vertices adjacent to both: $v$ and $w$, while section $s_{vw}=1$ or $s_{vw}=0$ depending if vertices $v$ and $w$ are connected. Furthermore, to simplify the notation, for a given excitation state $\ket{G}$ and  two selected subsystems  $v$ and $w$, the concurrence of the reduced state $\rho_{vw}^G$, will be denoted as 
\begin{equation}
C_{vw} := C (\rho_{vw}^G ) .
\end{equation}

I obtained the following separability criterion.

\begin{proposition}
\label{condition}
For a $k$-uniform hypergraph, the two-party concurrence $C_{vw}$ reads
\begin{equation}
C_{vw} =\text{max} \Big\{ 0, \dfrac{2}{|E|} \big( n_{vw} - \sqrt{s_{vw} \lambda } \big) \Big\}
\end{equation}
where $\lambda =|E| - d_v -d_w   +s_{vw}$.	
\end{proposition}

Furthermore, I obtain the following result relating factorisation of the hypergraph with separability of the related state. The \textit{product} of two disjoint hypergraphs  $(V_1,E_1) \sqcup (V_2,E_2)$ is a hypergraph $(V_1 \cup V_2,E_1 \sqcup E_2  )$ with vertices being the union $V_1 \cup V_2$ of vertices sets and with edges of the following form:
\[
E_1 \sqcup E_2 : =\sum_{e_1 \in E_1 ,e_2 \in E_2} e_1 \cup e_2 .
\] 
	
\begin{proposition}
\label{August1}
The excitation-state $\ket{G}$ corresponding to a $k$-hypergraph is separable, $\ket{G} =\ket{G}_{V_1} \otimes \ket{G}_{V_2}$, iff it is a product hypergraph with respect to the division $V_1 | V_2 $.
\end{proposition}

\begin{corollary}(Only for graphs) 
\label{cor5}
The excitation-state $\ket{G}$ is separable $\ket{G} =\ket{G}_{V_1} \otimes \ket{G}_{V_2}$ iff the relevant graph $G$ is complete bipartite graph, $G= K_{V_1 V_2}$. In fact, such a state forms a tensor product of two $|W\rangle$-like states:
\[
\ket{G} =\ket{W}_{V_1} \otimes \ket{W}_{V_2} .
\]
\end{corollary}

\noindent
We might reformulate separability criteria in terms of symmetries of the Dicke-like states.

\begin{corollary}
\label{Aarhus11}
A Dicke-like state  $\ket{D_N^2}_H$ with two excitations is separable across the bipartition $N_1|N_2 $ iff its group of symmetry is equal to $\mathcal{S}_{N_1} \times \mathcal{S}_{N_2}$.
\end{corollary}

Recall that the graph is called \textit{regular} if the degree $d_v$ is constant for any vertex $v$. Moreover, it is called \textit{distance-1 regular graph}, if it satisfies an additional condition: If vertices $v$ and $w$ are connected, they are connected with the same number of hyperedges. In other words, the section $s_{vw} $ takes the same values depending on the distance between vertices:
	\begin{align*}
		s_{vw}=&\begin{cases}
			s \quad \; \; &\text{for} \quad d(v,w)=1,\\
			0 \quad &\text{for} \quad  d(v,w)>1 .
		\end{cases} 
	\end{align*}
From \cref{condition}, I derive an expression for the concurrence $C_{vw}$ between subsystems corresponding to vertices of a distance-1 regular graph $G$.
	
\begin{proposition}
\label{A}
For connected, and distance-1 regular graph $G$ the concurrence $C_{vw}$  between two nodes $v$ and $w$ reads, 
\begin{align}
\label{eq7}
C_{vw}(G) = \dfrac{2}{|E|} \cdot
\begin{cases}
\text{max} \{ 0,C \} \quad  &\text{for} \quad  d(v,w)=1 \\
n_{vw} \quad &\text{for} \quad d(v,w)=2  \\
0 \quad &\text{for} \quad  d(v,w)>2
\end{cases} 
\end{align}
where $C= n_{vw} - \sqrt{s (|E| -2d+s)}$, $d$ is the degree of each node, and $s$ is a section for each adjacent vertices.
\end{proposition}
	
	\noindent
	Furthermore, an elementary calculations (presented in a detailed way in \cite{burchardt2021entanglement}) lead to the following result.

\begin{proposition}
\label{B}
Square of the generalized concurrence $C_{v |\text{rest}}^2$ between the particle $v$ and the rest of the system, expressed as a function of the number of edges $|E|$ and the number of vertices $N$, reads
\begin{equation}
\label{eq8}
C_{v |\text{rest}}^2 = 4 \; \dfrac{d_v \big( |E|-d_v\big) }{|E|^2} =4 \; \dfrac{k(N-k)}{N^2}.
\end{equation}
The second equation is valid under the assumption of the regularity of a $k$-hypergraph, i.e. the degrees of vertices are the same. 
\end{proposition}

\section{Examples of Dicke-like states}
\label{sec3}
	
In \cref{hyper,sec3a}, I presented two similar, but different, constructions of genuinely entangled states: excitation-states $\ket{G}$, related to a graph $G$, and Dicke-like states, determined by  a subgroup. A Dicke-like state can be considered as a special case of the excitation-states, which exhibits a certain symmetry structure. We combine both representations and introduce examples of excitation-states related to the graphs given by highly symmetric objects, such as \textit{regular polygons},  \textit{Platonic solids}, and \textit{regular plane tilings}. Such symmetric objects were already used in various contexts concerning multipartite entanglement including: quantification of entanglement of permutation-symmetric states \cite{Markham_2011,Martin15}, identification of quantumness of a state \cite{Goldberg_2020}, search for the maximally entangled symmetric state \cite{Aulbach_2010}, or general geometrical quantification of entanglement \cite{PhysRevA.85.032314,Giraud_2010} especially among states with imposed symmetries on the roots of Majorana polynomial \cite{Martin10,Martin15}.

Results from \cref{cccc}, allow us to discuss their entanglement properties. We computed the concurrence in two-partite subsystems for all presented examples. For excitation-states, the entanglement shared between a particular node $v$ and the rest of the system depends only on the number of parties $N$, as it was shown in \cref{{B}}. 
Hence, we may define the \textit{entanglement ratio} $\Gamma_v$ for the node $v$ as:
\begin{equation}
	\Gamma_v :=\dfrac{\sum_{i \neq v} C_{v|i}^2}{ C_{v|\text{rest}}^2} \in \left[ 0,1\right] , 
\end{equation}
which measures the ratio of entanglement shared between particular parties in a two-partite way in comparison to the amount of entanglement shared in the multi-partite way. Note that concurrence satisfies monogamy inequality (\ref{eq5}), hence the parameter  $\Gamma_v$ takes values in the range $[0,1]$. 

\begin{example} 
\label{eeerrr}
\textbf{Dicke states.} 
The Dicke states $\ket{D_N^k}$ are excitation-states for complete $k$-regular hypergraphs on $N$ vertices. Their group of symmetry is the entire permutation group $\mathcal{S}_N$. By \cref{eq7,eq8}, and elementary calculation, we show that the concurrence in two-partite subsystems reads,
\begin{align*}
C_{v|w} &=2 {{N}\choose{k}}^{-1}  \Bigg( {{N-2}\choose{k-1}}- \sqrt{{{N-2}\choose{k}}{{N-2}\choose{k-2}}} \Bigg) .
\end{align*} 
The entanglement ratio $\Gamma_v$ for a given node $v$ is equal to
\[
\Gamma_v  (\text{D}^{k}_N) 
=\frac{N-1}{
{{N-1}\choose{k}} 
{{N-1}\choose{k-1}}  }
\Bigg( {{N-2}\choose{k-1}}- \sqrt{{{N-2}\choose{k}}{{N-2}\choose{k-2}}} \Bigg)^2
.
\]
By \cref{A,B}, for the Dicke states $\ket{\text{D}^{k}_N}$ the entanglement ratio at infinite dimension is nonzero,
\begin{equation*}
\lim_{N\rightarrow\infty} \Gamma_v (\text{D}^{k}_N) = 2k-1- \sqrt{k(k-1)}.
\end{equation*}
In particular, for states related to graphs, $k=1$, we find
\begin{equation*}
\lim_{N\rightarrow\infty} \Gamma_v(\text{D}^{1}_N) = 3 - 2\sqrt{2}.
\end{equation*}
\end{example}

	\begin{figure}[ht!]
\begin{center}
\includegraphics[scale=1]{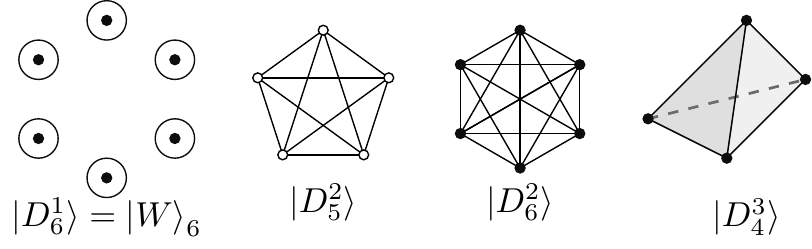}
\caption{Dicke state $\ket{\text{D}_N^k}$ are associated with a complete $k$-regular hypergraph with $N$ vertices. The graph and the corresponding state is completely symmetric, so the labels of the nodes can be omitted. Dicke states $\ket{\text{D}_N^2}$ with two excitations are related to complete graphs on $N$ vertices. A tetrahedron represents the Dicke states $\ket{\text{D}_4^1}$, $\ket{\text{D}_4^2}$ and $\ket{\text{D}_4^3}$, depending on whether we consider  the vertices, the edges or the faces of the tetrahedron.}
\label{ex1}
\end{center}
\end{figure}

\begin{example} 
\textbf{Cyclic states.} 
The simplest non-trivial subgroup of the permutation group $\mathcal{S}_N$ is a cyclic group $\mathcal{C}_N$. In general, the states $\ket{D_N^k}_{\mathcal{C}_N}$ are \textit{translationally invariant}, i.e. invariant under a cyclic permutation of qubits  \cite{TranslationallyInvariantStates,watson2020complexity}.  Such a family of states is widely considered in several 1D models in condensed-matter physics, like XY model or the  Heisenberg model. Note that such states $\ket{D_N^2}_{\mathcal{C}_N}$ with the number of excitations is equal to $k=2$ can be constructed as an excitation-state. Indeed, consider the cyclic graph on $N$ vertices. The relevant excitation-state, denoted by $\ket{C_N}$, matches perfectly $\ket{D_N^2}_{\mathcal{C}_N}$. Note that the group of automorphisms of a cyclic graph $C_N$ is a \textit{Dihedral} group $\mathcal{D}_{2N}$, where the lower index stands for the order of the group $|\mathcal{D}_{2N} |=2N$. Hence we conclude that the group of symmetries of $\ket{D_N^2}_{\mathcal{C}_N} = \ket{C_N}$ is a dihedral group $\mathcal{D}_{2N}$. Systems with dihedral $\mathcal{D}_{2N}$ symmetries were consider in a context of correlation theory of the chemical bond \cite{ding2020concept}. Molecules invariant under a rotation and inversion were investigated ibidem. The concurrence in two-partite subsystems takes the following value:
\begin{align*}
C_{v|w} ({C_N}) &=
\begin{cases}
2/N \quad \; \; &\text{for} \quad d(v,w)=2\\
0 \quad &\text{otherwise} 
\end{cases} ,
\end{align*} 
with a small correction for $N=4$, where $C^2_{v|w} =1$ for distance two vertices. The entanglement ratio $\Gamma_v$ for a given node reads
	\[
	\Gamma_v (C_N) =\dfrac{1}{N-2},
	\]
with the same correction in the case  $N=4$, for which $\Gamma_v=1$.
\end{example}

\begin{figure}[ht!]
\begin{center}
\includegraphics[scale=1]{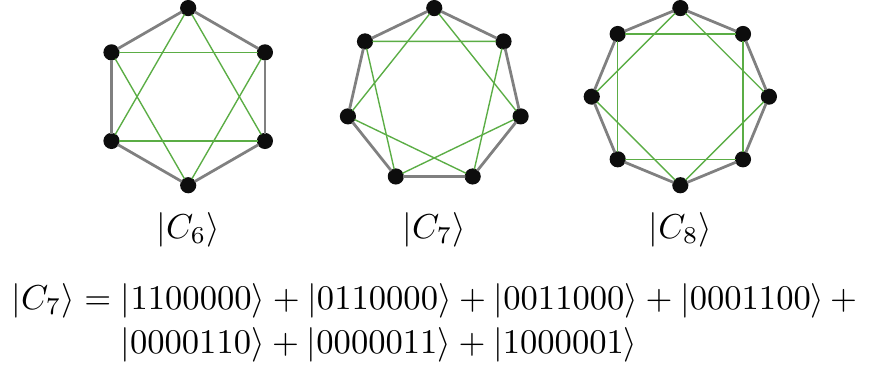}
\caption{Three cycle graphs: $C_6$, $C_7$ and $C_8$ indicated by gray edges. The form of the state $\ket{C_7}$, corresponding to $C_7$, is listed below. Entanglement  on two-parties subsystem is vanishing except of parties of distance two. Non-vanishing entanglement is indicated by green lines connecting relevant parties.}
\label{CyclicStates}
\end{center}
\end{figure}

\begin{example}
\label{Platonic states}
\textbf{Platonic states.} 
Platonic solids were used to construct quantum states in various contexts \cite{Giraud_2010,PhysRevA.85.032314,Goldberg_2020}. Here, with any  Platonic solid we associate two quantum excitation-states, by looking at the edges and faces of related solid. For instance, the tetrahedron is linked to the Dicke state $\ket{\text{D}_4^2}$ (by reading edges) and $\ket{\text{D}_4^3}$ (by reading faces), see \cref{ex1}. We denote such state by $\ket{P^e}_N$ and $\ket{P^f}_N$ respectively. An elementary argument from the representation theory shows that group of symmetries of the Platonic solid, determines the symmetry of related states. In such a way, we constructed states with symmetries $\mathcal{S}_4$, $\mathcal{S}_4 \times \mathcal{S}_2 $, and $\mathcal{A}_5 \times \mathcal{S}_2$ respectively. In particular, in the case of dodecahedron the alternating symmetry $A_5$ is observed, which is not easy to achieve. The concurrence  $C_{12}$ and the entanglement ratio $\Gamma_v$ for two-qubit systems obtained by partial trace of $(N-2)$ subsystems of distance-two of Platonic states determined by the edges of the solid is compared below.
\begin{table}[ht]
\centering 
\begin{tabular}{l c c c c c} 
\hline
			& $\quad \; \ket{P^e}_4\;$ & $\;\ket{P^e}_6$ & $\;\ket{P^e}_{8}$ &$\;\ket{P^e}_{12}$&$\;\ket{P^e}_{20}$ \\ [0.5ex] 
			\hline 
			Concurrence $C_{12}  $& 0.333 & 0.667 & 0.333 & 0.133 & 0.067 \\ [1ex]
			Ent. Ratio $\Gamma_v$ & 0.333& 0.500 & 0.444 & 0.160 & 0.074\\
			\hline 
\end{tabular}
\end{table}
\end{example}

\begin{figure}[ht!]
\includegraphics[scale=0.7]{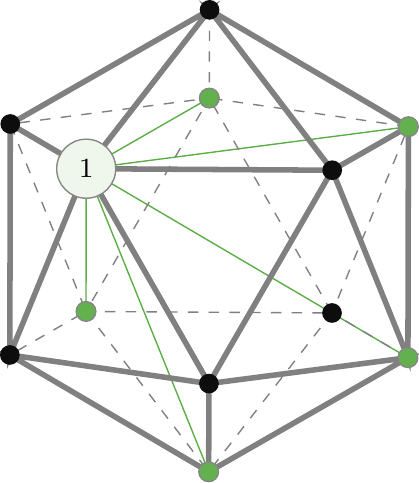}
\hspace{2cm}
\includegraphics[scale=0.8]{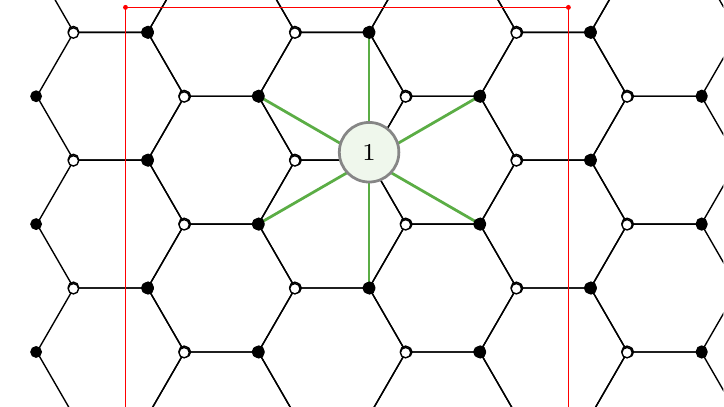}
\caption{On the left, Icosahedron, one of five Platonic solids. Its edges (indicated by gray) determine the excitation state $\ket{P^e}_{20}$ on 20 qubits. On the right, hexagonal tiling of a plane. The red rectangle indicates a finite region for which the boundary left-right and up-down is glued together. Such a region determines the excitation state $\ket{H}_{20}$ on 20 qubits. On both pictures, for a chosen node $1$, the entanglement in the two-party subsystem $C_{1 v}$ is vanishing except for subsystems corresponding to distance-two nodes, indicated by green lines and green nodes.}
\label{PlatonicStates}
\end{figure}

\begin{example}
\textbf{Regular $m$-polytope families.} 
There are three natural generalizations of Platonic solids in higher dimensions: the self-dual $m$-simplices and $m$-hypercubes with dual $m$-orthoplexes. Each of these polytopes provides a set of $k$-uniform hypergraphs for $1 \leq k \leq m-1$ defined by the set of their $k$-dimensional hyperedges. We denote the states corresponding to the $k$-dimensional hyperedges of $m$-simplex as $\ket{\text{S}^k_m}$ and analogously $\ket{\text{B}^k_{2^m}}$ for $m$-hypercubes and $\ket{\text{O}^k_{2m}}$ for $m$-orthoplexes, where the lower index stands for the number of subsystems $N$ in the state, note that $N=m$, $2^m$, $2m$ for $m$-simplex, $m$-hypercube, and $m$-orthoplexe respectively. The states related to the $m$-simplices are equivalent to the Dicke states $\ket{\text{S}^k_m} = \ket{D^k_m}$ and hence exhibit the full symmetry. On the other hand, the symmetry group of $\ket{\text{B}^k_{2^m}}$ and $\ket{\text{O}^k_{2m}}$ is given by the non-trivial hyperoctahedral group $B_m$. The states $\ket{\text{O}^{m-1}_{2m}}$ are separable with respect to the partition $12|34|...|(2m-1)2m$, with each pair of vertices being 2-distance or, in other words, lying on a common diagonal of the related $m$-orthoplex. Furthermore, one can easily calculate the concurrence $C_{vw}$ for the states $\ket{\text{O}^{2}_{2m}}$ connected to the 2-edges of the $m$-orthoplex,
\begin{equation}
		C_{vw} (\text{O}^{2}_{2m}) = \begin{cases}
			\text{max}\{0,\,\frac{\sqrt{2 m^2-4 m+3}-2 m+4}{2 m-2 m^2}\}   &  
			\!\! \text{for} \  d(v,w)=1 \\
			2/m    &\!\! \text{for} \  d(v,w)=2\\
			0    &\!\! \text{for} \  d(v,w)>2,
		\end{cases}
\end{equation}
and similarly the entanglement ratio
	\begin{equation}
		\Gamma_v (\text{O}^{2}_{2m}) = \begin{cases}
			1/(m-1)  \\ 
			 \frac{6 m^2-4 m\left(\sqrt{2 m^2-6 m+5}+5\right) +8 \sqrt{2 m^2-6 m+5}+19}{2 (m-1)^2},  
		\end{cases}
	\end{equation}
where the first case holds for $m \leq 3$ and the second one for $m>3$. With \cref{A,B} at hand, one may show that the entanglement ratio for the states $\ket{\text{O}^{2}_{2m}}$ converges to a nonzero value,
		\begin{equation}
			\lim_{n\rightarrow\infty} \Gamma_v (\text{O}^{2}_{2m}) = 3 - 2\sqrt{2}\approx 0.1716.
		\end{equation}
The situation is simpler for the hypercubic states $\ket{\text{B}^{2}_{2^m}}$, where the concurrence occurs only between the distance-2 vertices, $C_{12} (\text{B}^{2}_{2^m})= 2^{3-m}/m$, while the entanglement ratio reads 
\[
\Gamma_v (\text{B}^{2}_{2^m}) = \frac{4 (m-1)}{\left(2^m-2\right) m},
\]
which asymptotically tends to zero.
\end{example}

\begin{example} 
\textbf{Plane regular tilings.} 
Regular and semi-regular tessellation of the plane are yet another highly symmetrical objects. Regular tilings correspond to Dicke-like states. The group of symmetry is given by a relevant wallpaper group, restricted to the chosen size of the tiling. \cref{PlatonicStates} presents the hexagonal tiling of the plane and relevant excitation-state $\ket{H_N}$. For such a tiling, the concurrence in bipartite subsystems takes the following value:
\begin{align*}
		C_{v|w}({{H_N}}) &=
		\begin{cases}
			2/N \quad \; \; &\text{for} \quad d(v,w)=2\\
			0 \quad &\text{otherwise} 
		\end{cases} ,
\end{align*} 
for the tiling restricted to $N$ nodes (minimal size of a cut is $3\times 3$). Furthermore, the entanglement ratio $\Gamma_v$ for a given node $v$ reads,
\[
\Gamma_v ({H_N}) =\dfrac{4}{3}\dfrac{1}{N-2}.
\]
Although the value of the two-partite concurrence $C_{v|w}$ is the same as in the cyclic case, the parameter $\Gamma_v$ takes a larger value. 
\end{example}

\begin{figure}[ht!]
\includegraphics[scale=1]{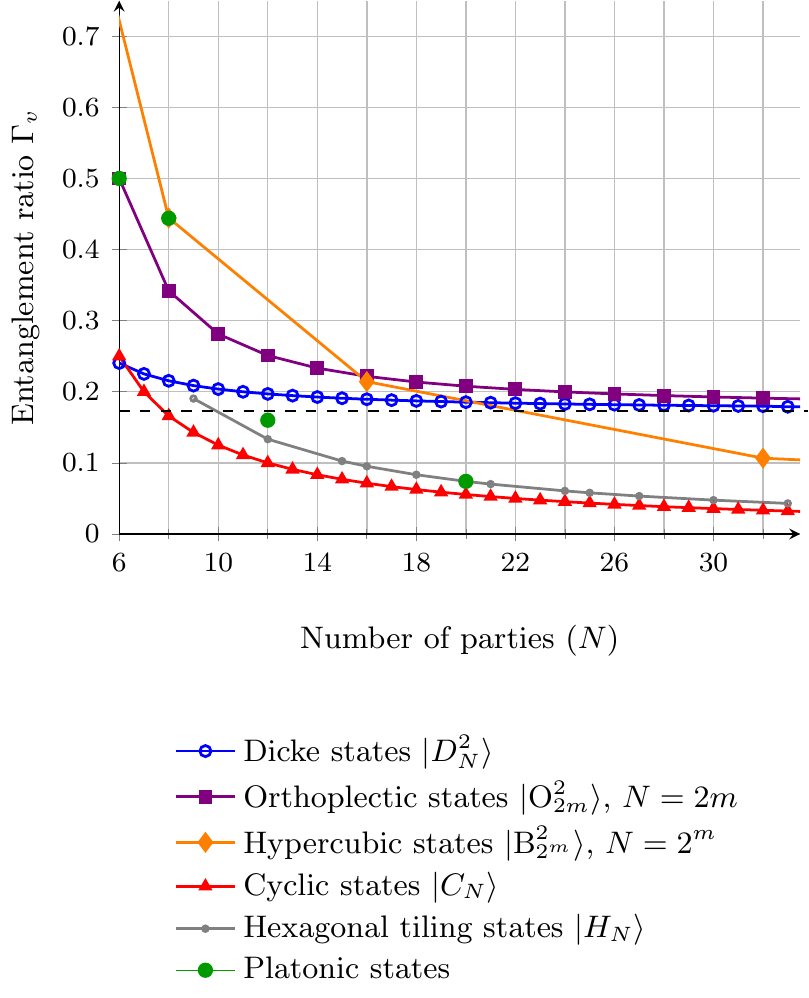}			
\caption{Distribution of entanglement for Dicke-like states with $k=2$ excitations is compared. The distribution of entanglement depends on the density of related graph. For states related to dense graphs, $\ket{D^2_N}$ and $\ket{\text{O}^2_{2m}}$, the ratio $\Gamma_v$ tends to a nonzero value $3 - 2\sqrt{2}\approx 0.17$ (indicated by the horizontal dashed line). For families of states in which the local degree does not scale with the size of a graph ($\ket{\text{B}^2_{2^m}}$, $\ket{C_N}$, $\ket{H_N}$), the complexity of $\Gamma_v$ is $\mathcal{O} (1/N)$. 
}
		\label{Comparison}
	\end{figure}

\section{Quantum circuits} 
	\label{Circuit}
	
We present below a quantum circuit efficiently transforming a  separable tensor product state into the excitations-state, related to any graph $G$. Presented construction was inspired by similar circuits for Dicke states recently developed in \cite{DickeCircuits,brtschi2019deterministic}. As we showed in \cite{burchardt2021entanglement}, our scheme uses between $\sim 4 |V|$ and $\sim 10 |E|$ CNOT gates, depending on the structure of a graph.

\textbf{First step}. Choose arbitrary vertex $v$, and consider all adjacent vertices $v_1 ,\ldots , v_d$, where $d$ is the degree of $v$. Suppose that each $v_i$ is related to the $i$th particle, while $v$ is related to $d+1$ particle. For $i=1, \ldots,d-2$, we consecutively apply the following three-qubit gates $U_{ \{i, d, d+1 \}}^{(1)}$ on parties $i,d,d+1$:
	\begin{align}
		\label{al0}
		\frac{1}{\sqrt{i}} \big(\ket{101} + \sqrt{i-1} \ket{011} \big) &\longmapsto  \ket{011} \\
		\ket{110} &\longmapsto \ket{110} \nonumber \\
		\ket{100} &\longmapsto \ket{100} \nonumber \\
		\ket{010} &\longmapsto \ket{010} \nonumber \\
		\ket{001} &\longmapsto \ket{001} \nonumber \\
		\ket{000} &\longmapsto \ket{000} \nonumber ,
	\end{align}
where the action on the remaining subspace is arbitrary. Application of $U_{ \{i,d, d+1 \}}^{(1)}$ operation is relevant to the graphical operation of deleting an edge $e= \{i,d+1 \}$, see \cref{circuit}. 
	
	Secondly, we apply the following three-qubit operator $U_{ \{d-1, d, d+1 \}}^{(2)}$ on parties $d-1,d,d+1$
	\begin{align}
		\label{al1}
\frac{1}{\sqrt{d-1}}\big(		\ket{101} + \sqrt{d-2} \ket{001} \big) &\longmapsto  \ket{001} \\
		\ket{011} &\longmapsto \ket{011} \nonumber \\
		\ket{100} &\longmapsto \ket{100} \nonumber \\
		\ket{010} &\longmapsto \ket{010} \nonumber \\
		\ket{000} &\longmapsto \ket{000} \nonumber ,
	\end{align}
where the action on the remaining subspace is arbitrary.
This is related to the graphical operation of deleting an edge $e= \{d-1,d+1 \}$, see \cref{circuit}.

Finally we apply the following two-qubit operator $U_{ \{d, d+1 \}}^{(3)}$ on parties $d, d+1$:
\begin{align}
\label{al2}
\frac{1}{\sqrt{d}}\big(\ket{11} + \sqrt{d-1} \ket{01} \big) &\longmapsto  \ket{01} \\
\ket{00} &\longmapsto \ket{00} \nonumber \\
\ket{10} &\longmapsto \ket{10}  \nonumber  
\end{align}
where the action on the remaining subspace is uniquely determined. This is relevant to deleting the only remaining edge: $e= \{d,d+1 \}$, see \cref{circuit}. There is a simple logic behind those three operations. We combine all terms having excitations on position $d+1$ into a single term with an excitation on this position. After applying these operations, the state takes the form of the superposition:
\begin{equation*}
\dfrac{1}{\sqrt{|E|}}
\Big( \sqrt{d} \ket{1}_{v} \otimes \ket{0\ldots 0}+  \ket{0}_{v} \otimes  
\sum_{e \in E  \setminus v} \ket{\psi_e}
\Big)  ,
\end{equation*}
where $E  \setminus v$ denotes a set of edges which do not contain vertex $v$.

\textbf{Second and the next steps}. Consider the graph $G' =(V \setminus v, E\setminus v)$ with deleted vertex $v$ and all adjacent edges. We repeat the procedure from the first step for an arbitrary vertex $w$ from the graph $G'$. We repeat this procedure further, for $G'' =(V \setminus v,w; E\setminus v,w)$, until we delete $N-1$ vertices, which fully separates the initial graph.
	
\textbf{Final step}. After applying presented procedure iteratively $N-1$ times, the state takes the form:
\begin{equation*}
		\dfrac{1}{\sqrt{|E|}} \sum_{v\in V} \sqrt{d_{v}'} \ket{1}_{v} \otimes \ket{0\ldots 0}_{v^c} ,
\end{equation*}
where $d_v'$ denotes the degree of the vertex $v$ of the graph reduced according to the procedure. The state above is similar to the state 
\[
\ket{\text{W}_N}:=\frac{1}{\sqrt{N}}
\big( \ket{10\cdots 0}+\cdots + \ket{0\cdots 01}\big)
. 
\]
The separation of $\ket{\text{W}_N}$ state is a well-known procedure \cite{Wgenerating}, and might be obtained by performing two-qubit gates $U_{1i}^{(4)}$:
	\begin{align*}
		\scriptscriptstyle
		\sqrt{d_{v_1}'+ \ldots +d_{v_{i-1}}' } 
		\textstyle
		\ket{10} +
		\scriptscriptstyle
		\sqrt{d_{v_i}'} 
		\textstyle
		\ket{01} &\longmapsto 
		\scriptscriptstyle
		\sqrt{d_{v_1}'+ \ldots + d_{v_{i}}' } 
		\textstyle
		\ket{10} \\
		\ket{00} &\longmapsto \ket{00} \nonumber 
	\end{align*}
	on particles $1$ and $i$. We refer to \cite{burchardt2021entanglement} for the precise estimation of the computational cost for the presented procedure.
	
\begin{figure}[ht!]
\includegraphics[scale=1]{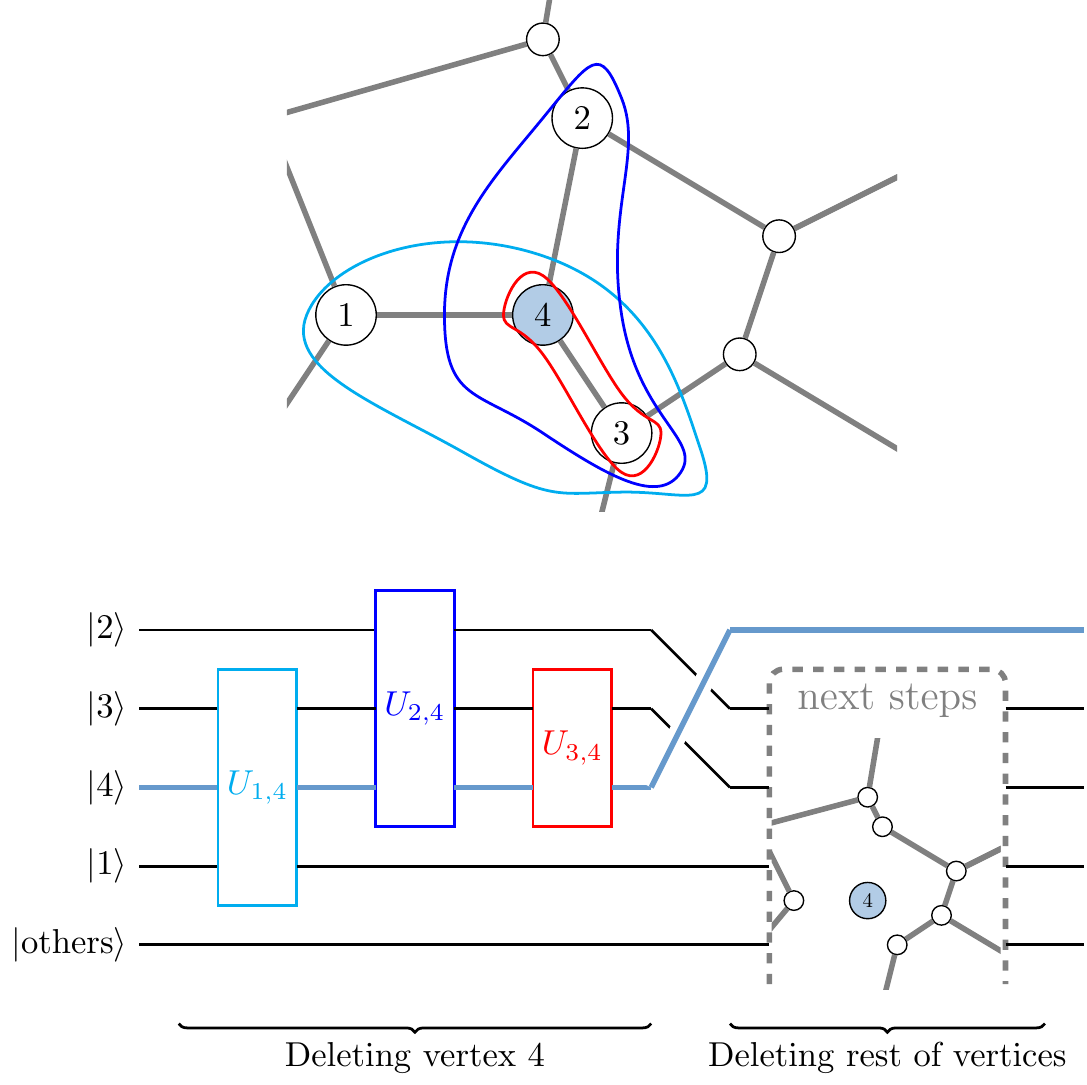}			
		\caption{Procedure sketched for vertex $4$ of degree $d=3$. Three operations described in the first step of the procedure: firstly, we apply unitary three-qubit operations $U_{ \{1,3,4 \}}^{(1)}$ on parties $1,3,4$, secondly $U_{ \{2,3,4 \}}^{(2)}$ on parties $2,3,4$, and finally the two-qubit gate $U_{ \{3,4 \}}^{(3)}$ on parties $3,4$. Resulting state is separable with respect to the $4$-th particle. This is graphically represented by deleting all adjacent edges. In the consecutive steps, we apply this procedure iteratively to the remaining vertices.}
\label{circuit}
\end{figure}
	
\begin{figure}[ht!]
\centering
\includegraphics[scale=0.79]{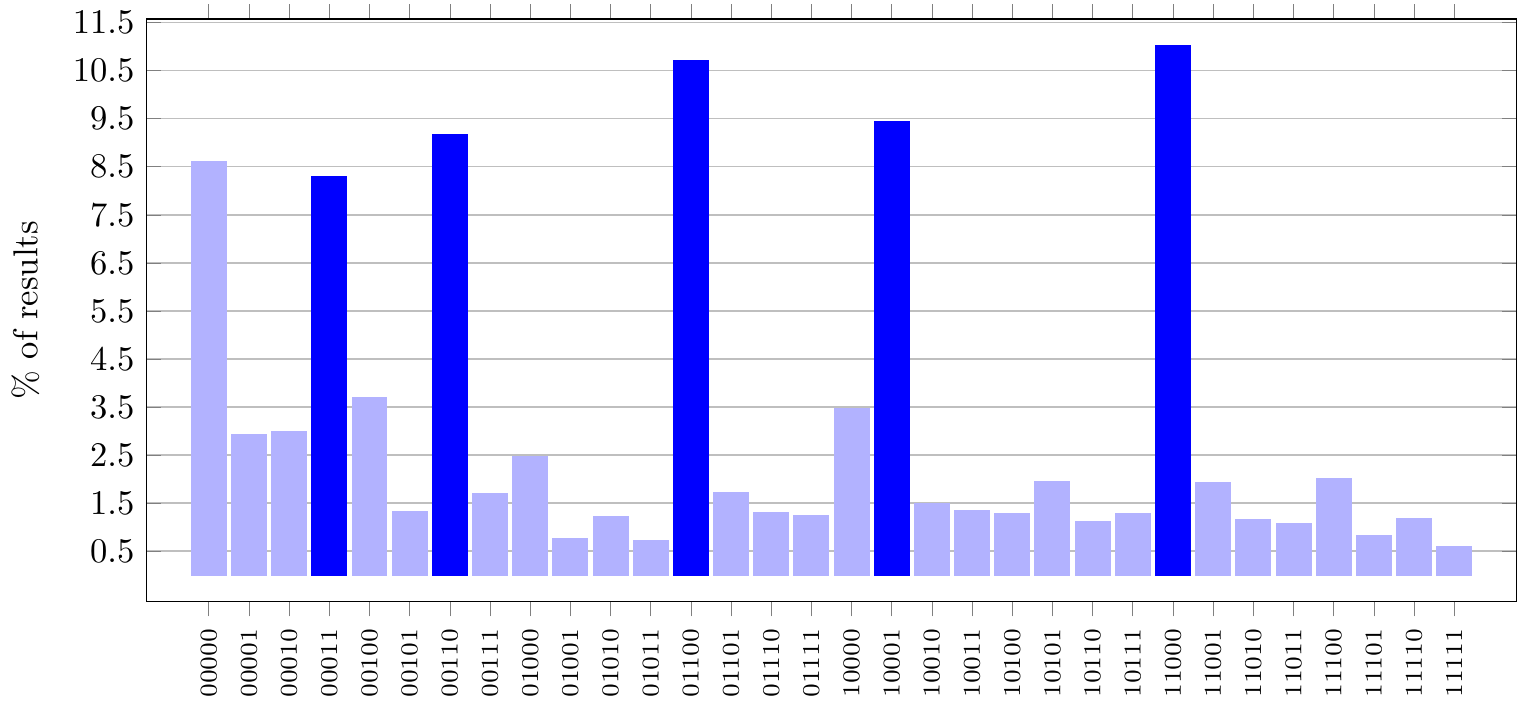}		
		\caption{Mean results obtained over total 740 000 samples acquired on IBM Vigo, Athens and Santiago Computers. Prominent peaks at all cyclic permutations of $\ket{00011}$ are plotted in dark, with an additional significant peak at $\ket{00000}$. Due to effects of decoherence we observe a background of all remaining states represented by light bars. }
\label{cyclic_hist}
\end{figure}

We simulated the simplest nontrivial cyclic state $\ket{C_5}$, in order to demonstrate the viability of the provided construction. We refer to \cite{burchardt2021entanglement} for a detailed description. The overall circuit requires $22$ CNOT operations, where the topology of the quantum computer is not taken into account. Such a circuit can be realised on the state-of-the-art 5-qubit quantum computers provided by IBM -- linear-topology Santiago and Athens with quantum volumes (QV) of 32 and Vigo with QV of 16 with T topology. In total we used more than 740 000 samples over all three computers, which gives distribution with significant values for all expected computational states proceeding from cyclic permutations of the state $\ket{00011}$ with the probability $0.487$  of finding the system in one of them, see \cref{cyclic_hist}. Further analises and the comparison of obtained results to a model of noise might be find in \cite{burchardt2021entanglement}.
	
\section{Hamiltonians}
\label{Hamiltonians}
	
The Dicke states were introduced as a ground states of the Hamiltonians with a well--defined number of excitations. The excitation-states can be also considered as ground states of  analogous Hamiltonians, with the same property of a well defined number of excitations. The single-mode Dicke Hamiltonian (known also as Tavis-Cummings model or generalized Jaynes-Cummings model) has the following form \cite{PhysRevE.67.066203,DickeModelRevisited}:
	\[
	H = \omega_0 \; J_z+ 
	\omega \; a^\dagger a +
	\dfrac{\lambda}{\sqrt{N}} \; \underbrace{(a^\dagger +a)(J_{-} + J_{+})}_{\text{interaction}}\; ,
	\]
	where $J_\alpha$ are the collective operators:
	\[
	J_z \equiv \sum_{i=1}^N \sigma_z^{i} , \quad 
	J_\pm \equiv \sum_{i=1}^N \sigma_\pm^{i} . 
	\]
By fixing interaction,  $\lambda =0$, the coupling term vanishes. Thus eigenstates have the simple tensor product form with a factor representing the Fock states of the field and the other factor as an eigenstate of the collective angular momentum operator $J_z$. Operator $J_z$  has the following degenerated eigenvalues: $- N/2, - N/2+1 ,\ldots ,N/2$. The further partition of its eigenspaces might be performed by the square of the total angular momentum operator $J^2$, which can be written as $J^2 \equiv J_{+} J_{-} + J_z (J_z+ \mathbb{I} )$, in terms of the raising and lowering operators. The Dicke states $\ket{j,m}_N$ (in the more classical notation) are eigenstates of both operators $J_z$ and $J^2$:
	\begin{align*}
		J_z \ket{j,m}_N &= m \ket{j,m}_N,\\
		J^2 \ket{j,m}_N &= j(j+1) \ket{j,m}_N 
	\end{align*}
with $m =-j, -j+1 ,\ldots ,j$ corresponding to the number of excitations, while $j =1/2 ,3/2 ,\ldots ,N/2$ (for $N$ odd) and $j =0,1,\ldots , N/2$ (for $N$ even) related to the cooperation number. In general, atomic configurations for $N>2$ contain entanglement \cite{Wang_2002} and are degenerated for $j < N/2$ \cite{DickeModelRevisited}. By taking the maximal value $j=N/2$, we obtain fully symmetric Dicke states $\ket{D_N^k}$,
	\[
	\ket{D_N^k} := \ket{N/2 ,k +N/2}_N .
	\]

Independently of the number of excitations $k$, the Dicke states $\ket{D_N^k}$ are eigenstates of the operator $J_{+} J_{-}$. Such an operator can be rewritten as a fully symmetric, two-body Hamiltonian, 
	\[
	H :=J_{+} J_{-} = \sum_{i=1}^N \sigma_{-}^{(i)} \sigma_{+}^{(i)}  +\sum_{i \neq j}^N \sigma_{-}^{(i)} \sigma_{+}^{(j)} ,
	\]
in which any two parties might exchange excitations. Note that the Dicke states $\ket{D_N^k}$ are eigenvectors corresponding to the non-degenerated maximal eigenvalue $k (N+1-k)$. In various physical scenarios, the particles are not distributed symmetrically as some interactions are more likely to occur. For instance, the two-body Hamiltonian of the system of $N$ particles with spin-$1/2$ considered in \cite{Huber_2016} generalizes the  Heisenberg model. 
	
We shall see that the excitation-states, and hence Dicke-like states, are ground states for Hamiltonians with an interaction term  analogous to $J_{+} J_{-}$. As in the Dicke case they have a well-defined excitation number. We restrict ourselves to the subspace of the operator $J_z$ relevant to two-excitations, i.e.  we set $k=2$. For any graph $G$, consider the four-body Hamiltonian, where any pair of edges might exchange excitations:
\begin{equation}
H_G :=  J_{+}^G J_{-}^G ,
\label{DH2}
\quad
\text{where}
\quad
	J_{\pm}^G \equiv \sum_{v v' \in E} \sigma_{\pm}^{(v)} \sigma_{\pm}^{(v')} \; .
\end{equation}
Notice the similarity between operators $J_{+}^G, J_{-}^G $ and $J_{+} J_{-}$. The maximal eigenvalue of the Hamiltonian $H_G$ is equal to the squared number of edges $|E|^2$. This energy level is non-degenerated and the corresponding eigenstate is exactly given by  the excitation-state $\ket{G}$. Indeed, the Hamiltonian $H_G$ is simply a sum of $|E|^2$ projections, which justifies the bound of the energy. The state, $\ket{G}$ saturates this bound, contrary to any other state with specified number of excitations, $k=2$.

Furthermore, the excitation-state $\ket{G}$ might be seen as a ground state of the three-body interaction Hamiltonian,
	\begin{equation}
	\label{3bodyH}
	H_G = \sum_{v \in V} \sum_{ \substack{ vv' \in E \\ v'v'' \in E}} \sigma_{-}^{(v)} \big( \sigma_{-}^{} \sigma_{+}^{} \big)^{(v')}  \sigma_{+}^{(v'')} .
	\end{equation}
Above Hamiltonian describes situation, in which the pairs of excitations might be exchanged only between connected edges. The maximal eigenvalue is equal to ${d \choose 2} |V|$, where $d$ denotes the degree of a vertex. In analogy to the Hamiltonian (\ref{DH2}) the maximal eigenvalue is not degenerated and the state $\ket{G}$ forms the relevant eigenvector. Observe that the Dicke state $\ket{D^2_N}$ is a ground state of the Hamiltonian associated with the fully connected graph $G$ on $N$ vertices.
	
It is worth mentioning that $\sigma_-\sigma_+ = (\mathbb{I} - \sigma_z)/2 = |1\rangle\langle1|$ and the Hamiltonian can be written in an alternative way:
\begin{equation}
	\label{3bodyHbis}
	H_G = \sum_{v \in V} \sum_{ \substack{ vv' \in E \\ v'v'' \in E}} \sigma_{-}^{(v)} |1\rangle\langle1|^{(v')}  \sigma_{+}^{(v'')} .
	\end{equation}
In this form the term $\sigma_-\sigma_+$ might be interpreted as conditional hopping interaction - if the site $v'$ is occupied, the hopping from $v$ to $v''$ is effected, otherwise no interaction happens. Such a behaviour is very much akin to what is called the quantum transistor. One of the simplest models of such a quantum transistor  \cite{loft2018quantum} involve left and right qubit and a two qubit gate, an interaction very similar to observed here.

	Furthermore, the excitation-state $\ket{G}$ related to the $k$-uniform hypergraphs $G$ is described as the singular ground state of the Hamiltonian with $2k$-body interaction, $H= J_{+}^G J_{-}^G $. The excitation-state $\ket{G}$ is a ground state of the Hamiltonian (\ref{DH2}), where the operators $J_{\pm}^G$ are defined as
	\[
	J_{\pm}^G \equiv \sum_{ \{v^1 , \ldots ,v^k \} \in E} \sigma_{\pm}^{(v^1)} \cdots\sigma_{\pm}^{(v^k)} .
	\]
	
	In summary, as the Dicke states $\ket{D_N^m}$ are uniquely determined as eigenvectors of operators $J_z$ and the two-body interaction Hamiltonian $H=J_{+} J_{-}$, the excitation-states are also eigenvectors of an operator $J_z$, nevertheless the further division of its subspaces differs from the one corresponding to the Dicke states. For any excitation state $\ket{G}$ one may construct a $(k+1)$-body interaction Hamiltonian with the state $\ket{G}$ being its non-degenerated ground state. Notice that the construction of the Hamiltonian depends on on the number of excitations. 

\section{Conclusions}

In this chapter, I introduced the notion of $H$-symmetric states, where $H$ is an arbitrary subgroup of permutation group (\cref{Aarhus9}). I show a natural example of $H$ symmetric state for any subgroup of symmetric group $H< \mathcal{S}_N$ (\cref{Aarhus8}). Note that any group $H$ might be realized as a subgroup of the symmetric group $S_N$ for $N$ arbitrary large. Furthermore, I introduce two related families of states with remarkable symmetric properties: in \cref{Aarhus10} excitation-states $\ket{G}$, related to a graph $G$, and in \cref{Dickelike} Dicke-like states, determined by an embedding of a subgroup $H< \mathcal{S}_N$. \cref{cccc} presents results concerning entanglement properties of general excitation states and Dike-like states (\cref{condition,August1,cor5,Aarhus11}). With some additional assumptions on the regularity of the graph, I obtained the form of $2$-partie concurrence (\cref{A}), and generalized concurrence (\cref{B}) of the related excitation states. In particular, I have shown that: 
\begin{enumerate}
\item 
Generalized concurrence $C_{v |\text{rest}}$ describing entanglement between given subsystem $v$ and the rest of the system depends only on the total number of nodes and uniformity of the graph.
\item 
Concurrence $C_{vw}$ between nodes $v$ and $w$ is positive only for distance-two nodes, and in states related to almost complete hypergraphs, i.e. $|E| \sim {{N}\choose{k}}$, also for distance-one nodes. 
\item 
Concurrence $C_{vw}$ between nodes of distance two is proportional to the number of shared neighbors. 
\end{enumerate}
\noindent
Furthermore, \cref{sec3} introduces several examples of excitation and Dicke-like states of given symmetries: cyclic, dihedral, symmetry of Platonic solids. \cref{Comparison} compares entanglement properties of introduced families of states. 
Moreover, I show two different methods of constructing introduced families of states: \cref{Hamiltonians} presents excitation states as ground states of Hamiltonians with $3$-body interactions, while \cref{Circuit} proposes relevant quantum circuits. The complexity of introduced circuits is comparable with the complexity of quantum circuits proposed for Dicke states known in the literature. A simple example of a cyclic state of five qubits is successfully simulated on available quantum computers: IBM - Santiago and Athens -- see \cref{cyclic_hist}.

\chapter{Absolutely maximally entangled states}
\label{chapter3}

Absolutely Maximally Entangled (AME) states of a multipartite quantum system are maximally entangled for every bipartition of the system. AME states are special cases of \emph{k-uniform} states characterized by the property that all of their reductions to k parties are maximally mixed. Both classes of states are crucial resources for various quantum information protocols. In this chapter, I briefly recall correspondence between AME states and classical combinatorial designs. I focus my attention on the different linear structures of classical designs that affect the structure of the related AME and k-uniform states. 

\section{AME and k-uniform states}

A multipartite quantum state $\ket{\psi} \in \mathcal{H}_d^{\otimes N}$ of $N$ parties with a local dimension $d$ each is called AME, if is maximally entangled for every of its bipartition, i.e. the partial trace 
\begin{equation}
\tr_S \ket{\psi}\bra{\psi} \propto \Id,
\end{equation} 
for any subsystem $S$ of $|S|=\lceil N/2 \rceil$ parties \cite{PhysRevA.69.052330}. 
AME states of $N$ partise with the local dimension $d$ are denoted as AME($N,d$). 
The class of AME states is being applied in several branches of quantum information theory: in quantum secret sharing protocols \cite{HelwigAME}, in parallel open-destination teleportation \cite{helwig2013absolutely}, in holographic quantum error correcting codes \cite{Pastawski2015HolographicQE,MazurekGrudka}, among many others. 
The state AME($N,d$) allows one to construct a \textit{pure quantum error correction code} (pure QECC), which saturates the Singleton bound \cite{HelwigAME}. 
Particular attention is paid to AME states of an even number of parties, those are equivalent to notions as perfect tensors \cite{Pastawski2015HolographicQE} or multiunitary matrices \cite{MultiUnitary}. AME states are special type of $k$-uniform states. 

\begin{definition}
\label{Aarhus13}
A quantum state $\ket{\psi} \in \mathcal{H}_d^{\otimes N}$ of $N$ is \emph{$k$-uniform} if its reduced density matrices are maximally mixed, i.e.
\[
\rho_S (\psi ) := \tr_{\overline{S}} \ket{\psi}\bra{\psi} \propto \Id
\]
for any subsystem $S$ of $k$ parties ($|\overline{S}|$=k) and the complementary subsystem $\overline{S}$. 
It is known that the uniformity $k$ cannot exceed $\lfloor N/2 \rfloor$ \cite{PhysRevA.69.052330}. 
States which saturate this bound, i.e. $\lfloor N/2 \rfloor$-uniform states, are called AME states. 
\end{definition}

\begin{example}
\label{ex1GHZ}
Greenberger–Horne–Zeilinger (GHZ) state
\[
\ket{\text{GHZ}} = \dfrac{1}{\sqrt{2}} \Big( \ket{00 0}  +\ket{111} \Big)
\] 
is a $1$-uniform state. Its natural generalization to $N$ parties with $d$ distinguishable energy levels:
\[
\ket{\text{GHZ}_d^N} =\dfrac{1}{\sqrt{d}}\Big( \ket{0\cdots 0} +\cdots +\ket{d-1 \cdots d-1} \Big)
\] 
is also $1$-uniform. 
\end{example}

\begin{example}
\label{ex2}
The following state of four qutrits
\begin{align*}
\ket{\text{AME(4,3)}} = \dfrac{1}{3} \Big(
&\ket{0000}+ \ket{0121}+ \ket{0212}+\\
&\ket{1110}+ \ket{1201}+ \ket{1022}+&&\\
&\ket{2220}+ \ket{2011}+ \ket{2102} \Big)&&
\end{align*}
is 2-uniform, so it is an AME state of 4 qutrits  \cite{Helwig2013AbsolutelyME}. 
It reveals larger entanglement properties than the corresponding $\ket{\text{GHZ}_3^4}$ state of four qutrits.
\end{example}

Note that both states presented in \cref{ex1GHZ,ex2} might be written in simple closed formulas: 
\begin{align}
\ket{\text{GHZ}_d^N} =&\dfrac{1}{\sqrt{d}}\sum_{i=0}^{d-1}\ket{i,\ldots, i},\nonumber\\ 
\ket{\text{AME(4,d)}} =&\dfrac{1}{d} \sum_{i,j=0}^{d-1}\ket{i,j,i+j,2i+j}. \label{AME43p} 
\end{align}
where the summation is understood modulo $d$.

AME($N,d$) states are maximizing entanglement properties among all $N$-parties states, each with $d$ levels \cite{HelwigAME}. 
There is no general construction of AME($N,d$) state, for an arbitrary number of parties $N$ and an arbitrary number of enery levels $d$.  Surprisingly, AME states do not exist for any numbers $N$ and $d$. Indeed, it was first observed by Higuchi and Sudbery in their study of bipartite entanglement that AME state of four qubits does not exist \cite{HIGUCHI2000213}. 
Until today, more of such negative results are known \cite{Felix72,Huber_2018}.

\section{Orthogonal arrays}

Combinatorial mathematics deals with the 
existence and properties of designs composed of elements of a finite set 
and arranged with certain symmetry and balance \cite{CD07}.
A simple example of a combinatorial design is given by a single Orthogonal array \cite{OA}. Orthogonal array is a combinatorial arrangements, tables with entries satisfying given orthogonal properties \cite{OAcars}. 
A tight connection between OAs and (quantum) error-correcting codes \cite{OA}, and maximally entangled states \cite{DiK} brought a new life for these combinatorial objects. 

As we already presented in \cref{Aarhus2}, an orthogonal array $\oa{r,N,d,k}$ is a table composed by $r$ rows, $N$ columns with entries taken from $0,\ldots,d-1$ in such a way that each subset of $k$ columns contains all possible combinations of symbols with the same amount of repetitions, see \cref{OA1}. The number $k$ is known as the \textit{strength} of the OA. 
We assume that OAs are simple, i.e. all rows are distinct. 
Notice, that the minimal number of rows in OA is equal to $r=d^k$, OAs saturating this bound are called of \emph{Index Unity}. 
In general \textit{the index} of the OA 
\[
\lambda =\dfrac{r}{d^k} 
\]
is always a natural number. 
An $\oa{r,N,d,k}$ is said to be an \textit{irredundant} orthogonal array (IrOA) if in any subset of $N-k$ columns all combinations of symbols are different \cite{DiK}. 

Usually, an orthogonal array is called \textit{linear} if the set of $r$ rows form a vector space over the Galois field $\GF{d}$. 
The linearity condition is equivalent to the following: for each two rows $R_1$ and $R_2$ of an OA, and any two elements $c_1,c_2 \in \GF{d}$, 
\begin{equation}
\label{linearity}
c_1 R_1 + c_2 R_2
\end{equation}
is also a row in the OA. 
Such an OA can be represented by a $N\times s$ \textit{generator matrix}, whose rows form a basis of aforementioned vector space, see \cref{OA1}. 
Up to the isomorphism of a linear space, the generator matrix $G$ can always be written in the
standard form $G=[\Id_s |A]$, where $A$ is an $(N-s)\times s$ matrix \cite{AME-QECC-Zahra}.

On the one hand, the generator matrix greatly suppresses the notation of an OA but also indicates the internal structure of OA. 
In fact, except for linear OAs, there is another class of OAs that might be written in a form of a generator matrix. 
Indeed, in place of the finite field $\GF{d}$, one may consider a finite ring $\mathcal{R}$ of order $|\mathcal{R}|=d$. 
In such a way all linear combinations of rows of the generator matrix forms a \textit{module} over the ring $\mathcal{R}$. 
A practical selection for a ring $\mathcal{R}$ is cyclic group $\mathbb{Z}_d$ or, more generally, simple sum of cyclic groups $\mathbb{Z}_{d_1} \oplus \ldots\oplus \mathbb{Z}_{d_s}$. 
In fact those are the only possibilities for commutative rings $\mathcal{R}$. 

\begin{figure}[ht!]
\centering
\[
{	
\begin{array}{*5{c}}
\tikzmarkk{up}{0}&0&0&\tikzmarkk{down}{0}\\
\Highlightt[first]
\Highlightt[second]
0&1&1&1\\
0&2&2&2\\
0&3&3&3\\
\tikzmarkk{up}{1}&0&1&\tikzmarkk{down}{2}\\
\Highlighttt[first]
1&1&2&3\\
1&2&3&0\\
1&3&0&1\\
\tikzmarkk{up}{2}&0&2&\tikzmarkk{down}{0}\\
\Highlightt[first]
\Highlightt[second]
2&1&3&1\\
2&2&0&2\\
2&3&1&3\\
\tikzmarkk{up}{3}&0&3&\tikzmarkk{down}{2}\\
\Highlighttt[first]
3&1&0&3\\
3&2&1&0\\
3&3&2&1\\
\\
\multicolumn{5}{c}{{\text{OA(16,4,3,1)}}}
  \end{array} 
}
  \;
\xleftarrow[\mathbb{Z}_4]{\text{over}}
  \;
   \left[ 
 \begin{array}{*4{c}}
1&0\\
0&1\\
  \end{array}
  \Big\vert
   \begin{array}{*4{c}}
1&2\\
1&1\\
  \end{array}
\right] 
  \;
\xrightarrow[\GF{4}]{\text{over}}
  \;
{	
\begin{array}{*5{c}}
0&0&0&0\\
0&1&1&1\\
0&2&2&2\\
0&3&3&3\\
1&0&1&2\\
1&1&0&3\\
1&2&3&0\\
1&3&2&1\\
2&0&2&3\\
2&1&3&2\\
2&2&0&1\\
2&3&1&0\\
3&0&3&1\\
3&1&2&0\\
3&2&1&3\\
3&3&0&2\\
\\
\multicolumn{5}{c}{{\text{OA(16,4,3,2)}}}
  \end{array} 
}
\]
\caption{In the middle, a generator matrix of an orthogonal array $\oa{d^2,4,d,k}$ for a prime power dimension $=4d=2^2$. 
Rows of an OA are given by all linear combinations $c_1 \left(1 ,0,1,2 \right)  + c_2 \left(0 ,1,1,1 \right)$ for $c_1,c_2 \in \GF{4}$ and $\mathbb{Z}_4$ respectively. Note a difference in OAs strength, $k=1$ and $k=2$ for $\GF{4}$ and $\mathbb{Z}_4$ respectively. Indeed, pairs of symbols $00$ and $02$ repeats twice on the second and fourth column in $\oa{d^2,4,d,1}$, while combinations $01$ and $03$ do not appear in those two columns. Both OAs are irredundant.}
\label{OA1}
\centering
\end{figure}

As we already have show in \cref{chap1}, OAs might be successfully used to construct $m$-resistant states of $N$-qudits, for which the entanglement of the reduced state of $N-m$ subsystems is fragile for the loss of any additional subsystem. In fact, states constructed in this method exhibit more exceptional property of being $k$-uniform, as we describe below.

For any $\oa{r,N,d,k}$ one may associate a pure quantum state $\ket{\psi}\in \mathcal{H}_d^{\otimes N}$ by reading all rows of OA and creating a superposition of $r$ terms in the computational basis \cite{DiK,OA}, see \cref{fig2AME}. 
Moreover, for IrOAa of a strength $k$, the corresponding quantum state is $k$-uniform \cite{DiK,OA}. 

Consider now a linear OAs over the module $\mathcal{R}$, $|\mathcal{R}|=d$ with the $N\times s$ generator matrix of the form $G=[\Id_s |A]$, where $A=(a_{ij})$ is an $(N-s)\times s$ array with elements from $\mathcal{R}$. The corresponding quantum state might be conveniently presented as follows:
\begin{align}
\label{Eq1}
\ket{\psi}
&:=\frac{1}{\sqrt{d^s}} \sum_{\vec{v} \in \mathcal{R}^{\otimes s}} \ket{G \vec{v} } 
\end{align}
where the multiplication and addition hidden inside an expression $G \vec{v} $ are determined by the structure of a ring $\mathcal{R}$, and vector $\vec{v} \in \mathcal{R}^{\otimes s}$ has exactly $d^s$ elements \cite{AME-QECC-Zahra}. 

\begin{example}
Consider the following generator matrix
\begin{equation}
G=   \left[ 
 \begin{array}{*4{c}}
1&0\\
0&1\\
  \end{array}
  \Big\vert
   \begin{array}{*4{c}}
1&2\\
1&1\\
  \end{array}
\right] ,
\label{Gmat}
\end{equation}
with entries from the ring $ \mathcal{R}=\mathbb{Z}_d$.  \cref{Eq1} associates the quantum state 
\begin{equation}
\ket{\psi}= \frac{1}{d}\sum_{i_1,i_2 =0}^{d-1} 
\ket{i_1} \ket{i_2} 
\ket{i_1 +i_2} \ket{ 2 i_1 +i_2},
\label{state4}
\end{equation}
where addition and multiplication are considered over the ring $\mathbb{Z}_d$. 
For any odd dimension $d$ ($2\nmid d$), the corresponding state $\ket{\psi} \in \mathcal{H}_d^{\otimes 4}$ is an $\ame{4,d}$ state, as the corresponding $\oa{d^2,4,d,2}$ has strength $k=2$, see \cref{fig2AME}. 
\end{example}

\begin{figure}
\centering
\[
   \left[ 
 \begin{array}{*4{c}}
1&0\\
0&1\\
  \end{array}
  \Big\vert
   \begin{array}{*4{c}}
1&2\\
1&1\\
  \end{array}
\right] 
  \;
\xrightarrow[\GF{3}]{\text{over}}
  \;
 \begin{array}{*4{c}}
    0 &
0&\tikzmark{left}{0}&0 \\
0&1&1&1\\
0&2&2&2\\
1&0&1&2\\
1&1&2&0\\
1&2&0&1\\
2&0&2&1\\
2&1&0&2\\
2&    2&1&\tikzmark{right}{0}
\Highlight[first]
  \end{array}
  \;
  \xrightarrow[\ket{\text{AME(4,3)}}\propto]{\substack{\text{corresponding}\\\text{state}}}
  \;
 \begin{array}{*1{l}}
\ket{0000}+ \\
\ket{0111}+ \\
\ket{0222}+\\
\ket{1012}+\\
\ket{1120}+\\
\ket{1201}+\\
\ket{2021}+ \\
\ket{2102}+ \\
\ket{2210}\\
  \end{array}
\]
\caption{The orthogonal array of unity index $\oa{9,4,3,2}$ on the left and repeated in the center. Each subset consisting of two columns contains all possible combinations of symbols. Here, two such subsets are highlighted. The relevant quantum state is obtained as a superposition of states corresponding to consecutive rows of the array.}
\label{fig2AME}
\end{figure}


\section{Different linear structures}
\label{Aarhus14}

For dimension $d=4$, \cref{OA1} presents two OAs corresponding to the ring structures $\mathcal{R}= \GF{4}, \mathbb{Z}_4$. Two quantum states $\ket{\psi}_{\GF{4}}$, and $\ket{\psi}_{\mathbb{Z}_4}$ might be constructed simply by reading consecutive rows of respective OAs. As the corresponding OAs has different strength, quantum states $\ket{\psi}_{\GF{4}}$, and $\ket{\psi}_{\mathbb{Z}_4}$ are 1-, and 2-uniform respectively, see \cref{OA1}. In general, states obtained by the linear action of a Galois field $\GF{d}$ exhibit larger uniformity $k$ then those based on the other ring structures $\mathcal{R}$, $|\mathcal{R}|=d$. 

Once more, consider the generator matrix $G$ defined by \cref{Gmat}. 
In dimension $d=9$, we may consider three different ring structures: $\mathcal{R}= \GF{9}, \mathbb{Z}_9$, and $\mathbb{Z}_3 \oplus \mathbb{Z}_3$. 
The ring $\mathbb{Z}_9$ consists of nine numbers $0,\ldots,8$ with the usual multiplication and addition modulo 9. 
We denote the related quantum state (\ref{state4}) by $\ket{\psi}_{\mathbb{Z}_9}$. 

Elements of $\mathbb{Z}_3 \oplus \mathbb{Z}_3$ are given by pairs of numbers $(a,b)$, where $a,b=0,1,2$, while the multiplication and addition is considered on respective positions
\begin{align*}
(a,b)+(a',b')&=(a +a' ,b+b') ,\\
(a,b)\cdot(a',b')&=(a \cdot a' ,b\cdot b'), 
\end{align*}
modulo 3. 
Note that a pair $(a,b)$ might be canonically associated with the number $3a+b \in \mathbb{Z}_9$, which we further use for linking with a pure quantum state (\ref{state4}) denoted by $\ket{\psi}_{\mathbb{Z}_3 \oplus \mathbb{Z}_3} \in \mathcal{H}_9^{\otimes 4}$ written in the computational basis, $(a,b)\rightarrow 3a+b \rightarrow \ket{3a+b}$. 

The structure of a Galois field $\GF{9}$ might be presented as an addition and multiplication of nine polynomials in $x$ variable $a x +b$, where $a,b =0,1,2$. Addition and multiplication of aforementioned polynomials is considered modulo an irreducible degree 2 polynomial $x^2 +1$. Although addition of polynomials $a x +b$ matches the addition of pairs $(a,b)$, the multiplication structure is different, see \cref{figGalois9}. 
Same as before, each polynomial $a x +b$ might be associated with the natural number $3a+b \in \mathbb{Z}_9$, which defines quantum state $\ket{\psi}_{\GF{9}}$ via \cref{state4} written in the computational basis. 

Unlike as for rings $\mathcal{R}= \GF{4}, \mathbb{Z}_4$, three arrays $\oa{81,4,9,2}$ corresponding to three different ring structures $\mathcal{R}= \GF{9}, \mathbb{Z}_9$ and $\mathbb{Z}_3 \oplus \mathbb{Z}_3$ are all of strength $2$. Therefore, the related quantum states $\ket{\psi}_{\mathbb{Z}_9}$, $\ket{\psi}_{\mathbb{Z}_3 \oplus \mathbb{Z}_3}$, and $\ket{\psi}_{\GF{9}}$, are all AME. In \cref{Aarhus15}, we discuss the problem of equivalence of such states.

\begin{figure}
\begin{center}
\includegraphics[scale=0.9]{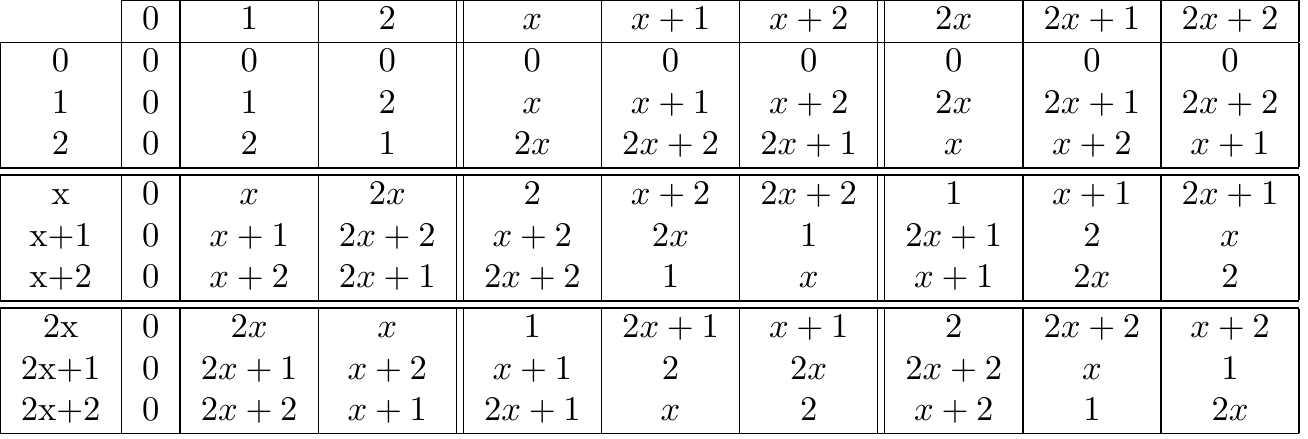}
\end{center}
\caption{Elements of the Galois field $\GF{9}$ are given by nine polynomials $ax+b$, where $a,b=0,1,2$. 
Their multiplication is considered modulo an irreducible polynomial $x^2 +1$.
Exact structure of a multiplication is presented in the table, while the addition structure is the same as for the ring $\mathbb{Z}_3 \oplus \mathbb{Z}_3$. 
Note that there are no zero divisors in $\GF{9}$, i.e pairs of non-zero elements which product is zero. 
Each polynomial $ax+b$ might be later associated with a number $3a+b=0,\ldots, 8$.} 
\label{figGalois9}
\end{figure}



\section{Conclusions}

In this chapter, I recall the well-established notion of Absolutely Maximally Entangled and $k$-uniform states (\cref{Aarhus13}) and their correspondence to classical combinatorial designs. I focus my attention on the different linear structures of classical designs that affect the structure of the related AME and k-uniform states. This chapter does not contain any significant contribution by the author and is intended to introduce some notions used in the consecutive chapters.

\chapter{Thirty-six entangled officers of Euler}
\label{ch4}

In this chapter, I discuss a quantum variant of the famous problem of $36$ officers of Euler. While the classical problem of Euler is known to have no solutions, I present an analytical form of an AME$(4,6)$ state, which might be seen as a quantum solution to the Euler's problem. Moreover, I show a coarse-grained combinatorial structure behind constructed state. Considerations of such a coarse-grained combinatorial structures might lead to a successful approach for constructing genuinely entangled states beyond the stabilizer approach. An extension of some parts of this chapter, in which the author's contributuion was not substantial can be found in the joint work \cite{RatherBur20}. 

\section{Thirty-six officers of Euler}
\label{chap51}

In 1779, Euler examined the now-famous officer problem \cite{Euler36}: 
``Six different regiments have six officers, each one belonging to different ranks. 
Can these 36 officers be arranged in a square formation so that each row and column contains one officer of each rank and one of each regiment?'' 
As Euler observed, the possibility of such an arrangement is equivalent to existence of \textit{Graeco-Latin squares} of order 6. 
Single Latin square of order $d$ is filled with $d$ copies of $d$ symbols arranged in a square in such a way that no row or column of the square contains the same symbol twice. 
Two orthogonal Latin squares (OLS), also called Graeco-Latin square, is an arrangement of ordered pairs of symbols, for instance, one Greek character and one Latin, into $d \times d$ square. 
Each symbol appears exactly once in each row and column, while each pair of symbols appears exactly once in the entire OLS, see \cref{OLS3}. 

\begin{figure}[ht!]
	\centering
\begin{equation} \
\begin{array}{|c|c|c|} \hline 
{ A\alpha}  & B\beta& {C\gamma}\\ \hline
{ C\beta} & A\gamma  & { B\alpha}\\ \hline
{B\gamma}    & C\alpha  &{A \beta} \\ \hline
\end{array} 
\ = \ 
\begin{array}{|c|c|c|} \hline 
{ A \text{\spade}}  & K \text{\club} &  {\color{red}  Q \text{\diamond} }  \\  \hline
{ Q\text{\club}}  & {\color{red} A \text{\diamond}} &  K \text{\spade}    \\  \hline
{ \color{red} K \text{\diamond} }    & Q \text{\spade} &  A \text{\club}   \\ \hline
\end{array} 
\ = \ 
\begin{array}{|c|c|c|} \hline 
{ 0,0}  & 2,1&  1,2  \\  \hline
1,1  & 0,2 &  2,0 \\  \hline
2,2   & 1,0&0,1  \\ \hline
\end{array} 
\nonumber
 \end{equation}
\caption{An example of Graeco-Latin square of order $d = 3$, relevant to AME(4,3) state. 
In the middle, Greek and Latin letters are replaced by 
ranks and suits of cards, on the right by pair of numbers.}
\label{OLS3}
\end{figure}

In \cref{chapter3} we discussed a connection between combinatorial designs, such as OLS, and multipartite entangled states, and error-correcting codes. In particular, to any OLS, one may associate a quantum state $\ket{\psi} \in \mathcal{H}_d^{\otimes 4}$
\begin{equation}
\label{state}
\ket{\psi}= 
\sum_{i,j,k,\ell =0}^{d-1}
T_{ijk\ell} \dfrac{1}{d}
\ket{i}\ket{j} \ket{k}\ket{\ell} ,
\end{equation}
where $T_{ijk\ell}=1$ if the pair $(k,\ell)$ is an entry in $i$-th row and $j$-th column, while $T_{ijk\ell}=0$ otherwise. 
Conditions imposed on OLS translates onto fact that for the following bipartition of four indices into pairs: $ij|k \ell,\; ik|j \ell,\; i \ell|jk$, the corresponding matrices are permutations, i.e. $T_{ij}^{k\ell}, T_{ik}^{j\ell}, T_{i\ell}^{jk}$ are permutation matrices. 

Unitary matrices are a natural generalization of permutation matrices. 
A tensor $T_{ijk\ell}$ for which each of matrices $T_{ij}^{k\ell}, T_{ik}^{j\ell}, T_{i\ell}^{jk}$ is a unitary matrix is called \textit{perfect}. Furthermore, such a square matrix $U=T_{ij}^{k\ell}$ is called 2-unitary, as the three matrices: $U=T_{ij}^{k\ell}$, and related \textit{reshuffled} $\mathcal{U}^{\R}=T_{ik}^{j\ell}$, and \textit{partialy transposed} $(\mathcal{U}^{\R})^{\PT}=T_{i\ell}^{jk}$ matrices are all unitary. 
Perfect tensor provides an \textit{isometry} between any pair of its indices. 
The partial trace of the state $\ket{\psi}$ related to a perfect tensor is maximally mixed for any bipartition of the system, i.e. $\rho_{S}\propto \Id$ for $|S|=2$, and hence is related to an AME($4,d$) state \cite{PhysRevA.69.052330}.

Two OLS of order $d$ exist for $d=3,4,5$ and any natural number $d\ge 7$ \cite{JCD01}. 
On the other hand, OLS of order $2$ and $6$ do not exist, which can be easily observed for order $2$, and what was proven by an exhaustive case study for order $6$ \cite{GastonTarry}. 
The only local dimension $d$ for which the existence of a quantum version of OLS was not decided is six \cite{HIGUCHI2000213,AME_IQOQI,Table_AME,yu2020complete}. 

\section{AME state of four quhex}

In \cite{RatherBur20}, we present a solution for the quantum version of the Euler's problem of 36 officers. 
More specifically, we provide there an analytical form of an AME($4,6$) state, which can be associated with an example of orthogonal quantum Latin squares  of order six \cite{MV16,ComDesi,Ri20,RatherBur20}. In order to find a $2$-unitary matrix of order $36$, we used an iterative numerical technique based on nonlinear maps in the space of unitary matrices $\mathbb{U}(d^2)$ \cite{SAA2020}. One of key challenges was to find an appropriate seed matrix, which generates a numerical $2$-unitary matrix of order $d^2=36$. Details of the procedure and an example of a seed that leads to the $2$-unitary solution might be find in \cite{RatherBur20}

The author's main contribution to the solution of this problem, was finding an analytical form of the presented state, obtained with relation (\ref{state}) from a 2-unitary matrix $U_{36}$, which determines the tensor $T_{ijk \ell}$. 
Any 2-unitary matrix remains 2-unitary under a multiplication by local unitary operators. 
Using this freedom, we applied a search algorithm over the group 
$\mathbb{U}(6) \otimes \mathbb{U}(6)$ of local unitary operations, 
to orthogonalize certain rows and columns in a given numerical 2-unitary matrix $U$ and its rearrangements $U^{\R}$, $(U^{\R})^{\PT}$. 
The particular choice of the orthogonality relations corresponds to the block structure of the eventually obtained analytical solution. 
These tools can be generalized to construct multi-unitary operators and corresponding AME states in other local dimensions and number of parties, and can potentially yield states that are not created by presently known techniques.

In order to present our solution for the problem of 36 entangled officers of Euler, we express
coefficients of a quantum  state $\ket{\psi} \in \mathcal{H}_6^{\otimes 4}$ (four qudits state) via a four--index tensor $T_{ijk\ell}$ as it is presented on Eq. \ref{state}. 
Non-vanishing terms of 
$T_{ijk\ell}$ might be conveniently written in form of a table, see Fig. \ref{fig2}. 
Provided construction is based on root of unity of order 20, denoted by $\omega=\exp(i \pi/10)$. 
There are three non-zero amplitudes:
\begin{align*}
a=&[\sqrt{2}(\omega+\overline{\omega})]^{-1}=[5+\sqrt{5}]^{-1/2}\simeq 0.3717 \\
b=&[\sqrt{2}(\omega^3+\overline{\omega}^{3})]^{-1}=[(5+\sqrt{5})/20]^{1/2} \simeq 0.6015\\
c=&1/\sqrt{2}\simeq 0.7071
\end{align*}
which might be expressed in terms of $\omega$. 
Two relations $a^2+b^2= c^2=1/2$ and $b/a=\varphi=(1+\sqrt{5})/2$ the golden ratio, determine uniquely amplitudes appearing in the solution and explain why the constructed AME state deserves to be called the {\sl golden} AME {\sl state}. 
Checking the property of being an AME state comes down to verification of several equations involving roots of unity of order 20, which I discuss in a detailed way in \cref{AMEcheckApp}.

\begin{figure}[ht!]
\centering
\includegraphics[scale=0.68]{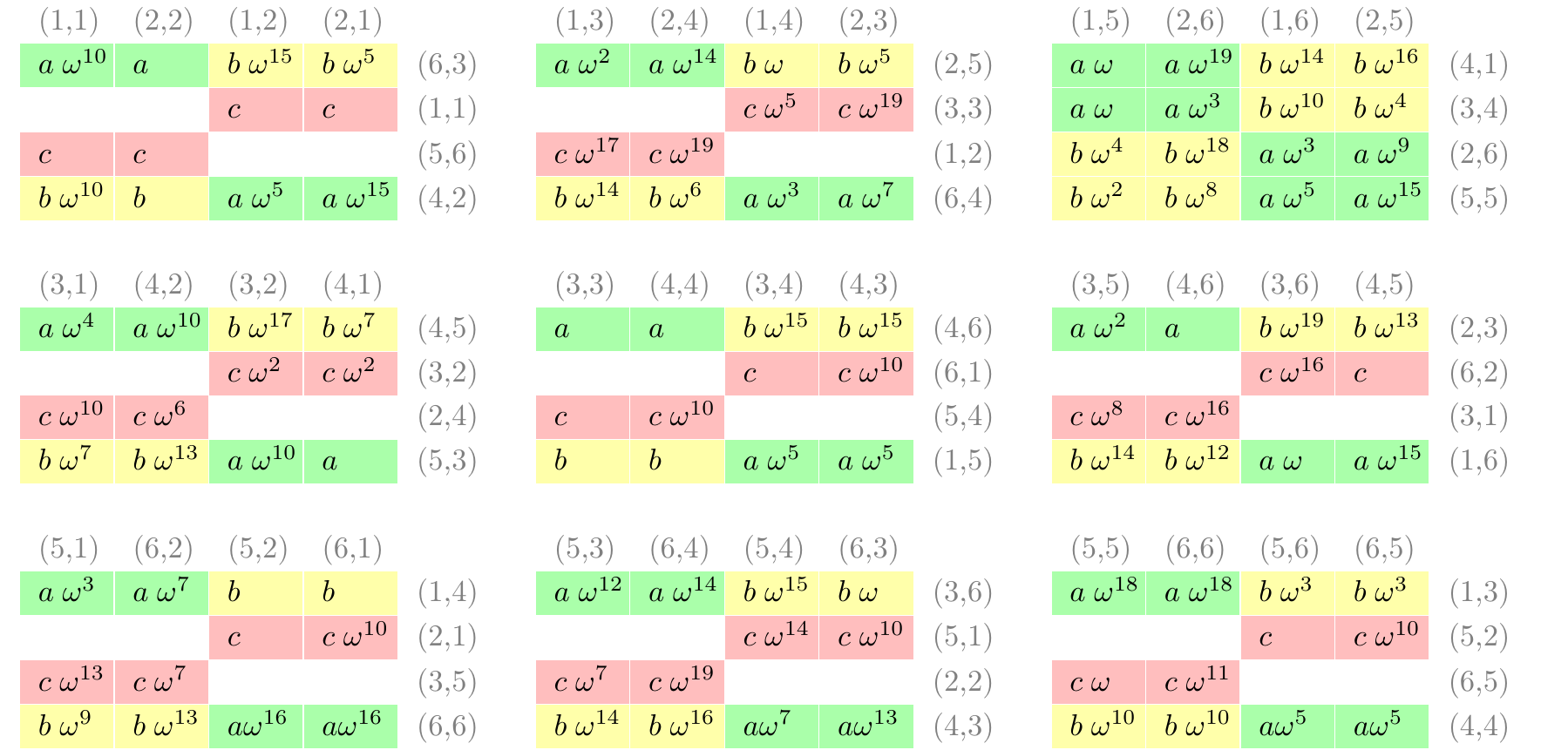}
\caption{Non-vanishing elements 
of a perfect tensor $T_{ijk\ell}$ related to the AME(4,6) state. A pair of indices $(i,j) $ are indicated in rows, $(k,l) $ in columns respectively. 
Treating each coordinates $(i,j)$ as a position (row and column) of Euler's officer, its rank and regiment are in superposition of two or four canonical ranks and regiments. 
}
\label{fig2}
\end{figure} 

\section{Structure of the AME(4,6) state}

Recall that a classical OLS corresponds to a $2$-unitary permutation matrix. 
Since there is no solution to the original Euler's problem, the $2$-unitary permutation matrix of size 36 does not exist. 
Nonetheless, the $2$-unitary matrix presented on \cref{fig2} has the structure of nine $4\times 4$ blocks. 
Moreover, the block structure is preserved under reshuffling and the partial transpose of the matrix. 
In that sense, we found the \textit{invariant} block structure in the original $36\times 36$ with respect to reshuffling and partial transpose operations. 
In fact, the problem of finding 36 entangled officers of Euler splits into two sub-problems: finding such a block invariant structure at first, and then select adequate non-zero elements inside the chosen structure. 
Such an approach indicates the plausible direction of search for other AME and $k$-uniform states. 
In fact, the notion of such an invariant structure might be formalized via certain combinatorial designs. 
Grouping symbols of indices $k$ and $\ell$ in the presented perfect tensor $T_{ijk \ell}$ in pairs: $1,2\rightarrow A /\alpha$, and $3,4\rightarrow B / \beta$, and $5,6\rightarrow C/ \gamma$ for $k/ \ell $ respectively, results in coarse-grained OLS, which reveals the described block structure, see \cref{fig3}.

\begin{figure}[ht!]
	\centering
\begin{equation} \
\begin{array}{|c|c|c|c|c|c|} \hline 
{ A\alpha}  & A\beta& {C\gamma}&
{ C\alpha}  & B\beta& {B\gamma}\\ \hline
{ C\alpha}  & C\beta& {B\gamma}&
{ B\alpha}  & A\beta& {A\gamma}\\ \hline
{ B\gamma} & B\alpha  & { A\beta}&
{ A\gamma} & C\alpha  & { C\beta}\\ \hline
{ A\gamma} & A\alpha  & { C\beta}&
{ C\gamma} & B\alpha  & { B\beta}\\ \hline
{C\beta}    & C\gamma  &{B \alpha} &
{B\beta}    & A\gamma  &{A \alpha} \\ \hline
{B\beta}    & B\gamma  &{A \alpha} &
{A\beta}    & C\gamma  &{C \alpha} \\ \hline
\end{array} 
\nonumber
 \end{equation}
\caption{A coarse-grained OLS of order 6, which reveals the block structure of a perfect tensor $T_{ijk \ell}$. 
Indices of non-vanishing elements of the tensor $T_{ijk \ell}$ are presented: $i$ in row, $j$ in column, while a pair of coarse-grained indices $k,\ell$ in relevant entry. 
Each pair of symbols repeats exactly four times on the grid. 
Moreover, each symbol on each position repeats exactly twice in each row and column. 
}
\label{fig3}
\end{figure} 

Furthermore, presented AME($4,6$) state might be written in the following from, which reveals its block structure
\begin{equation}
\ket{\Psi}= \sum_{i,j=1}^6 \ket{i,j} \otimes \ket{\psi_{i,j}},
\end{equation}
where each of states $\ket{\psi_{i,j}} \in \mathcal{H}_{36} \cong \mathcal{H}_{3} \otimes \mathcal{H}_{3}  \otimes \mathcal{H}_{2\cdot 2} $ has the following form:
\begin{equation}
\ket{\psi_{i,j}}=
\Big(
\underbrace{\ket{2 \Big\lfloor  \dfrac{j}{2}\Big\rfloor +2 i+1}}_{\mathcal{H}_{3}} 
\otimes
\underbrace{\ket{j + 2\Big\lfloor  \dfrac{i}{2}\Big\rfloor +1}}_{\mathcal{H}_{3}} 
\Big)
\otimes
\id \otimes U_{ij} 
\Big(
\underbrace{\dfrac{1}{\sqrt{2}}\ket{00} +\dfrac{1}{\sqrt{2}}\ket{11}}_{\mathcal{H}_{2\cdot 2}} 
\Big)
\end{equation}
where $i,j=1,\ldots,6$ and $U_{ij} $ are unitary matrices of order 2, which form might be read from \cref{fig44}. Note that pairs of indices 
\begin{equation}
\Big(
2 \Big\lfloor  \dfrac{j}{2}\Big\rfloor +2 i+1 
\;,\;
j + 2\Big\lfloor  \dfrac{i}{2}\Big\rfloor +1
\Big)
\end{equation}
are in accordance with the entries of $6\times 6 $ coarse-grained OLS presented on \cref{fig3} after associating Greek character and one Latin characters with numbers: $A, \alpha \rightarrow 0;  \;
B, \beta \rightarrow 1; \;
C,\gamma \rightarrow 2$. 

\newsavebox{\smlmat}
\savebox{\smlmat}{
$ U_{12} =\left( 
\begin{smallmatrix} 0 & \gamma \\ 
\gamma \omega^{10} & 0  
\end{smallmatrix}
\right) $.
}

\begin{figure}
\begin{center}
\includegraphics[scale=0.89]{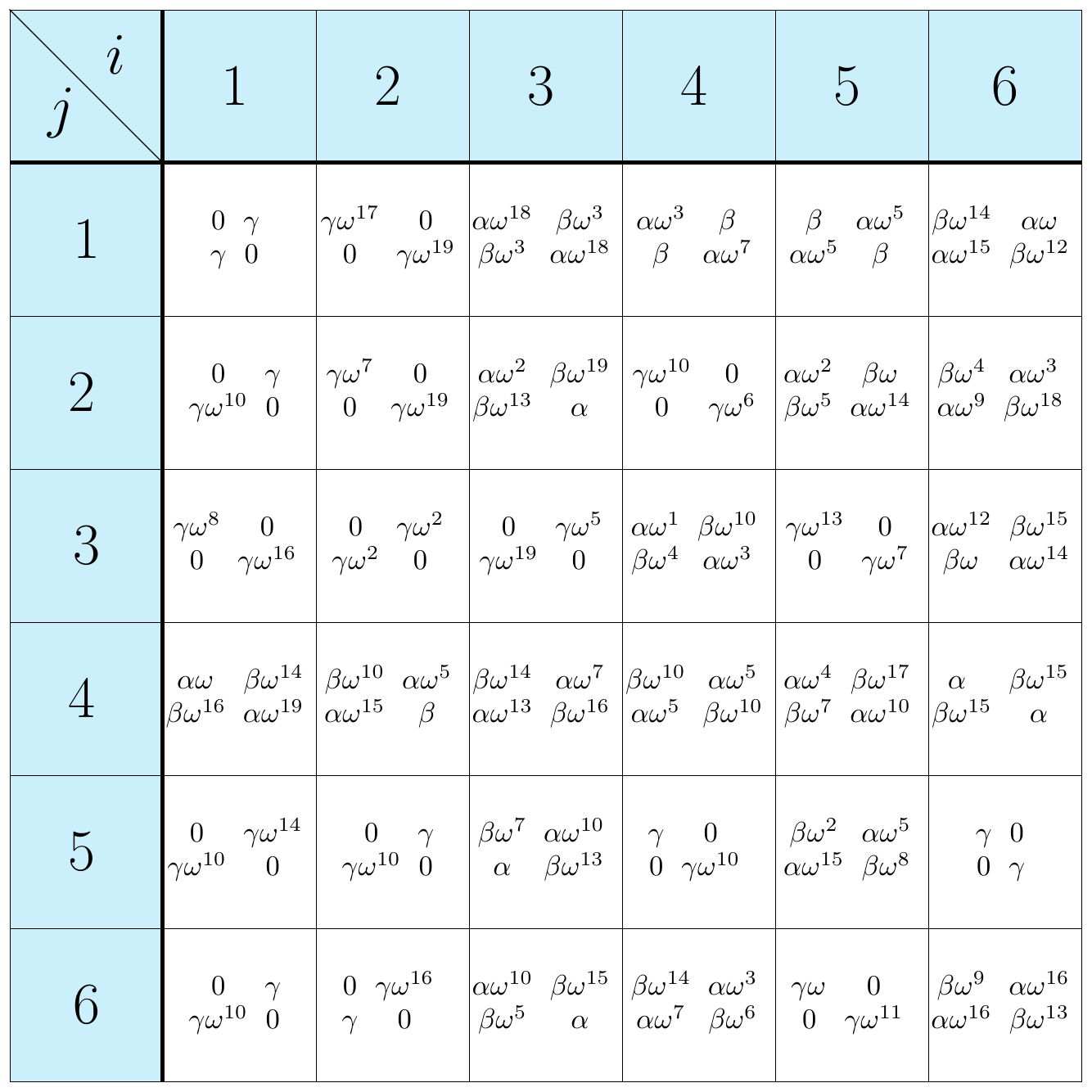}
\end{center}
\caption{An exact form of $2\times 2$ unitary matrices $U_{ij}$ for $i,j=1,\ldots,6$. For the sake of visibility, the correct matrix notation has been replaced only by indication of matrix  entries. For example, for $ i = 1, j = 2 $ the related matrix reads 
\usebox{\smlmat}. 
}
\label{fig44}
\end{figure}


Most of the examples of quantum error correction codes belong to the class of, so-called, additive (stabilizer) states  \cite{Cross_2009,AlsinaStab}. In particular all hitherto known AME states, are either stabilizer states \cite{AlsinaStab}, or might be derived from the stabilizer construction \cite{BurchardtRaissi20}. 
We examined all stabilizer sets of four quhex in their standard form and did not find an AME$(4,6)$ state, hence the presented state is not a stabilizer state. 

Nonadditive quantum codes are in general more difficult to construct, however, in many cases, they outperform the stabilizer codes \cite{NonadditiveCodes}. 
Thus far, the stabilizer approach practically contained the combinatorial approach to constructing AME and $k$-uniform states. 
As we demonstrated, the consideration of coarse-grained combinatorial structures might be successful in constructing genuinely entangled states and have advantages over the stabilizer approach. 
It is tempting to believe that the quantum design presented here
 will trigger further research on quantum combinatorics.

Even though OLS of order six does not exist, we present a coarse-grained OLS of that order, which structure stands behind the provided AME state. 
This sheds some light on how to construct quantum nonadditive AME states, and more generally QECCs, when the stabilizer approach fails.


\section{Verification of 2-unitarity}
\label{AMEcheckApp}

In this section we show that each of three matrices $T_{ij}^{k\ell}, T_{ik}^{j\ell}, T_{i\ell}^{jk}$ related to the tensor $T_{ijk\ell}$ is unitary. In order to simplify the notation, we introduce a matrix $\mathcal{U}$ of size $36$ with entries, determined by the tensor $T_{i\ell}^{jk}$, in the following way $\mathcal{U}_{p,s}=T_{ijkl}$, with $p=j+6(i-1)$ and $s=\ell(k-1)$, i.e. the matrix corresponding to the flattening $T_{i\ell}^{jk}$. Similarly, by $\mathcal{U}^{\R}$ and $(\mathcal{U}^{\R})^{\PT}$, we denote matrices corresponding to the flattening $T_{ik}^{j\ell}$ and $ T_{i\ell}^{jk}$ respectively. Operations $\R$ and $\PT$ are known as \textit{reshuffling} and the \textit{partial transpose} respectively, compare with \cref{chap51}. 

To show the tensor $T_{ijk\ell}$ is perfect, we shall verify that three related matrices: $\mathcal{U}$, $\mathcal{U}^{\R}$ and $(\mathcal{U}^{\R})^{\PT}$ are unitary. 
In fact, each of those matrices has the block structure with nine $4\times 4$ sub-matrices. 
Hence our task simplifies to verification that constituent blocks are unitary matrices. 

Consider the matrix $\mathcal{U}$. 
Interestingly, except for one block component in the matrix $\mathcal{U}$,
all eight remaining blocks are equivalent (up to a multiplication of rows and columns by adequate phases) to the following $4\times 4$ unitary matrix:
\[V=
\begin{bmatrix}
a & a & b & b\\
0 & 0 & c & -c\\
c& -c & 0 & 0\\
b & b & -a & -a
\end{bmatrix}.
\]
Orthogonality between rows in the matrix above might be presented as pairs of antipodal points on the complex plane, for example, the orthogonality between the first two rows reads
\begin{equation}
\label{constell}
bc \big(1-1\big)=0.
\end{equation} 
The right top corner in \cref{fig6} presents the exceptional block of matrix $\mathcal{U}$. 
Six orthogonality relations between rows read
\begin{align}
a^2\big(\omega^{8}  +\omega^{-8}\big) + b^2 \big(\omega^{4} +\omega^{-4}\big)&=0, \\
a b \big( 1  +\omega^{2} +  \omega^{-8} -1\big)&=0,  \\ \nonumber
a b \big( \omega^{-2}  +\omega^{2} +  \omega^{-8} + \omega^{8}\big)&=0,
\end{align}
up to a phase factor, where $\omega=\exp(i \pi/10)$. 
Each equation might be presented as a \textit{unitarity rectangle} - a constellation of four points on the complex plane which sum up to zero, as it is shown on Fig.~\ref{fig6}. 
Observe that the second and third equations above are both related to two pairs of antipodal points on the complex plane. 
Geometric interpretation of the amplitudes $a$ and $b$
is shown in Fig.~\ref{penta}.

\begin{figure}[ht!]
\includegraphics[scale=0.85]{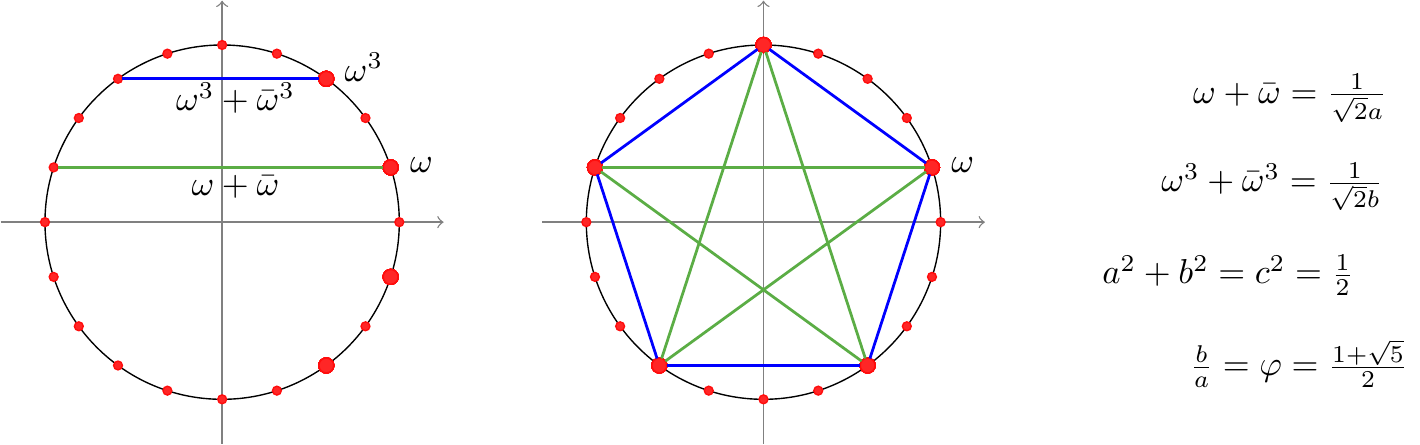}
\caption{Geometric interpretation of the three amplitudes $a,b$ and $c$ appearing in  the AME($4,6$) state with use of a regular pentagon. Amplitudes $a,b$ contain the phase $\omega=\exp(i\pi/10)$ and are related to the golden ratio $\varphi = \frac{1+\sqrt{5}}{2}$. }
\label{penta}
\end{figure}

\begin{figure}[ht!]
\centering
\includegraphics[scale=0.35]{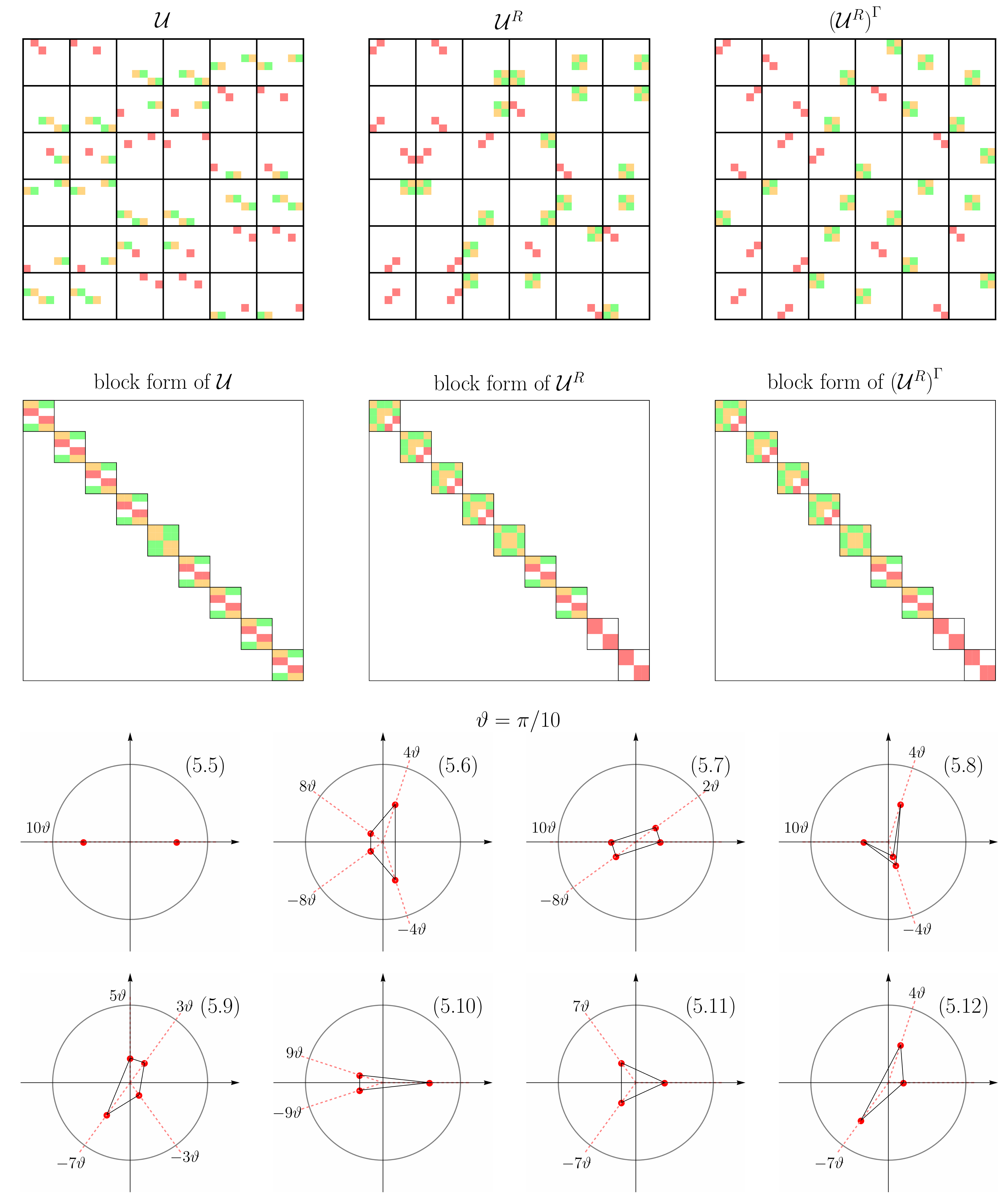}
\caption{Non-vanishing elements of three unitary matrices $\mathcal{U}$, $\mathcal{U}^{\R}$ and  $(\mathcal{U}^{\R})^{\PT}$ of order 36 are presented on top. 
The modulus of a non-vanishing element is represented by the intensity of the background color, where green, orange and red are related to constants $a$, $b$, and $c$ respectively.  
Each matrix has the structure of nine $4\times 4$ blocks, the sub-structure within the blocks is different. 
Orthogonality relations between pairs of rows in block matrices may be presented as a constellation of points on the complex plane which sum up to zero. 
Constellations related to Eqs.~(\ref{constell}-\ref{constell2}) are indicated. 
}
\label{fig6}
\end{figure}

Although both matrices $\mathcal{U}^{\R}$ and $ (\mathcal{U}^{\R})^{\PT}$ enjoy the structure of nine $4\times 4$ blocks,  similar to $\mathcal{U}$,
the particular arrangement inside their blocks is significantly different from the $\mathcal{U}$ matrix. 
Blocks in $\mathcal{U}^{\R}$ and $(\mathcal{U}^{\R})^{\PT}$ are of four distinct types up to multiplication of their rows and columns by phase factors, see Fig.~\ref{fig6}. 
Orthogonality relations between rows of both matrices reflect their complex structure. 
In particular, we distinguished five additional orthogonality relations given by the following equations:
\begin{align}
a^2 \omega^{4} +ab \big(\omega^{10}  +\omega^{-4}\big) + b^2 \omega^{-4} &=0, \\
a^2 \omega^{-3} +ab \big(\omega^{5}  +\omega^{3}\big) + b^2 \omega^{-7} &=0, \\
ab \big(\omega^{-4}  +\omega^{-6}\big) +  bc \omega^{5}  &=0, \\
ab \big(\omega^{-8}  +\omega^{-2}\big) +  ac \omega^{5}  &=0, \\
\label{constell2}
a^2 + b^2 \omega^{4} + bc\omega^{-7} &=0.
\end{align}
Related constellations are presented on Fig.~\ref{fig6}. 
The above-listed equalities provide orthogonality between rows in the three matrices $\mathcal{U}$, $\mathcal{U}^{\R}$ and $(\mathcal{U}^{\R})^{\PT}$, hence they imply that the matrix $\mathcal{U}$ is $2$-unitary.

Interestingly, in the presented solution each officer is entangled with at most three other officers (out of remaining $35$). This  implies the  matrix $\mathcal{U}$ is sparse. 
  
Amplitudes $a,b,c$ which appear in the presented construction might be defined as the unique solution of the following three algebraic equations: $a^2+b^2 =c^2 =1/2$ and $b/a = \varphi = (1+\sqrt{5} )/2$, see Fig.~\ref{penta}. 
Notice the similarities between algebraic equations which lead to values $a,b,c$, and the algebraic equations which lead to the amplitudes in an another AME state, in the heterogenous $2\times 3\times 3\times 3 $ system presented in \cite{Huber_2018}. 

The phases of the coefficients shown in \cref{fig2}, being multiples of $\omega=\exp(i \pi/10)$,  are chosen in such a way that all  $36$ quantum states $\ket{\psi_{ij}}$, each represented by a single row of the matrix $\mathcal{U}$, are equivalent to the standard, two-qubit Bell state. This is easy to see, for any state formed by two coefficients of moduli, $|c|=1/\sqrt{2}$, as states $\ket{\psi_{11}}$ or $\ket{\psi_{56}}$ represented in the second and the third row in
the upper left block in \cref{fig2}. One can show that this property holds also for other states. For example, the state $\ket{\psi_{63}}$ corresponding to the first line of
the aforementioned block can be written in the product basis as
\[
\ket{\psi_{63}} =
a\omega^{10}\ket{{11}}+
b\omega^{15}\ket{{12}}+
b\omega^{5}\ket{{21}}+
a\ket{{22}}.
\]
Thus the partial trace reads, 
${\rm Tr}_{\rm B} \ket{\psi_{63}}\bra{\psi_{63}}=
{\rm diag}(a^2+b^2, a^2+b^2)={\mathbb I}/2$, 
which proves that  $\ket{\psi_{63}}$ is locally equivalent to the maximally entangled Bell state. A similar reasoning works for all other states consisting of four terms and represented in \cref{fig2} by green and yellow elements. In such a way, all $36$ states, corresponding to $36$ entangled officers of Euler, can be considered as  maximally entangled, two qubit states. 

\section{Conclusions}

In this chapter I presented the famous combinatorial problem of 36 officers was posed by Euler, and its quantum version. I present an analytical form of an AME$(4,6)$ state of four subsystems with six levels each, which might be seen as a quantum solution to Euler's problem of finding two orthogonal Latin squares of order six. The existence of such a state is equivalent to the existence of a $2$-unitary matrix of order 36 or a perfect tensor with four indices, each running from one to six. This result allows us to optimally encode a single quhex into a three quhex state. Furthermore, I show a coarse-grained combinatorial structure behind the constructed state. Further analysis of such coarse-grained combinatorial structures might potentially lead to the construction of other genuinely entangled states beyond the stabilizer approach.

\chapter{Classification of absolutely maximally entangled states}
\label{chap5}

In this chapter of the thesis, I present techniques for verifying whether two AME states are equivalent with respect to Stochastic Local Operations and Classical Communication (SLOCC). 
I falsify the conjecture that for a given multipartite quantum system all AME states are SLOCC-equivalent. 
I also show that the existence of AME states with minimal support of 6 or more particles results in the existence of infinitely many such non-SLOCC-equivalent states. 
Moreover, I present AME(5,$d$) states which are not SLOCC-equivalent to the existing AME states with minimal support. Proofs of the statements presented in this Chapter, as well as an extension of some parts, in which the author’s part was not substantial can be found in the joint work 
\cite{BurchardtRaissi20}.

\section{Local equivalence of AME and k-uniform states}
\label{LUSLOCC}

The initial motivation addressed in this chapter is the question whether different constructions of AME states are equivalent by any local transformation. Note, that for AME states all reduced density matrices are maximally mixed. Therefore, the classical method for verification of local equivalence, which is  comparison of Schmidt rank and coefficients, will obviously fail \cite{RanksReducedMatrices}. 

Contrary to the AME states, it was already shown that some $k$-uniform states are not locally equivalent \cite{raissi2019new}, this result was indeed based on a comparison of Schmidt ranks of reduced density matrices. Nevertheless, the aforementioned rank argument is never conclusive for SLOCC-verification of two $k$-uniform states of minimal support. 

In this Chapter, we tackle those two cases, i.e $k$-uniform states of the minimal support and AME states, for which the bound on the number $k$ is saturated, and provide general techniques of SLOCC-equivalence verification between such states. 

We begin with the formal definition of SLOCC- and LU-equivalence. Two $N$-qudit states $\Ket{\psi}$ and $\Ket{\phi}$ are LU-equivalent if one can be transformed into another by local unitary operators, i.e
\begin{equation}
\Ket{\phi} = U_1 \otimes \cdots \otimes U_N \Ket{\psi}.
\end{equation}
Similarly, two states $\Ket{\psi}$ and $\Ket{\phi}$ are SLOCC-equivalent if and only if there exists a \emph{local invertible} operator connecting those states \cite{ThreeQub}:
\begin{equation}
\label{AArhusNoc}
\Ket{\phi} = O_1 \otimes \cdots \otimes O_N \Ket{\psi}.
\end{equation}
Since LU- and SLOCC-equivalences are equivalence relations, the state space might be naturally partitioned into \emph{LU classes} and \emph{SLOCC classes} respectively. As a consequence of Kempf-Ness theorem \cite{KempfNess}, two AME states, or more generally $k$-uniform states are SLOCC-equivalent if and only if they are LU-equivalent \cite{GourWallach}.

\section{Local equivalences, case 2k<N}
\label{31}
We begin with introducing the notion of a \textit{support} of a state and specific class of local unitary operations, namely \textit{monomial matrices}. Both are essential for the further investigations of the classification problem of $k$-uniform states.

The \textit{support} of a state $\ket{\psi}$ of $N$ subsystems with $d$ levels each is the the number of non-zero coefficients in the computational basis. Note, that the support of $k$-uniform state is at least $d^{k}$, since the partial trace over $N-k$ particles is an identity matrix $\Id_{d^{k}}$. The $k$-uniform state with support achieving its lower bound $d^{k}$ is called a \textit{minimal support} state. Note that examples of AME and $1$-uniform states presented in \cref{ex1GHZ,ex2}, \cref{chapter3} are all of the minimal support.

\begin{definition}
A unitary matrix $M$ is called \textit{monomial} if one of the following holds:
\begin{enumerate}
\item $M$ has exactly one nonzero element in each row and each column,
\item $M$ is a product of a permutation and an invertible diagonal matrix,
\item $M$ does not change the support of any quantum state.
\end{enumerate}
The local monomial operation of the multipartite quantum states are denoted as $LM$-eoperations.
\end{definition}

It is straightforward to verify that conditions 1-3 are equivalent. 
Obviously, any LM-operation provides the LU-equivalence between two states of minimal support.  
Interestingly, the reverse statement holds true for $k$-uniform states for which $2k<N$. This simplifies a search for LU-equivalence between two $k$-uniform states of minimal support to the search within the $LM$ class. In \cite{BurchardtRaissi20}, I showed the following.

\begin{proposition}
\label{prop1}
For $2k<N$, each LU- or SLOCC-equivalency between two $k$-uniform states of minimal support is in fact LM-equivalency.
\end{proposition}

\begin{corollary}
\label{coro1}
For $2k<N$, two $k$-uniform states of minimal support are LU- or SLOCC-equivalent if and only if they are LM-equivalent.
\end{corollary} 

The above result turned to be useful in a more refined problem of SLOCC-verification between two AME states, with not necessary minimal support. Below, we present examples of two families of AME($5,d$) states.

\begin{example}
\label{ex5}
The following state
\[
\ket{\text{AME'(5,d)}} =\dfrac{1}{d} \sum_{i,j=0}^{d-1} \ket{i,j,i+j,2i+j,3i+j} 
\]
is an AME state with minimal support for all prime dimensions $d\geq 5$ \cite{AME-QECC-Zahra}. 
\end{example}

It is known that AME states with minimal support of five or six qubits do not exist \cite{Huber_2018}. Nevertheless, the construction of AME($5,2$) state with non-minimal support is known \cite{Rains1999NonbinaryQC,ComDesi}.

\begin{example}
\label{ex3}
The following state
\begin{equation*}
\ket{\text{AME(5,d)}} =
\dfrac{1}{\sqrt{d^3}}
\sum_{i,j,\ell=0}^{d-1} \omega^{(3i+j) \ell} 
\ket{i,j,i+j,2i+j +\ell,\ell} ,
\end{equation*}
where $\omega$ is $d$th root of unity, is $\ket{\text{AME(5,d)}} $ state for any integer number $d\geq 2$ \cite{ComDesi,Rains1999NonbinaryQC}. 
\end{example}

In \cite{BurchardtRaissi20}, I showed that states presented in \cref{ex3} and \cref{ex5} are not locally equivalent. In such a way, I falsified the conjecture that for a given multipartite quantum system all AME states are SLOCC-equivalent.

\begin{proposition}
\label{AME55}
For any prime local dimension $d$, two families of AME(5,d) states: $\ket{\text{AME(5,d)}}$ and $\ket{\text{AME'(5,d)}}$
presented in \cref{ex3} and \cref{ex5} respectively, are not SLOCC-/LU-equivalent. 
\end{proposition}

\section{Local equivalences, case 2k=N}
\label{32}

Interestingly, for AME($2k,d$) states of even number of parties $N=2k$, the statement of \cref{prop1} does not hold anymore. 
We illustrate it on the example below.

\begin{example}
Matrix proportional to the tensor product of the Fourier transform $F_3$:
\begin{equation}
\label{F(3)}
\big(F_3\big)^{\otimes 4}  =\dfrac{1}{9}
\begin{pmatrix} 
1 & 1 & 1 \\ 
1 & \omega & \overline{\omega} \\ 
1& \overline{\omega} & \omega
\end{pmatrix}^{\otimes 4} 
\end{equation}
provides an automorphism of AME($4,3$) state presented in \cref{ex2} in \cref{chapter3}. 
\end{example}

One can suspect that Fourier matrices are not the only non-monomial matrices which may provid the LU-equivalence between AME($2k,d$) states of minimal support. This turned out not to be true, if we consider AME states in composed dimension \cite{BurchardtRaissi20}. For sufficiently small dimensions, however, I derived some further conclusions on the form of matrices which can provide LU-equivalence between AME states. 
We define the class of matrices beyond the monomial class.

\begin{definition}
Let $d$, $q$ be arbitrary positive integers. A \emph{Butson-type} complex Hadamard matrix of order $d$ and complexity $q$ is a unitary matrix in which each entry is a complex $q$th root of unity scaled by the factor $1/\sqrt{d}$. 
The set of Butson-type matrices is denoted by $\text{BH}(d,q)$.
\end{definition}

Note that, in literature, \cite{Butson,Tadej}, Butson-type matrices are usually defined without scaling factor $1/\sqrt{d}$. In such a way, they are not unitary, but proportional to the unitary matrices. In particular, the class of matrices $\text{BH}(d,2)$ is simply the class of real Hadamard matrices of order $d$. In our analysis, we focus on matrices of the type $\text{BH}(d,d)$. Those matrices together with the monomial matrices provide all possible LU-equivalences of AME states. Note that the Fourier matrices belong to this class of matrices, $F_d \in \text{BH}(d,d)$.

\begin{proposition}
\label{prop1=}
Consider two AME($2k$,$d$) states of minimal support, where $k$ and $d$ are sufficiently small (see \cref{small}). 
Each LU-equivalency between them is of one of the following forms:
\begin{enumerate}
\item tensor product of Butson-type matrices $ B_i \in \text{BH}(d,d)$ multiplied by LM matrices from each side; 
\item or LM matrices itself.
\end{enumerate}
\end{proposition}

\begin{remark}
\label{small}
There is the following restriction on numbers $d$ and $k$ imposed in the statement of \cref{prop1=}:
\[
\Big(k+1 \Big) \Big( 1 + \sqrt[\leftroot{-3}\uproot{3}{k-1}]{k} \Big) \leq d .
\]
for $k>1$ and $2<d$ for $k=1$. 
This bound is related to the necessary condition for existence and extension of combinatorial designs called \emph{mutually orthogonal hypercubes}. 
We discuss them in detail in \cite{BurchardtRaissi20}. For example, for $k=2,3,4,5,6$ the smallest value of $d_{\text{min}}$ which does not satisfy the bound above is presented in a \cref{Aarhus21}. However, the given bound is not tight. In addition, \cref{Aarhus21} presents $d_{\text{max}}$ value, which is the maximal value of a local dimension $d$ for which \cref{prop1=} does not hold (we found a counterexample). We conjecture, that the asymptotic behavior of these values is $(k-1)^2$.
\end{remark}

\begin{table}[ht!]
  \begin{center}
    \begin{tabular}{|c|c|c|c|} 
    \hline
     k& AME(2k,d) & $d_{\text{min}}$ &$d_{\text{max}}$ \\
      \hline
    1&  AME(2,d) & $ 3$& $ 3$\\
    2&  AME(4,d) & $ 9$& $ 9$\\
    3&   AME(6,d) & $ 11$& $ 16$\\
    4&   AME(8,d) & $ 13$& $ 25$\\
        \hline
    \end{tabular}
  \end{center}
  \label{Aarhus21}
\caption{Values $d_{\text{min}}$ and $d_{\text{max}}$ for respective values of uniformity $k$. Note that the statement of \cref{prop1=} holds true for related AME(2k,d) states in any local dimension $d<d_{\text{min}}$.}
\end{table}

Note that for small values of $d=3, \ldots,12$, the Butson matrices are classified and contains $1,2,1,4,1,143,23,51,1,4497733$ matrices (classified up to monomial matrices) equivalently \cite{Lampio}. Tables of Butson matrices are available in \cite{BUlib,BUlib2}. 
Although, the class of Butson matrices $ \text{BH}(d,d)$ grows rapidly with the dimension $d$; it seems that the subclass of such matrices that might be involved in LU-equivalences of AME states is significantly smaller, and more specific classification of matrices providing eventual LU-equivalences of AME states is needed. 
We conjecture that all such matrices are Fourier transform and tensor product of such.

\section{One-uniform states}
\label{1-uniform states}
In three consecutive section, we apply results described in \cref{31} and \cref{32} for various classes of $k$-uniform states for $k=1,2,3$. As we shall see the analysis of their LU/SLOCC-equivalences differs substantially. 

We begin with $1$-uniform states. All $1$-uniform states of $N$-qudit system with the minimal support have the following form:
\[
\ket{\psi} =\dfrac{1}{\sqrt{d}} \sum_{i=0}^{d-1} \omega_i \ket{j^1_i ,\ldots ,j^N_i} ,
\]
where indices $j^\ell_i$ runs over all levels $0,\ldots,d-1$ for all indices $\ell$. 
One can observe that they are pairwise LU-equivalent, in particular equivalent to the generalized GHZ state $\ket{\text{GHZ}_d^N} $, presented in \cref{ex1GHZ} in \cref{chapter3}. 
Indeed, the following local transformation 
\[
U_1 \big( \ket{j^1_i} \big) =  \big( \omega_i^{-1} \ket{j_i^1} \big) , \quad
U_\ell \big( \ket{j_i^\ell} \big) =  \big(  \ket{j^1_i} \big) 
\]
for systems $\ell=2,\ldots,N$ gives aforementioned LU-equivalence. In that way, we conclude that all 1-uniform states of minimal support are LU-equivalent.

\section{Two-uniform states}
\label{2-uniform states}

Consider a family of $2$-uniform states  of $N$-qudit system with minimal support of the following form:
\begin{align*}
\ket{\phi_\alpha} :=  \dfrac{1}{d} \Bigg(&  \alpha\ket{0,\ldots,0} + \\
&\sum_{i,j \neq (0,0)} \ket{i,j} \otimes \ket{\phi_{i,j}} \Bigg)
\end{align*}
which is indexed by a complex numbers $\alpha$, $|\alpha |=1$. 
Each of such a state is LU-equivalent to $\ket{\phi_{\alpha =0}} $ by the following unitary transformations:
\begin{align*}
U_1 =& \text{diag} \Big( ({\overline{\omega_\alpha}})^{n-1}, 
1 ,\ldots ,
1  \Big), \\
U_i= & \text{diag} \Big(
({\overline{\omega_\alpha}})^{d-1},
\omega_\alpha ,\ldots ,
\omega_\alpha \Big) ,
\end{align*}
for $i=2,\ldots,N$, where $\omega_\alpha = \sqrt[\leftroot{-3}\uproot{3}{d (N-1)}]{\alpha} $ is an arbitrary root.
Therefore, all states $\ket{\phi_\alpha}$ are pairwise equivalent. Furthermore, if the exceptional phase stands by a different term, the similar transformation of such a state onto $\ket{\phi_{\alpha =0}}$ might be obviously given. 
Therefore, all $2$-uniform states of the form: 
\begin{equation*}
\ket{\phi_{\omega}} =  \dfrac{1}{d} \sum_{i,j } \omega_{i,j}\ket{i,j} \otimes  \ket{\phi_{i,j}}
\end{equation*}
are equivalent. 
We conclude this discussion in the following corollary.

\begin{corollary}
\label{phases}
Two $2$-uniform states of $N$-qudit system with the minimal support which differs only by phases,
\begin{align*}
\ket{\psi} =&\sum_{I \in \mathcal{I}} \omega_I \ket{I} , \\
\ket{\psi '} =&\sum_{I \in \mathcal{I}} \omega_I ' \ket{I} ,
\end{align*}
where the sum runs over multi-index set $\mathcal{I} \subset [d]^N$ of size $| \mathcal{I} | =d^2$, are always LU-equivalent 
(and hence belong to the same SLOCC class).
\end{corollary}

From \cref{phases}, it becomes clear that the diversity of possible phases in front of each term in the $2$-uniform state with minimal support does not reflect in the number of SLOCC classes. 
In fact, each $2$-uniform state with minimal support is equivalent to the one with all equal phases. 

Henceforward, the enumeration of SLOCC classes of $2$-uniform states with the minimal support might be restricted only to states with all phases equal. 
In that way, it coincides with the classification of related orthogonal array. 

\begin{corollary}
Classification of $2$-uniform states with minimal support for $N>4$ particles is equivalent to the classification of relevant orthogonal arrays, OA($d^2$,N,$d$,$2$), up to permutation of indices on each position. Potentially, for $N=4$ two AME(4,d) states of minimal support might be in the same SLOCC class, even though the corresponding OA are not equivalent.
\end{corollary}

In literature, the classification of OAs is usually considered up to permutations of rows and columns \cite{HedayatIndex1OA,BushStudies}. Note that permutation of columns resembles the physical operation of  exchanging subsystems. Hence, while verifying the SLOCC-equivalence one should always indicate whether such operations are considered \cite{FourQubits,FourQubits8}. 
In particular, by the classification of OAs, there is at most one OA($d^k$,$N$,$d$,$k$) for dimensions $d=2,\ldots ,17$ and for any strength $k$ and any number $N$ \cite{OAlib}. 
Therefore, $2$-uniform state with minimal support and the local dimension $d=2,\ldots ,17$ are always SLOCC equivalent or SLOCC equivalent after permutation of parties. 

\begin{conjecture}
All $2$-uniform states of minimal support are LU-equivalent, and hence represent the same SLOCC class.
\end{conjecture}

\section{Three-uniform states}
\label{3-uniform states}

In \cref{2-uniform states}, we have shown that the number of LU/SLOCC classes for $2$-uniform states of minimal support greatly coincides with the number of related OAs which are non-isomorphic. 
In particular, two states which differ only by phases are LU- and SLOCC-equivalent. 
As we shall see, this is in a strong contrast to the $3$-uniform states. 

We shall illustrate it on the example of AME($6,d$) states. Note, that for any $d\geq 4$ there exists an AME($6,d$) state with minimal support \cite{AME-QECC-Zahra}, its precise form might be obtained by reading consecutive rows of the corresponding OA($d^3$,6,$d$,3) from the OAs table \cite{OAlib}. 
One may enhance successive terms of an AME($6,d$) state with any phase factor $|\omega | =1$. Such family of states:
\begin{equation*}
\ket{\text{AME(6,d)}_\omega} =  \dfrac{1}{ d\sqrt{d}} \Bigg( 
\sum_{i,j,k=0}^{d-1} \omega_{i,j,k}  \ket{i,j,k} \otimes \ket{\psi_{i,j,k}}  \Bigg) 
\end{equation*}
is obviously a family of AME($6,d$) states. 
We focus our attention on states with all phases are equal to unity $\omega_{i,j,k} =1$ with one exception: $\omega_{0,0,0}=\alpha$. We denote them as $\ket{\psi_\alpha}$. In \cite{BurchardtRaissi20}, I showed that states $\ket{\psi_{e^{i\phi_1}}}$ and $\ket{\psi_{e^{i\phi_2}}}$ are not LU-equivalent unless $\phi_1 = \phi_2 +  \pi t$ for some integer $t$. In such a way, I obtained a continuous family of non-LU-equivalent AME($6,d$) states with minimal support. 

\begin{corollary}
With the above notation, the AME($6,d$) states, $d\geq 4$:
\begin{equation}
\ket{\text{AME(6,d)}_{e^{i\phi}}} :=  
\dfrac{1}{ d\sqrt{d}} \Bigg(
e^{i\phi} \ket{000000} + 
\sum_{i,j,k\neq (0,0,0)}  \ket{i,j,k} \otimes \ket{\psi_{i,j,k}}  \Bigg) 
\end{equation}
are pairwise in different LU- and SLOCC-classes for all phases $\phi \in [0,\pi )$.
\end{corollary}

In fact, for any $k$-uniform state with minimal support where $k>2$, the similar construction of continuous non LU-equivalent family might be provided.

\begin{corollary}
\label{AMEminSUP6}
If there exists a $k$-uniform state with minimal support $\ket{\psi}$ where $k>2$, then there are infinitely many pairwise non LU- and SLOCC-equivalent $k$-uniform states
\end{corollary}

\section{Number of non SLOCC-equivalent AME orbits}

In this section, we shortly summarize the number of non-SLOCC-equivalent AME states (the number of orbits) and AME states with minimal support. 

The existence of AME states with minimal support for $N,d <8$ was analyzed \cite{Bernal75} based on the table of OAs and similar combinatorial designs. According to the discussion presented in the previous sections, if $N\geq 6$ existence of AME($N,d$) state with minimal support persuade to infinitely many non-SLOCC-equivalent such states, see \cref{AMEminSUP6}.  in \cref{table1} and \cref{table2} respectively. 

Note, that verification of the existence of AME states which are not necessarily of the minimal support is more complex problem. Several results concerning this problem \cite{HIGUCHI2000213,AME(4,Felix72,Bernal75,Huber_2018} are summarizes in the tables of AME states \cite{Table_AME}. Although the exact classification of AME states up to SLOCC-equivalence is not known, in some specific cases non-trivial lower bound is given.

\begin{table}
\begin{center}
\includegraphics[width=0.75\textwidth]{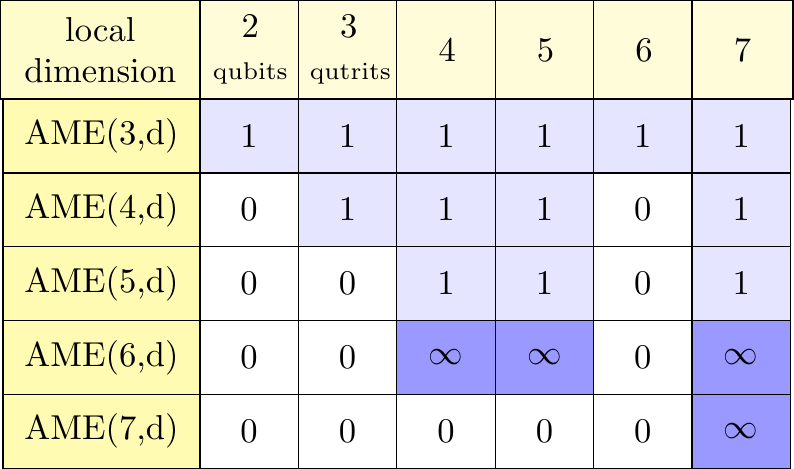}
\end{center}
\caption{\label{table2} 
The exact number of not SLOCC-equivalent AME orbits with minimal support presented on a differently shaded blue background. Note that classification is done up to the permutation of qudits.}
\end{table}

\begin{table}
\begin{center}
\includegraphics[width=0.75\textwidth]{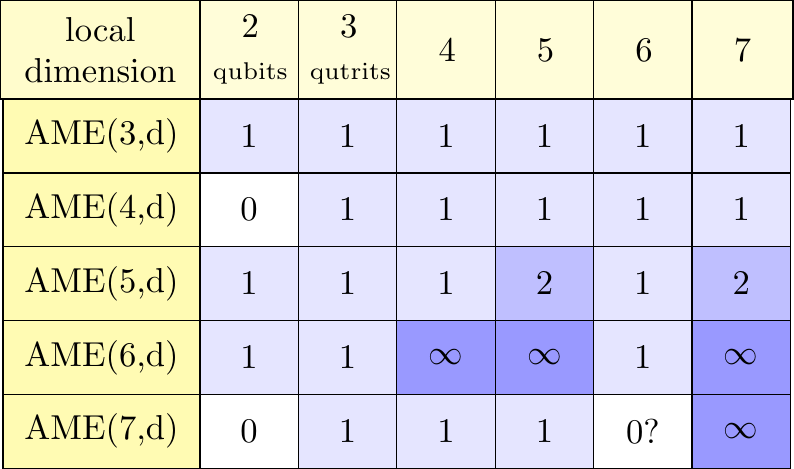}
\end{center}
\caption{\label{table1} 
The minimal number of non-SLOCC-equivalent AME orbits. The question mark by zero value suggests that the existence of the relevant state is dubitative, while $0$ itself emphasizes that the relevant state certainly does not exist. Note that the AME($4,6$) was constructed in \cref{ch4}.}
\end{table}

\section{Further discussion and open problems}
\label{Further discussion and open problems}

As we have shown, the LU- and SLOCC-classification of $k$-uniform states, even of minimal support, is in fact a complex project, which involves many open mathematical problems, in particular:
\begin{enumerate}
\item Existence and extension of mutually orthogonal Latin hypercubes.
\item Classification of Hadamard matrices of Butson type B($d,d$).
\item Classification/uniqueness of OAs of index unity (without permutation).
\end{enumerate}
among others. In addition, we discuss some open problems regarding LU- and SLOCC-classification of $k$-uniform states with minimal support in a detailed way. 

Firstly, consider any two $k$-uniform states of minimal support $\ket{\psi}$ and $\ket{\psi '}$ with all phases equal to $1$. We have showed that both states are LU-equivalent if and only if there exist local permutation matrices relating $\ket{\psi}$ and $\ket{\psi '}$. On the other hand, both states are in one-to-one correspondence with OAs of index unity. 
In such a way, the existence of local permutation is equivalent to an isomorphism between two OAs of index unity, and hence LU-classification is equivalent to the classification of OAs of index unity. 
Classification of OAs of index unity is an open mathematical problem. For small number of parties $N$, uniformity $k$, and the local dimension $d$ is small $d<9$, it is known that all OAs of index unity are isomorphic \cite{BULUTOGLU2008654,doi:10.1080/00401706.1992.10484952,stufken2007,Seveso_2018}. 

\begin{conjecture}
\label{productCon}
All OAs of index unity with small local dimension $d<9$ are isomorphic by permutations of symbols on each level. 
Equivalently, all $k$-uniform states with minimal support and all terms phases equal are LU-equivalent.
\end{conjecture}

Secondly, in \cref{prop1} shows that any LU-operator between two $k$-uniform states of minimal support consists of products of phase (diagonal) and permutation matrices (for $2k<N$). 
Furthermore, considering states such states with various phases, we showed that not all of them are LU-equivalent for $k>2$. 
The precise description of SLOCC classes, however, is not given. 

Thirdly, \cref{32} discusses the significant difference between $k$-uniform states of minimal support with $2k<N$ and $2k=N$ particles. We showed that LU-equivalence between two $k$-uniform states with $2k=N$ might be decomposed into multiplication of Butson-type matrices and local monomial  matrices. 
From our analysis, it is not yet clear whether such LU-equivalences are beyond local monomial equivalences. 
Indeed, in all examples where LU-equivalence involves Butson-type matrices, related states were always LM-equivalent. 
Therefore we conjecture that \cref{coro1} holds true in the case $2k =N$ (even though \cref{prop1} does not hold anymore). 

\begin{conjecture}
\label{Aarhus16}
AME(2k,d) states with minimal support are LU-equivalent if and only if they are LM-equivalent
\end{conjecture}

\label{Aarhus15}

Finally, we address the problem of equivalence of AME($2k$,$d$) states in composite dimensions. It turned out, that in many cases composite dimensions allows for the construction of AME states based on OAs of different linear structures. 
\cref{chapter3} presents quantum states related to OAs with different linear structures. In particular, in \cref{Aarhus14}, we show two states $\ket{\psi}_{\GF{4}}$, and $\ket{\psi}_{\mathbb{Z}_4}$ corresponding to OAs of different strength. It turned out that quantum states $\ket{\psi}_{\GF{4}}$, and $\ket{\psi}_{\mathbb{Z}_4}$ are 1-, and 2-uniform respectively, see \cref{OA1}. Notice that an immediate consequence is that both states are not SLOCC-equivalent. Indeed, one may compare the ranks of reduced density matrices, which agree for SLOCC-equivalent states \cite{RanksReducedMatrices}.

Nevertheless, \cref{Aarhus14} presents three AME states $\ket{\psi}_{\mathbb{Z}_9}$, $\ket{\psi}_{\mathbb{Z}_3 \oplus \mathbb{Z}_3}$, and $\ket{\psi}_{\GF{9}}$ of four particles with the local dimension $d=9$. They are related to OAs of a linear structures over rings $\mathcal{R}= \GF{9}, \mathbb{Z}_9$ and $\mathbb{Z}_3 \oplus \mathbb{Z}_3$ respectively. Notice that $\oa{81,4,9,2}$ corresponding to those three different ring structures are not equivalent by relabelling symbols. This might be scrutinized by analysis of existing tables of OAs \cite{OAlib}. This, however, does not imply that the corresponding quantum states are not LU-equivalent \cite{BurchardtRaissi20}, however, the numerical analysis suggests so. Note that results obtained in \cref{prop1=} do not allow for the verification of SLOCC-equivalence of the aforementioned states.

\begin{conjecture}
\label{proppppp2}
Consider three rings $\mathcal{R}= \GF{9}, \mathbb{Z}_9$ and $\mathbb{Z}_3 \oplus \mathbb{Z}_3$ with nine elements presented in \cref{figGalois9}. Three states constructed via
\begin{equation}
\ket{\psi}= \frac{1}{d}\sum_{i_1,i_2 =0}^{8} 
\ket{i_1} \ket{i_2} 
\ket{i_1 +i_2} \ket{ 2 i_1 +i_2},
\end{equation}
where addition and multiplication are considered over the relevant ring $\mathcal{R}$, are pairwise non SLOCC-equivalent AME(4,9) states.
\end{conjecture}

\section{Conclusions}

In this chapter, I develop general techniques of SLOCC-verification between $k$-uniform and AME states. In particular, I show that two $k$-uniform states of $N$ particles with the minimal support, where $k<2N$, are SLOCC-equivalent iff they are locally monomial (LM)-equivalent (\cref{prop1,coro1}). Furthermore, when uniformity achieves its upper bound, i.e. $k=2N$, such equivalence is always provided by a Butson-type matrix or monomial matrix (\cref{prop1=}). This restriction is valid, however, only for small local dimensions $d$ and number of parties $N$ (\cref{small}), in particular for arbitrary $N$ and $d<9$. I illustrate the usefulness of the provided criteria on various examples. Firstly, I show that the existence of AME states with minimal support of 6 or more particles yields the existence of infinitely many such non-SLOCC-equivalent states (\cref{AMEminSUP6}). Secondly, I show that some AME states cannot be locally transformed into existing AME states of minimal support (\cref{AME55}).

\chapter{Roots of polynomial invariants}
\label{chap6}

In this chapter, I provide necessary and sufficient conditions for generic $N$-qubit states to be equivalent under SLOCC operations using a single polynomial entanglement measure. More precisely, I investigate how the roots of the entanglement measure behave under SLOCC operations. I demonstrate that SLOCC operations may be represented geometrically by Möbius transformations on the roots of the entanglement measure on the Bloch sphere. I show that if the states are SLOCC-equivalent, then the roots of the polynomial entanglement measure for each state must be related by a Möbius transformation, which is straightforward to verify. I use this procedure to show that the roots of the $3$-tangle measure classify $4$-qubit generic states. Moreover, I propose an alternative method to obtain the normal form of a $4$-qubit state which bypasses the possibly infinite iterative procedure. An extension of presented results in which the author's part was not substantial can be found in the joint work \cite{Bur21SLIP}.

\section{Polynomial Invariant Measures}

An \textit{entanglement measure} is any function $E(\ket{\psi})$ defined for all pure states of $N$ qubits which vanishes for all separable states. As we discussed in Introduction, a particularly desired feature of an entanglement measure is invariance under SLOCC operations. Apart from the normalization of states, a SLOCC operation acting on $N$-qubit state $\ket{\psi}$ can be represented by the action of a local invertible operator with a determinant equal to one, $\mathcal{O}_{\vec{N}} \in \SLC^{\otimes N}$, where $\SLC$ is the special linear group of complex matrices of order 2. In this chapter, we consider only the subclass of invertible matrices with determinant one, unlike in \cref{chap5}, see \cref{AArhusNoc}. To distinguish the subclass of invertible matrices with determinant one from the class of all invariable matrices, we shall denote them with the calligraphic font $\mathcal{O}$ in contrast to the notation $O$ used in \cref{chap5} for all invertible matrices. 

A given entanglement measure $E(\ket{\psi})$ defined on the system of $N$ qubits is called a \textit{$\SL$-invariant polynomial of homogeneous degree $h$} if it is polynomial in the coefficients of a pure state $\ket{\psi}$ and satisfies
\begin{equation}
\label{AarhusWC}
E\big( \kappa \; \mathcal{O}_{\vec{N}} \ket{\psi} \big) =\kappa^h  E\big(\ket{\psi} \big)
\end{equation}
for each real constant $\kappa > 0$ and operator $\mathcal{O}_{\vec{N}} \in \SLC^{\otimes N}$~\cite{ThreeQub,SLOCCallDim,Eltschka_2014}. 
Note that \cref{AarhusWC} relates the entanglement measure of any non-normalized state with the entanglement measure of the related normalized state by fixing operator $\mathcal{O}_{\vec{N}}$ as an identity operator and choosing an appropriate constant $ \kappa $. 
A \textit{$\SL$-invariant polynomial} ($\SLIP$) measure of homogeneous degree $h$ will be denoted as $\SLIP_N^h$, where the upper index indicates the degree $h$ of the polynomial and the lower index is related to the number of qubits $N$. 

Among various approaches to the problem of quantification and classification of entanglement, the one via $\SLIP$ measures turned out to be a particularly useful. The most famous examples of such measures are concurrence and three-tangle, which measure the 2-body and 3-body quantum correlations of the system respectively \cite{PhysRevLett.80.2245,DistributedEntanglement}. $\SLIP$ measures provide a convenient method for entanglement classification and its practical detection at the same time. For example, it was demonstrated that almost all SLOCC equivalence classes can be distinguished by ratios of $\SLIP$ measures \cite{SLOCCallDim}. Moreover, any given two states are SLOCC-equivalent if a \textit{complete set of SLIP measures} achieves the same values for both of them \cite{PolInv4qubits}. However, the size of such a set grows exponentially with the number of qubits $N$, making it intractable to use this approach to decide SLOCC-equivalence between states with more than four qubits \cite{Love07}. 

Any $\SLIP$ measure $E$ might be further extended to the set of mixed states by determining the largest convex function on such set which coincides with $E$ on the set of pure states  \cite{Uhlmann98}. Despite the simple definition of a \textit{convex roof extension}, its evaluation requires non-linear minimization procedure, and through this, is a challenging task for a general density matrix \cite{Osborne_2006,Regula_2014,
PhysRevLett.114.160501,Regula_2016a,SG18}. An auspicious attempt to address this task was carried out by introducing the so-called \textit{zero-polytope}, the convex hull of pure states with vanishing $E$ measure \cite{OsterlohTangles,OsterlohWernerStates,
OsterlohExactZeros,Osterloh3}. In the simplest case of rank-2 density matrices $\rho $, the zero-polytope can be represented as a convex polytope inscribed in a Bloch sphere, spanned by the roots of $E$ \cite{OsterlohWernerStates,RegulaGeoTanglePRL}. I adapt this approach, and focus only on the vertices of the zero-polytope, equivalently the roots of polynomial invariant.

\section{System of roots}

We begin with the general discussion of roots of a $\SLIP$ measure, which adapts approach presented in \cite{OsterlohWernerStates,RegulaGeoTanglePRL}. Firstly, consider a $(N+1)$-partite qubit state $\ket{\psi}$. Such a state can be uniquelly written as
\begin{equation}
\label{Peq1}
\ket{\psi}=\ket{0} \ket{\psi_0} + \ket{1}\ket{\psi_1}\,,
\end{equation}
which provides the canonical decomposition of its reduced density matrix 
\[
\rho =\ket{\psi_0}\bra{\psi_0} +\ket{\psi_1}\bra{\psi_1}
\]
obtained by tracing out the first qubit. Notice that both states $\ket{\psi_0}$ and $\ket{\psi_1}$ are in general neither normalized nor orthogonal. Secondly, consider the following family of non-normalized states
\begin{equation}
\label{psi0psi1}
\ket{\psi_z} = z \ket{\psi_0} +  \ket{\psi_1},
\end{equation}
where $z\in \hat{\mathbb{C}} $ is taken from the extended complex plane $\hat{\mathbb{C}}$, which contains all complex numbers plus infinity. We shall refer to such a representation as the \textit{extended plane representation}. 

Furthermore, consider any $\SLIP_N^h$ measure $E$ defined on the set of $N$-partite pure qubit states. By the definition, $E$ is polynomial in the coefficients of $\ket{\psi_z}$, hence it is also polynomial in the complex variable $z$  \cite{OsterlohTangles}.
In such a way, the polynomial 
\[
E( z\ket{\psi_0} + \ket{\psi_1} )
\]
has exactly $h$ roots: $\zeta_1 ,\ldots ,\zeta_h $ (which may be degenerated and/or at infinity), in accordance with the degree of $E$. 

Using the complex number $z$, the states $\ket{\psi_z}$ can be mapped to the surface of a sphere via the standard \textit{stereographic projection}
\[(\theta,\phi ):=
(\text{arctan} \,1/ |z|,\; -\text{arg}\,z )\]
written in spherical coordinates.
Furthermore, a point on the unit 2-sphere $(\theta,\phi)$ can be associated with the following quantum state
\begin{equation}
\ket{\widetilde{\psi}_{z}} :=
 \text{cos} \dfrac{\theta}{2} \ket{\psi_0} +
 \text{sin} \dfrac{\theta}{2} e^{i \phi} \ket{\psi_1}
\end{equation}
with $z =\text{ctg}({\theta}/{2}) \: e^{-i \phi}$. Note that $\ket{\psi_0}$ lies in the North pole and $\ket{\psi_1}$ lies in the South pole, see \cref{Fig:StereoProjection}. We shall refer to such a representation as the \textit{Bloch sphere representation}. Note the following proportionality of states: ${\ket{\widetilde{\psi}_{z}} \propto \ket{\psi_{z}}}$, and that neither of these states is normalized, since $\ket{\psi_0}$ and $\ket{\psi_1}$ are not normalized in general either.

\begin{figure}[ht!]
\begin{center}
\includegraphics[width=.7\columnwidth]{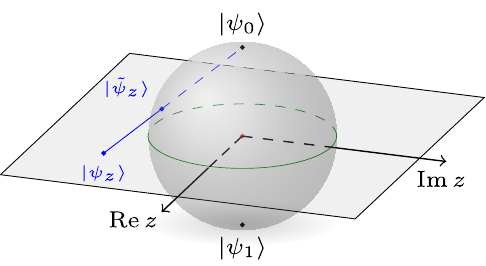}
\caption{The stereographic projection relates the family of states $\ket{\psi_z}$ on the extended complex plane with the family of states $\ket{\widetilde{\psi_z}}$ on the Bloch sphere. The spherical coordinates $({\theta},\phi)$ and the complex coordinate $z$ are related by the stereographic projection $z =\text{ctg} ({\theta}/{2}) \: e^{-i \phi}$.
}
\label{Fig:StereoProjection}
\end{center}
\end{figure}

\section{Local operations on the system of roots}

Each linear invertible operator $\mathcal{O} =\begin{psmallmatrix}
a&b\\
c&d
\end{psmallmatrix}$, might be associated with a \textit{Möbius transformation} $z\mapsto z':= \frac{az+b}{cz+d}$, which maps the extended complex plane $\hat{\mathbb{C}}$ into itself \cite{bengtsson_zyczkowski_2006,doi:10.1142/S0219749912300045}. The composition of two such transformations is related to the multiplication of the associated operators. Furthermore, $z\mapsto z':= \frac{dz-b}{-cz+a}$ is an inverse Möbius transformation related with
$\mathcal{O}^{-1} =
\begin{psmallmatrix}
d&-b\\
-c&a
\end{psmallmatrix}$.
Although Möbius transformations are typically represented on the extended complex plane, one may also represent them as transformations on the Bloch sphere via the stereographic projection. This correspondence between invertible operators and Möbius transformations represented on the Bloch sphere was already successfully used for SLOCC classification of permutation-symmetric states \cite{Bastin_2009,CrossRatioNqubits,Three-tangleMajorana}. Studies of effects of SLOCC operations on the system of roots might be summarized in the theorem below.

\begin{theorem}
\label{T1}
Consider an $(N+1)$-partite pure quantum state $\ket{\psi} =\ket{0}\ket{\psi_0}+\ket{1}\ket{\psi_1}$.
The roots $\zeta_i$ of any $\SLIP_N^h$ entanglement measure associated to the partial trace of the first qubit:
\begin{enumerate}
\item are invariant under invertible operators, i.e. invariant under $\textbf{1}\otimes\mathcal{O}_{\vec{N}} \in \SLC^{\otimes N}$ operators;
\item transform via an inverse Möbius transformation $\zeta_i '=\frac{d\zeta_i-b}{-c\zeta_i+a}$ w.r.t the
$\mathcal{O} =\begin{psmallmatrix}
a&b\\
c&d
 \end{psmallmatrix}\otimes \textbf{1}^N \in \SLC$ operator.
\end{enumerate}
\end{theorem}

\noindent
Before we proceed with the proof of \cref{T1}, we shall emphasize that normalizing the states $\ket{\psi_0}$ and $\ket{\psi_1}$ after the action of the local operator $\ma{O}_{\vec{N+1}}\in \SLC^{\otimes N +1}$, as is the case in existing related works \cite{OsterlohTangles,OsterlohWernerStates,Osterloh3,RegulaGeoTanglePRL,RegulaGeoTanglePRA},
spoils the consistency of the stereographic projection with the roots transformations. As a consequence, the action of SLOCC operators on the states $\ket{\psi_z}$ would no longer be given by the associated Möbius transformation, and the statements presented in \cref{T1} would no longer hold true.

\begin{proof}[Proof of \cref{T1}]
As we discussed, a $N+1$ partite qubit state $\ket{\psi} \in \mathcal{H}_2^{\otimes (N+1)}$ might be written as
\begin{equation}
\label{Peq111}
\ket{\psi}=\ket{0} \ket{\psi_0} + \ket{1}\ket{\psi_1}.
\end{equation}
This form provides the canonical decomposition of the reduced density matrix 
\[\rho_{} =\ket{\psi_0}\bra{\psi_0} +\ket{\psi_1}\bra{\psi_1}\] 
for the two non-normalized states $\ket{\psi_0}, \ket{\psi_1}\in \mathcal{H}_2^{\otimes N}$. Consider now a reversible operator
$\mathcal{O} =\begin{psmallmatrix}
a&b\\
c&d
\end{psmallmatrix} \in \SLC$ acting on the first qubit. Under the action of above operator, the state $\ket{\psi}$ is transformed into
\eq{
\ket{\psi'}:=
\mathcal{O} \ket{\psi}=\ket{0} \Big(a\ket{\psi_0}+b\ket{\psi_1}\Big) + \ket{1} \Big(c\ket{\psi_0}+d\ket{\psi_1}\Big) = \ket{0}\ket{\psi'_0} + \ket{1}\ket{\psi'_1}\,
}
where
\ea{
\ket{\psi'_0} & := a\ket{\psi_0}+b\ket{\psi_1} \label{p0l}\,, \\
\ket{\psi'_1} & := c\ket{\psi_0}+d\ket{\psi_1} \label{p1l}\,.
}
Consider now any superposition of states $\ket{\psi'_0}$ and $\ket{\psi'_1}$. In particular, observe that
\begin{align*}
\ket{\psi'_z} := z \ket{\psi'_0}+\ket{\psi'_1}
&=z \Big(a\ket{\psi_0} +b \ket{\psi_1} \Big)  +  c\ket{\psi_0} +d \ket{\psi_1} \\
&=\left(a z+b\right) \ket{\psi_0} +\left( cz+d \right) \ket{\psi_1}\\
&\propto \dfrac{az+b}{cz+d}\ket{\psi_0 }+ \ket{\psi_1}.
\end{align*}
In other words, we have
\eq{
\ma{O}\ket{\psi_z} = \ket{\psi_{z'}}\,, \quad z' = \frac{az+b}{cz+d}\,,
}
i.e. the operator $\ma{O}$ transforms states in the extended plane representation by applying a Möbius transformation on the index $z$. Suppose that $\zeta_i$ is a complex root of a polynomial function $E$, i.e. $E (\zeta_i \ket{\psi_0}+ \ket{\psi_1})=0$. Acting on the first qubit with operator $\ma{O}$, the density matrix obtained by tracing out the first qubit becomes $\ket{\psi'_0}\bra{\psi'_0} +\ket{\psi'_1}\bra{\psi'_1}$, and hence the entanglement measure $E$ will vanish for a roots $\zeta'_i$, such that $E (\zeta'_i \ket{\psi'_0}+ \ket{\psi'_1})=0$. Using Eqs.~(\ref{p0l})-(\ref{p1l}), the later equation transforms into
\eq{
E\left((c\zeta_i'+d)\left(\dfrac{a\zeta_i'+b}{c\zeta_i'+d} \ket{\psi_0}+ \ket{\psi_1}\right) \right)=0
}
where the factor $(c\zeta_i'+d)$ is not relevant. Comparing with the equation for the roots before the action of $\ma{O}$, we conclude that the roots transform according to the inverse Möbius transformation as
\eq{\label{Mzeta}
\zeta'_i =\dfrac{d\zeta_i-b}{-c\zeta_i+a} \,,
}
under the action of the operator $\ma{O}$. Consequently, the roots of the zero-polytope transform with respect to the inverse Möbius transformation associated to the operator $\mathcal{O} =\begin{psmallmatrix}
a&b\\
c&d
\end{psmallmatrix}$.

Furthermore, consider multi-local operators $\mathcal{O}_{\vec{N}}= \mathcal{O}_1\otimes\ldots\otimes\mathcal{O}_N$ acting on the remaining qubits of the state $\ket{\psi}$ from \cref{Peq111}.
The state $\ket{\psi}$ transforms accordingly as
\eq{
\ket{\psi'}:=
\mathcal{O}_{\vec{N}}  \ket{\psi}=\ket{0}
\underbrace{\mathcal{O}_{\vec{N}} \ket{\psi_0} }_{:=\ket{\psi_0'}}+
\ket{1}
\underbrace{\mathcal{O}_{\vec{N}} \ket{\psi_1} }_{:=\ket{\psi_1'}}.
}
After the action of $\mathcal{O}_{\vec{N}}$, a value of entanglement measure $E$ reads
\[
E \Big( z\ket{\psi_0'}+ \ket{\psi_1'} \Big)=
E \Big( \mathcal{O}_{\vec{N}} \big(z\ket{\psi_0}+ \ket{\psi_1}\big) \Big)\,.
\]
However, since $E$ is $\SLC^{\otimes N}$ invariant function, one may conclude that $E (z \ket{\psi_0}+ \ket{\psi_1} )=0$ iff $E (z \ket{\psi_0'}+ \ket{\psi_1'} )=0$. Hence the roots of both polynomial equations are the same. As a consequence, the roots of the the zero-polytope remain unchanged under the action of $\mathcal{O}_{\vec{N}}$. This concludes the proof of \cref{T1}.
\end{proof}

\section{States discrimination} 
\label{Aarhus17}

The decomposition (\ref{Peq1}) can be performed with respect to any other subsystem, each with its own system of roots. Any local operator $\mathcal{O}_k
=\begin{psmallmatrix}
a&b\\
c&d
\end{psmallmatrix}$ acting on the $k$-th subsystem will influence independently the corresponding $k$-th system of roots and the associated zero-polytope according to the Möbius transformation $\zeta_i \mapsto  \frac{d\zeta_i-b}{-c\zeta_i+a}$. A global action of a local operator $\mathcal{O}_1 \otimes \cdots \otimes \mathcal{O}_{N+1}$ shall affect all roots and thus related zero-polytopes. Since a Möbius transformations are bijections of the Bloch sphere, the total number of roots will always be preserved \cite{RegulaGeoTanglePRA}. Furthermore, since Möbius transformations are fully classified, the existence of a local transformation between two given states becomes straightforward to verify.

In this way, \cref{T1} provides novel solution for the problem of discriminating $(N+1)$-qubit states up to the SLOCC-equivalence \cite{PolInv4qubits,Zhang_2016,KempfNessToEntanglement,BurchardtRaissi20}.
To verify if two pure states are SLOCC-equivalent, one can thus use the following procedure:

\begin{itemize}
\item[\textbf{1)}] Choose any $\SLIP_N^h$ entanglement measure of degree $h \geq 3$ and calculate its roots for each subsystem for both states. Note that a generic state will always have $h$ roots for each subsystem.
\item[\textbf{2)}] Focus on one subsystem $i$, where $1\leq i\leq N+1$ and choose 3 of the total $h$ roots from each state.
\item[\textbf{3)}] Write the unique Möbius transformation between the two triplets of roots. Derive the local operator $\ma{O}_i$ associated to such transformation. 
\item[\textbf{4)}] Choose a different set of $3$ roots of the second state and repeat step 3). Repeat it for all $3! {{h}\choose{3}}$ possibilities.
\item[\textbf{5)}] Repeat steps 3) and 4) for all other subsystems. Consider the tensor products of all the local operators obtained. This results in a finite set of operators of the form $\ma{O}_1 \otimes \dots \otimes \ma{O}_{N+1}$.
\item[\textbf{6)}] If the two given $(N+1)$-qubit states are SLOCC-equivalent, one of obtained operators must transform one state into the other. Otherwise, they are not SLOCC-equivalent.
\end{itemize}

Presented procedure has two main important features. 
Firstly, it factorizes the problem of finding SLOCC-equivalence. Indeed, local operations are determined separately for each subsystem. Secondly, it discretizes the initial discrimination task. Indeed, there are at most $(3! {{h}\choose{3}})^{N+1}$ local operators which might provide SLOCC equivalence between initial states.

\begin{figure}[ht!]
\begin{center}
\includegraphics[width=.95\columnwidth]{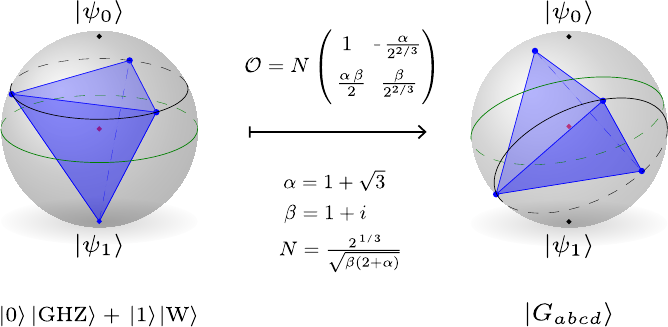}
\caption{The system of four roots (represented as blue dots) related to the 3-tangle polynomial measure $\tau^{(3)}$ evaluated on the first subsystem of the state $1/\sqrt{2}(\ket{0}\ket{\textrm{GHZ}}+\ket{1}\ket{\textrm{W}})$.
This system of four points can be mapped into a normal system (i.e. symmetrically related points $z,-z,1/z,-1/z$) by a Möbius transformation. Similar local transformations can be performed with respect to other subsystems, transforming the states into a state in the normal form.
}
\label{MobiusTrans}
\end{center}
\end{figure}

\section{Normal system of roots and cross-ratio}

Any three distinct points on the sphere can be transformed onto any other three distinct points via a unique Möbius transformation. This is not the case for four points on the sphere. For a given four points on the extended complex plane $z_1,z_2,z_3, z_4$, one may associate a so-called \textit{cross-ratio}
\begin{equation}
\label{CRR}
\lambda
\big(z_1,z_2,z_3,z_4 \big):=
\dfrac{z_3 -z_1}{z_3 -z_2}\dfrac{z_4 -z_2}{z_4 -z_1}\,,
\end{equation}
which is preserved under any Möbius transformations \cite{CrossRatioNqubits,bengtsson_zyczkowski_2006}. Two systems of four distinct points are related via Möbius transformations if and only if their associated cross-ratios are the same. The cross-ratio is not invariant under permutations of points, however, and depending on the ordering taken for the four points, it takes six related values \cite{CrossRatioNqubits}:
\[
\lambda, \frac{1}{\lambda}, 1-\lambda ,\frac{1}{1-\lambda},\frac{\lambda-1}{\lambda},\frac{\lambda}{\lambda-1}.
\]
A particular interesting set of four points is one of the form $z,1/z, -z,-1/z $, which we call a \textit{normal system}.

Consider such a set of four related complex points $\Phi=\{z,\frac{1}{z},-z,-\frac{1}{z} \}$. It is convinient to associate with them the cuboid spanned by eight points:
\[
\Phi\cup\bar{\Phi}=
\Big\{z,\frac{1}{z},-z,-\frac{1}{z} , \bar{z},\frac{1}{\bar{z}},-\bar{z},-\frac{1}{\bar{z}} \Big\},
\]
as it is presented on \cref{G24}. Notice that all six faces of the cuboid are parallel to one of the canonical planes: $XZ$, $XY$, or $YZ$. This property is equivalent to the initial assumption that the set of points $\Phi$ is in normal form.
Clearly, all rotations of the Bloch ball (associated to the unitary operations) preserve the form of the cuboid.
Nevertheless, only a special subgroup of all rotations leave faces of the cuboid parallel to $XZ$, $XY$, or $YZ$.
This special subgroup $\mathcal{G}_{24}$ contains exactly $24$ elements and is generated by three rotations of $\pi /2$ around $X$, $Y$, and $Z$ axis:
\begin{align}
\mathbf{R}_x ({\pi}/{2}) = &
\begin{psmallmatrix}\text{cos} \;\pi /4  & -i\;\text{sin}\;\pi /4\\-i\;\text{sin}\;\pi /4 & \text{cos}\;\pi /4\end{psmallmatrix}
=\frac{1}{\sqrt{2}}
\begin{psmallmatrix}1 & -i\\-i &1 \end{psmallmatrix}
,\; \label{Indeed} \\
\mathbf{R}_y ({\pi}/{2}) = &
\begin{psmallmatrix}\text{cos} \;\pi /4  & -\text{sin}\;\pi /4\\\text{sin}\;\pi /4 & \text{cos}\;\pi /4\end{psmallmatrix}
=\frac{1}{\sqrt{2}}
\begin{psmallmatrix}1 & -1\\1 &1 \end{psmallmatrix}
,\; \label{Indeed2} \\
\mathbf{R}_z ({\pi}/{2}) = &
\begin{psmallmatrix}e^{-i \pi /4 } & 0\\0 & e^{i \pi /4 }\end{psmallmatrix}
=\frac{1}{\sqrt{2}}
\begin{psmallmatrix}1-i  & 0\\0 & 1+i \end{psmallmatrix}. \label{Indeed3}
\end{align}
Notice that $\mathcal{G}_{24}$ is a group of rotations preserving the regular cube (the group of orientable cube symmetries). All rotations in the $\mathcal{G}_{24}$ group preserve the normal-form structure of $\Phi$. Therefore, the normal form is uniquelly determined up to $24$ rotations in the $\mathcal{G}_{24}$ group.

\begin{figure}[ht!]
\begin{center}
\includegraphics[scale=1.3]{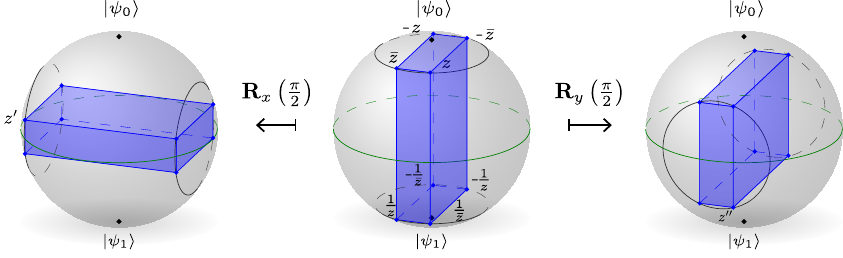}
\caption{A normal system of four points $z,1/z,-z,-1/z$ together with the conjugate points $\bar{z},1/\bar{z},-\bar{z},-1/\bar{z}$ span the cuboid whose faces are parallel to the $XZ$, $XY$ and $YZ$ planes. There are exactly $24$ rotations of the Bloch sphere which preserve this property, composing the elements of the group $\ma{G}_{24}$.
Two of them: rotation by a $\pi /2$ angle around $X$ and $Y$ axes are presented on the left and right respectively. The system of roots transforms according to Eqs.~(\ref{Indeed}-\ref{Indeed3}), which gives
$z\mapsto z':= \frac{z-i}{-iz+1}$ and $z\mapsto z'':= \frac{z-1}{z+1}$ for the two presented rotations.
}
\label{G24}
\end{center}
\end{figure}

We show that any set of four points may be mapped into a normal system, for which $z,1/z, -z,-1/z $.

\begin{proposition}
\label{24}
Any system of non-degenerated four points $z_1, z_2, z_3, z_4$ on the Bloch sphere can be transformed onto the normal form $z,\frac{1}{z},-z,-\frac{1}{z} $ via some Möbius transformation $T$. Transformation $T$ is uniquely defined up to $24$ rotations in the group $\mathcal{G}_{24}$.
\end{proposition}

\begin{proof}
Consider a complex number $\lambda$. There exists a complex number $z$, such that the cross-ratio of the four points is equal to $\lambda$, i.e.
\begin{equation}
\Big(z,\frac{1}{z};-z,-\frac{1}{z} \Big)=
\lambda\,.
\label{CRvia}
\end{equation}
Indeed, the cross-ratio on the left side equals ${4 z^2} /{(1+z^2)^2}$, and the equation ${4 z^2} /{(1+z^2)^2}=\lambda$ has exactly four complex solutions
\begin{equation}
\label{eq23}
z_0 =\frac{4-2\lambda + \sqrt{1-\lambda}}{2\lambda}
,\; \frac{1}{z_0},\;-z_0,\;-\frac{1}{z_0}.
\end{equation}
In such a way, for a given value $\lambda$ there exists a unique $\lambda$-normal system, such that the cross-ratio of its vertices is given by $(z_0,\frac{1}{z_0};-z_0,-\frac{1}{z_0} )=\lambda$.
Note that replacement of the vertex $z_0$ by any other vertex $z_0,{1}/{z_0},-z_0$, or $-{1}/{z_0}$ does not change the value of the cross-ratio $(z_0,\frac{1}{z_0};-z_0,-\frac{1}{z_0} )=\lambda$. Furthermore, there is a unique Möbius transformation $T$ which maps $z_1,z_2,z_3$ onto $z_0, {1}/{z_0},-z_0$, with the remaining $z_4$ mapped onto $-{1}/{z_0}$. Observe that the value $z_0$ is unique up to its inverse, opposite and opposite inverse elements, according to \cref{eq23}, with the corresponding Möbius transformations associated to the matrices $T, \sigma_x T ,\sigma_y T$ and $\sigma_z T$.
Each of those transformations maps the entire set of points $\{z_1,z_2,z_3,z_4\}$ onto the same set of points $\{z_0, {1}/{z_0},-z_0,-{1}/{z_0}\}$, (the exact bijection between those two sets of roots is different for each transformation).

Depending on the chosen order of four points $\{z_1,z_2,z_3,z_4\}$, the corresponding cross-ratio takes one of six values:
$\lambda, \frac{1}{\lambda}, 1-\lambda ,\frac{1}{1-\lambda},\frac{\lambda-1}{\lambda}$ and $\frac{\lambda}{\lambda-1}$. 
For each value, there is a corresponding set of solutions of the form $\{z_0, {1}/{z_0},-z_0,-{1}/{z_0}\}$ via \cref{eq23} with related four Möbius transformations. hence, in total there are $24$ Möbius transformations that map any four non-degenerated points onto a normal system. Each transformation is related to an element of the group $\mathcal{G}_{24}$ which has exactly $24$ elements.
\end{proof}

\section{SLOCC classification of four qubit states}

In this section I discuss how one may use the normal system of roots for SLOCC-classification of small systems.

Firstly, consider the three-qubit case. Genuinely entangled pure three-qubit states are SLOCC-equivalent to one of two states \cite{threeQubits}
\[
\ket{\textrm{GHZ}}=\frac{1}{\sqrt{2}}(\ket{000}+\ket{111}),
\]
\[
\ket{\textrm{W}}=\frac{1}{\sqrt{3}}( \ket{001}+\ket{010}+\ket{100}),
\] 
as we discussed in Introduction. Consider the 2-tangle $\tau^{(2)}$ \cite{PhysRevLett.80.2245} as the entanglement measure, and its roots as we discussed so far. One may use the roots of $\tau^{(2)}$ to distinguish between the two classes. Indeed, all reduced density matrices of the $\ket{\textrm{W}}$ state have a single root, while there are always two distinct roots for the $\ket{\textrm{GHZ}}$ state \cite{RegulaGeoTanglePRL}.

Secondly, consider the four-qubit case. Contrary to the case of three qubits, there are infinitely many SLOCC classes of four qubit states \cite{threeQubits}. Nevertheless, four-qubit states were successfully divided into nine families, most of which contains an infinite number of SLOCC-classes  \cite{FourQubits,Djokovic4qubits,SpeeKraus}. The family with the most degrees of freedom, so called $G_{abcd}$ family, is represented by states of the following form \cite{FourQubits}:
\begin{align*}
\ket{G_{abcd}}&=\frac{a+d}{2} \big(\ket{0000}+\ket{1111} \big)+\frac{a-d}{2} \big(\ket{0011}+\ket{1100} \big)
\\
&+\frac{b+c}{2}  \big(\ket{0101}+\ket{1010} \big)+\frac{b-c}{2} \big(\ket{0110}+\ket{1001} \big),
\end{align*}
where $a^2 \neq b^2 \neq c^2 \neq d^2 $ are pairwise different. Note that any generic 4-qubit states, i.e. 4-qubit states with random coefficients belonging to the $G_{abcd}$ family, since it has the most degrees of freedom. 

Choosing the 3-tangle $\tau^{(3)}$ \cite{DistributedEntanglement} as the entanglement measure, we computed roots of states $\ket{G_{abcd}}$. As we shall see, each state has four non-degenerate roots already in the normal form. Indeed, the state $\ket{G_{abcd}}$ has the following decomposition with respect to the first subsystem $\ket{G_{abcd}} =\ket{0}\ket{\psi_0} +\ket{1}\ket{\psi_1}$, where
\begin{align*}
\ket{\psi_0} &= \frac{a+d}{2}\ket{000} +\frac{a-d}{2}\ket{011} +\frac{b+c}{2}\ket{101} +\frac{b-c}{2}\ket{110}\,, \\
\ket{\psi_1} &= \frac{a+d}{2}\ket{111} +\frac{a-d}{2}\ket{100} +\frac{b+c}{2}\ket{010} +\frac{b-c}{2}\ket{001}\,.
\end{align*}
Suppose that the three-tangle vanishes, $\tau^{(3)} (\zeta \ket{\psi_0} +\ket{\psi_1} )=0$. Since $\tau^{(3)}$ is a $\SLC^{\otimes 3}$ invariant measure, for any local operators $\mathcal{O}_1$, $\mathcal{O}_2$, $\mathcal{O}_3$ we have
\[
\tau^{(3)} \Big( (\mathcal{O}_1\otimes\mathcal{O}_2\otimes\mathcal{O}_3)
\big(\zeta\ket{\psi_0} +\ket{\psi_1} \big)\Big)=0\,.
\]
Thus
\begin{align*}
\ket{\psi_0}=&(\sigma_x \otimes \sigma_x \otimes \sigma_x) \ket{\psi_1}\,,\\
\ket{\psi_1}=&(\sigma_x \otimes \sigma_x \otimes \sigma_x) \ket{\psi_1}\,,
\end{align*}
where $\sigma_x,\sigma_y,\sigma_z$ are Pauli matrices.
Hence, taking all local operators $\mathcal{O}_1$, $\mathcal{O}_2$, and $\mathcal{O}_3$ equal to $\sigma_x$, we conclude that
\begin{equation}
0=\tau^{(3)} \Big( (\sigma_x \otimes \sigma_x \otimes \sigma_x)
\big(\zeta\ket{\psi_0} +\ket{\psi_1} \big)\Big)=
\zeta\ket{\psi_1} +\ket{\psi_0}
\propto
\frac{1}{\zeta}\ket{\psi_0}+\ket{\psi_1},
\end{equation}
and $1/\zeta$ is another root of $\tau^{(3)}$. Similarly, considering $(\sigma_y \otimes \sigma_y \otimes \sigma_y)$ and $(\sigma_z \otimes \sigma_z \otimes \sigma_z)$, one may find another two roots $-\zeta, \;-1/\zeta$ of a measure $\tau^{(3)}$. Furthermore, it shows that the roots of $\tau^{(3)}$ evaluated on any state from the $G_{abcd}$ family are symmetrical with respect to rotations around $X,Y,Z$ axes by the angle $\pi$. Writting $\tau^{(3)} (z \ket{\psi_0} +\ket{\psi_1} )=0$ explicitely, we obtain the equation
\begin{equation}
\tau^{(3)} (z \ket{\psi_0} +\ket{\psi_1} ) = A z^4 - 2(2 B+A)z^2 + A = 0,
\label{eqq}
\end{equation}
where $A = (b^2 - c^2)  (a^2 - d^2)$ and $B = (c^2-d^2)  (a^2 - b^2)$. The above equation is non-degenerated iff $A,B,A+2B\neq 0$, which happens iff $a^2 \neq b^2 \neq c^2 \neq d^2 $ are pairwise different.

Recall that the normal form of roots is unique up to the group of rotations $\mathcal{G}_{24}$. In such a way, the problem of SLOCC-equivalence of states $\ket{G_{abcd}}$ becomes solvable, with a discrete amount of solutions. Indeed, two states in the $G_{abcd}$ class are SLOCC equivalent iff one can be obtained from the other by the action of an element of the finite class of operators $\mathcal{O}\in \mathcal{G}_{24}^{\otimes 4}$. We thus find that exactly 192 states of the form $\ket{G_{a bcd}}$ are SLOCC-equivalent.

\begin{proposition}
Two states $\ket{G_{a bcd}}$ and $\ket{G_{a' b'c'd'}}$ are SLOCC-equivalent iff their coefficients are related by the following three operations: 
\begin{enumerate}
\item multiplication by a phase factor $(a', b',c',d')=e^{i\phi} (a ,b,c,d)$, 
\item and permutation $\sigma \in S_4$ of coefficients $(a', b',c',d')=\sigma (a ,b,c,d)$, 
\item and change of sign in front of two coefficients from $a,b,c,d$.
\end{enumerate} 
\label{propGabcd}
\end{proposition}

Before we proceed with the proof, note that the symmetry in \cref{propGabcd} is given by the Weyl group of Cartan type $D_4$. Such symmetry has already been observed among generators of four-qubit polynomial invariants exhibit this type of symmetry \cite{PolynomialInvariantsofFourQubits,Djokovic4qubits,Holweck_2014}. As a consequence, this result constitutes a new relation between 4-qubit invariants and the convex roof extension of 3-tangle $\tau^{(3)}$. This may shed some light on the problem of generalizing the CKW inequality \cite{DistributedEntanglement} for four qubit states \cite{GourWallach,ProblemTangle,Regula_2014,Eltschka_2014,Regula_2016a},
and beyond
\cite{ProblemTangle,GourWallach,MonogamyEqualitiesDeg2and4}.

\begin{lemma}
\label{lemmmmm}
Any local operator
$\mathcal{O}=
\mathcal{O}_1 \otimes \mathcal{O}_2 \otimes \mathcal{O}_3 \otimes \mathcal{O}_4 \in \SLC^{\otimes 4}$
which transforms states
$\ket{G_{a' b'c'd'}} \propto \mathcal{O} \ket{G_{a bcd}}$ with $a^2 \neq b^2 \neq c^2 \neq d^2 $, is of the form $\mathcal{O}_i \in \mathcal{G}_{24} $.
\end{lemma}

\begin{proof}
Consider a local operator $\mathcal{O}_1$ acting on the first qubit which transforms the state $ \ket{G_{abcd}}$ onto $\ket{G_{a'b'c'd'}}$. Operator $\mathcal{O}_1$ also transforms corresponding systems of roots denoted as $\Lambda$ and $\Lambda '$, respectively, via the action of the related Möbius transformation. Note that both systems $\Lambda$ and $\Lambda '$ are in the normal form, therefore, according to \cref{24}, we have that $\mathcal{O}_i \in \mathcal{G}_{24} $. Similar analysis with respect to all other qubits shows that $\mathcal{O}_2,\mathcal{O}_3,\mathcal{O}_4 \in \mathcal{G}_{24}$.
\end{proof}

\begin{proof}[Proof of \cref{propGabcd}]
\cref{lemmmmm} shows that the search for a general form of an SLOCC-equivalence between the states $\ket{G_{a bcd}}$ and $\ket{G_{a' b'c'd'}}$ might be restricted to the search within the finite class of operators $\mathcal{O} \in \mathcal{G}_{24}^{\otimes 4}$. Note that the group $\mathcal{G}_{24}$ has only 24 elements, and hence one may numerically verify that there are exactly $8\times 24 =192$ states in the $G_{abcd}$ family which are SLOCC-equivalent to $\ket{G_{abcd}}$ by $\mathcal{O}\in \mathcal{G}_{24}^{\otimes 4}$.
For example, the following tensor operation 
\begin{align}
\label{tuples1}
\mathbf{R}_x \Big(\frac{\pi}{2}\Big) \otimes
\mathbf{R}_x \Big(\frac{\pi}{2}\Big)\otimes
\mathbf{R}_x \Big(\frac{\pi}{2}\Big) \otimes
\mathbf{R}_x \Big(\frac{\pi}{2}\Big) &
\end{align}
transforms state $\ket{G_{abcd}}$ into $\ket{G_{-b-acd}}$.
This might be written as a transformation of a tuples of indices: the tuple $(a,b,c,d )$ is transformed into the tuple $(-b,-a,c,b)$. Similarly, the operators showed on the following right hand sides provide the corresponding transformations of the tuple $(a,b,c,d)$ on the left side:
\begin{align}
\label{tuples2}
\nonumber
\mathbf{R}_y \Big(\frac{\pi}{2}\Big) \otimes
\mathbf{R}_y \Big(\frac{\pi}{2}\Big)\otimes
\mathbf{R}_y \Big(\frac{\pi}{2}\Big) \otimes
\mathbf{R}_y \Big(\frac{\pi}{2}\Big) &\quad \longleftrightarrow \quad (a,\blue{d},c,\blue{b}), \\
\nonumber
\mathbf{R}_z \Big(\frac{\pi}{2}\Big) \otimes
\mathbf{R}_z \Big(\frac{\pi}{2}\Big)\otimes
\mathbf{R}_z \Big(\frac{\pi}{2}\Big) \otimes
\mathbf{R}_z \Big(\frac{\pi}{2}\Big) &\quad \longleftrightarrow \quad (\blue{-d},b,c,\blue{-a}), \\
\nonumber
\mathbf{R}_y ({\pi}) \otimes
\mathbf{R}_y ({\pi})\otimes
\bs{1} \otimes
\bs{1} &\quad \longleftrightarrow \quad
(a,\blue{-b},\blue{-c},d), \\
\nonumber
\mathbf{R}_x ({\pi}) \otimes
\mathbf{R}_x ({\pi})\otimes
\bs{1} \otimes
\bs{1} &\quad \longleftrightarrow \quad
(a,b,\blue{-c},\blue{-d}), \\
\nonumber
\mathbf{R}_y ({\pi}) \otimes
\bs{1} \otimes
\mathbf{R}_y ({\pi})\otimes
\bs{1} &\quad \longleftrightarrow \quad
(\blue{d},\blue{c},\blue{b},\blue{a}), \\
\mathbf{R}_x ({\pi}) \otimes
\bs{1} \otimes
\mathbf{R}_x ({\pi})\otimes
\bs{1} &\quad \longleftrightarrow \quad
(\blue{b},\blue{a},\blue{d},\blue{c})\,. \nonumber
\end{align}
In addition, the tuples $ (a,b,c,d)$ and $(-a,-b,-c,-d)$ represent the same state, since states are defined up to the global phase. Note that any composition of the above operations also provides SLOCC equivalences between states of the form $\ket{G_{abcd}}$. The eight aforementioned transformations of tuples generate all possible permutations of the $a,b,c,d$ indices, together with the change of a sign of any two or all four indices. Note that there are exactly $24$ permutations in total and for each permutation the signs can be matched in exactly $1+ {{4}\choose{2}} +1 =8$ ways. This gives $192$ tuples representing SLOCC equivalent states, which perfectly matches the numerical result.

Note yet another trivial manipulation with indices $a,b,c,d$ comes from multiplying by a global phase. This operation transforms the indices as
\[
(e^{i\theta} a,\;e^{i\theta} b,\; e^{i\theta} c,\; e^{i\theta} d ) \sim
(a,b,c,d)\,,
\]
resulting in the same quantum state for any real number $\theta \in [0,2 \pi)$. In particular, for $\theta=\pi$, the system of opposite indices determines the same state as the initial one, i.e. $(-a,-b,-c,-d)\sim (a,b,c,d)$.
\end{proof}


Finally, we show a link between normal systems of roots and the normal form of pure states of four qubits in the $G_{abcd}$ class. We say that a state is in its normal form if all its reductions are maximally mixed \cite{NormalFormVerstraete}. The process to determine the normal form of a state (if it exists) is straightforward, although, it may also turn out to be an infinite iterative process \cite{NormalFormVerstraete}. However, presented results of \cref{T1} applied to the four qubits states in the $G_{abcd}$ family show that this difficulty can be avoided. Indeed, consider the representative state of the $G_{abcd}$ family is in normal form \cite{NormalFormVerstraete} and the corresponding roots associated with the three-tangle $\tau^{(3)}$ measure also form a normal system. One may calculate the roots associated with any state in the $G_{abcd}$ class, and transform them into a normal system of roots using a Möbius transformation. Furthermore, the associated SLOCC operator to convert the initial state into a state in the normal form.
We illustrate this procedure by transforming the widely discussed four-partite state $\frac{1}{\sqrt{2}}(\ket{0}\ket{\textrm{GHZ}}+\ket{1}\ket{\textrm{W}})$ \cite{Osterloh3,OsterlohWernerStates} into its normal form, see \cref{MobiusTrans}. Without presented technique, the standard way of obtaining the normal form would result in an infinite iterative procedure.

\section{Conclusions}
In this chapter, I showed how a single $\SL$-invariant polynomial entanglement measure $\SLIP_N^h$ defined in $N$-quibit system can be used to derive necessary and sufficient conditions for any two $(N+1)$-qubit states to be equivalent under SLOCC operations. This was possible by showing that the roots of any $\SLIP_N^h$ measure transform via Möbius transformations under the SLOCC operations performed on the subsystems (\cref{T1}). In that way, SLOCC equivalence between two states is implied by the easily verifiable existence of a Möbius transformation relating the aforementioned roots for each subsystem, as it is described in a detailed way in the procedure presented in \cref{Aarhus17}. Furthermore, I demonstrated that the 3-tangle measure $\tau^{(3)}$ is enough to classify 4-qubit generic states (\cref{propGabcd}). Lastly, I present a procedure of determining the normal form of states in the $G_{abcd}$ family that circumvents the possibility of an infinite iterative process of the standard procedure.

\chapter{Concluding remarks}
\section{Conclusions}

This dissertation covers various aspects of the quantification and classification of entanglement in multipartite states and the role of symmetry in such systems. The main results of the research in this thesis can be summarized in the following.

\begin{enumerate}
\item 
\cref{chap1} introduces the notion of $m$-resistant states of $N$-partite systems and presents it by analogy with topological links (\cref{Aarhus1}). I present two general methods of constructing $m$-resistant states, one based on the Majorana representation of symmetric states (\cref{Aarhus4}), the other on the combinatorial objects known as orthogonal arrays (\cref{prop111,Aarhus5}).
%
%
\item 
In \cref{chap2}, I introduce the notion of $H$-symmetric states, where $H$ is an arbitrary subgroup of permutation group -- see \cref{Aarhus9}. \cref{Aarhus8} provides an elementary example of a quantum state with the given group of symmetry. 
\item 
I introduce two families of states with remarkable symmetric properties: in \cref{Aarhus10} excitation-states $\ket{G}$, related to a graph $G$, and in \cref{Dickelike} Dicke-like states, determined by an embedding of a subgroup $H< \mathcal{S}_N$. Furthermore, I investigate entanglement properties of excitation states related to regular graphs, in particular the form of $2$-partie concurrence (\cref{A}), and generalized concurrence (\cref{B}) of such states.
%
%
\item 
I present two different methods of constructing excitation states: in \cref{Hamiltonians} excitation states are obtained as ground states of Hamiltonians with $3$-body interactions, while in \cref{Circuit} as an outcome of a quantum circuit, with satisfactory complexity. A simple example of five qubit excitation is successfully simulated on available quantum computers: IBM - Santiago and Athens (\cref{cyclic_hist}).
\item 
In \cref{ch4}, I present the famous combinatorial problem of 36 officers that was posed by Euler, and its quantum version. I show an analytical form of an AME$(4,6)$ state, which might be seen as a quantum solution to the Euler's problem. Moreover, I present a coarse-grained combinatorial structure behind constructed AME$(4,6)$ state, and elaborate how such coarse-grained combinatorial structures might lead to the construction of other genuinely entangled states beyond the stabilizer approach.
%
%
\item 
\cref{chap5} presents general techniques of verification whether two $k$-uniform or AME states are SLOCC-equivalent. In particular, I show that SLOCC equivalence between two $k$-uniform states of $N$ particles with the minimal support has one of the following forms: if $k<2N$, it is provided by local monomial matrix (\cref{prop1,coro1}), while if $k=2N$ by a Butson-type complex Hadamard matrix or monomial matrix (\cref{prop1=}).
%
%
\item
I apply general results concerning the SLOCC equivalence of AME states. In particular, I show that some AME states cannot be locally transformed into existing AME states of minimal support (\cref{AME55}), which falsifies the conjecture that for a given multipartite quantum system all AME states are locally equivalent. Furthermore, I show that the existence of AME states with minimal support of 6 or more particles yields the existence of infinitely many such non-SLOCC-equivalent states (\cref{AMEminSUP6}). As an immediate consequence, I show that not all AME states belong to the class of stabilizer states.
%
%
\item 
In \cref{chap6}, I show how a single entanglement measure can be used to derive necessary and sufficient conditions for any two multipartite quantum qubit states to be equivalent under SLOCC operations. Precisely, I show that under an SLOCC operation the roots of any appropriate measure transform via related Möbius transformation (\cref{T1}). SLOCC equivalence between two states is implied by the easily verifiable existence of a Möbius transformation relating the aforementioned roots for each subsystem, as it is precisely described in the procedure presented in \cref{Aarhus17}.
%
%
\item 
I apply general results described in \cref{chap6} to $4$-qubit systems. In particular, I demonstrate that the 3-tangle measure is enough to classify 4-qubit generic states (\cref{propGabcd}). Lastly, I present a procedure of determining the normal form of a generic $4$-qubit state that circumvents the possibility of an infinite iterative process of the standard procedure.
\end{enumerate}

\section{Open problems}

As we have demonstrated so far, the quantification and classification of entanglement for multipartite states is an ambitious long-term project. We hope that the results presented in this dissertation will contribute to the development of these complex issues. In addition, we outline related problems and important open questions for future research.

\begin{enumerate}
\item[{A.}]
As we elaborated in \cref{chap1}, it is not clear if $m$-resistant states of $N$-qubit do exists for any $m=0,1,\ldots, N-2$, especially among symmetric states. One may also pose the more general problem of verification whether, for any number of parties $N$ and number $m$, there exists an $m$-resistant $N$-qudit state, see \cref{Aarhus6}.
%
%
\item[{B.}] 
Admissible groups of symmetries in multipartite quantum systems of $N$-qudits. For any subgroup of symmetric group $H< \mathcal{S}_N$, what is the minimal local dimension $d$ for which there exists a $H$-symmetric state $\ket{\psi} \in \mathcal{H}^{\otimes N}_d$, see \cref{Aarhus7}? In particular, which groups of symmetries are admissible in systems of $N$ qubits?
\item[{C.}] 
\cref{ch4} presents a coarse-grained combinatorial structure behind constructed AME$(4,6)$ state. Is there any general method for constructing such objects in any dimension? Based on such coarse-grained combinatorial structures, is it possible to provide a general construction of other genuinely entangled states beyond the stabilizer approach?
\item[{D.}] 
Uniqueness of AME states in small dimensions. Verify if all OAs of index unity with small local dimension $d<9$ are isomorphic by permutations of symbols on each level (\cref{productCon}). From this would follow that all related $k$-uniform states with minimal support and all phases equal are LU-equivalent.
\item[{E.}] 
Precise form of SLOCC-equivalence between AME($2k$,$d$) states. 
Which subclass of Butson-type complex Hadamard matrices might appear in the LU-equivalence between to AME($2k$,$d$) states of the minimal support? Verify if all AME($2k$,$d$) states with minimal support are LU-equivalent if and only if they are LM-equivalent (\cref{Aarhus16}).
%
%
\item[{F.}] 
Linear structures of orthogonal arrays and corresponding $k$-uniform and AME states. Verify if  AME states related to different linear structures of classical combinatorial designs are not equivalent. In particular, verify whether three AME(4,9) states presented in \cref{proppppp2} are SLOCC-equivalent. 
%
%
\item[{G.}] 
CKW-equalities \cite{DistributedEntanglement}. \cref{propGabcd} constitutes a new relation between 4-qubit invariants and the convex roof extension of 3-tangle. Is it possible  to understand presented relation in context of generalized CKW (in)equalities of four qubit states?
\item[{H.}] 
Classification of large multipartite systems. \cref{chap6} shows how generic 4-qubit states may be classified by the 3-tangle measure. Is it possible to provide such classification for generic $N$-qubit states, $N>4$, via a single entanglement measure?
\end{enumerate}

\newpage
\addcontentsline{toc}{chapter}{Bibliography}
\bibliographystyle{alpha}
\bibliography{PhD_Adam_Burchardt.bbl}

\appendix
\chapter*{Appendix \\ $\quad$ \\ List of selected publications and preprints}
\addcontentsline{toc}{chapter}{List of publications}

This PhD Dissertation is based on the following publications and preprints available online. As they present some further results not discussed here, I list them for the convenience of the reader.
\begin{enumerate}
\item[{[A]}]
G. M. Quinta, R. Andr\'{e}, \textbf{A. Burchardt}, and K. \. Zyczkowski
\textit{Cut-resistant links and multipartite entanglement resistant to particle loss}, 
\href{https://journals.aps.org/pra/abstract/10.1103/PhysRevA.100.062329}{Phys. Rev. A 100, 062329 (2019)}. 
\item[{[B]}]
\textbf{A. Burchardt}, Z. Raissi, 
\textit{Stochastic local operations with classical communication of absolutely maximally entangled states}, 
\href{https://journals.aps.org/pra/abstract/10.1103/PhysRevA.102.022413}{Phys. Rev. A 102,022413 (2020)}.
\item[{[C]}]
\textbf{A. Burchardt}, J. Czartowski, and K. \. Zyczkowski, 
\textit{Entanglement in highly symmetric multipartite quantum states},
\href{https://journals.aps.org/pra/abstract/10.1103/PhysRevA.104.022426}{Phys. Rev. A 104, 022426 (2021)}.
\item[{[D]}]
S. Rather${}^\ast$, \textbf{A. Burchardt}${}^\ast$, W. Bruzda, G. Rajchel-Mieldzio{\'c}, A. Lakshminarayan, K. {\.Z}yczkowski, 
\textit{Thirty-six entangled officers of Euler: Quantum solution to a classically impossible problem}, 
\href{https://journals.aps.org/prl/abstract/10.1103/PhysRevLett.128.080507}{Phys. Rev. Lett. 128, 080507 (2022)}. 
\\${}^\ast$Contributed equally
\item[{[E]}]
\textbf{A. Burchardt}, G. M. Quinta, R. Andr\'{e}, 
\textit{Entanglement Classification via Single Entanglement Measure}, 
\href{https://arxiv.org/abs/2106.00850}{ArXiv: 2106.00850 (2021)}.
\label{appE}
\end{enumerate}

\newpage

\chapter*{List of all publications and preprints}

\begin{enumerate}
\item[{[1]}]
M. Borkowski, D. Bugajewski, \textbf{A. Burchardt}, 
\textit{On Topological Properties of Metrics Defined via Generalized 'linking Construction'}, 
\href{https://www.hindawi.com/journals/jfs/2017/4901762/}{Journal of Function Spaces 2017,4901762 (2017)}.
\item[{[2]}]
W. K{\l{}}obus, \textbf{A. Burchardt}, A. Ko{\l{}}odziejski, M. Pandit, T. V\'{e}rtesi, K. {\.Z}yczkowski, W. Laskowski, 
\textit{k-uniform mixed states}, 
\href{https://journals.aps.org/pra/abstract/10.1103/PhysRevA.100.032112}{Phys. Rev. A 100,032112 (2019)}.
\item[{[3]}]
G. M. Quinta, R. Andr\'{e}, \textbf{A. Burchardt}, K. {\.Z}yczkowski, 
\textit{Cut-resistant links and multipartite entanglement resistant to particle loss}, 
\href{https://journals.aps.org/pra/abstract/10.1103/PhysRevA.100.062329}{Phys. Rev. A 100,062329 (2019)}.
\item[{[4]}]
\textbf{A. Burchardt}, 
\textit{Algebras with two multiplications and their cumulants}, 
\href{https://link.springer.com/article/10.1007/s10801-019-00898-3}{Journal of Algebraic Combinatorics 52,157-186 (2020)}.
\item[{[5]}]
\textbf{A. Burchardt}, Z. Raissi, 
\textit{Stochastic local operations with classical communication of absolutely maximally entangled states}, 
\href{https://journals.aps.org/pra/abstract/10.1103/PhysRevA.102.022413}{Phys. Rev. A 102,022413 (2020)}.
\item[{[6]}]
\textbf{A. Burchardt}, 
\textit{The top-degree part in the matchings-jack conjecture}, 
Electronic Journal of Combinatorics 28,P2.15 (2021).
\item[{[7]}]
\textbf{A. Burchardt}, J. Czartowski, K. {\.Z}yczkowski, 
\textit{Entanglement in highly symmetric multipartite quantum states},
\href{https://journals.aps.org/pra/abstract/10.1103/PhysRevA.102.022413}{Phys. Rev. A 104, 022426 (2021)}. 
\item[{[8]}]
J.Paczos, M. Wierzbi{\'{n}}ski, G. Rajchel-Mieldzio{\'{c}}, \textbf{A. Burchardt}, K. {\.Z}yczkowski, 
\textit{Genuinely quantum SudoQ and its cardinality}, 
\href{https://journals.aps.org/pra/abstract/10.1103/PhysRevA.104.042423}{Phys. Rev. A 104, 042423 (2021)}.
\item[{[9]}]
B. Yu, P. Jayachandran, \textbf{A. Burchardt}, Y. Cai, N. Brunner, V. Scarani, 
\textit{Absolutely entangled sets of pure states for bipartitions and multipartitions}, 
\href{https://journals.aps.org/pra/abstract/10.1103/PhysRevA.104.032414}{Phys. Rev. A 104, 032414 (2021)}.
\item[{[10]}]
S. Rather${}^\ast$, \textbf{A. Burchardt}${}^\ast$, W. Bruzda, G. Rajchel-Mieldzio{\'c}, A. Lakshminarayan, K. {\.Z}yczkowski, 
\textit{Thirty-six entangled officers of Euler: Quantum solution to a classically impossible problem}, 
\href{https://journals.aps.org/prl/abstract/10.1103/PhysRevLett.128.080507}{Phys. Rev. Lett. 128, 080507 (2022)}. 
\\${}^\ast$Contributed equally
\item[{[11]}]
\textbf{A. Burchardt}, G. M. Quinta, R. Andr\'{e}, 
\textit{Entanglement Classification via Single Entanglement Measure}, 
\href{https://arxiv.org/abs/2106.00850}{ArXiv: 2106.00850 (2021)}. 
\item[{[12]}]
Z. Raissi, \textbf{A. Burchardt}, E. Barnes, 
\textit{General stabilizer approach for constructing highly entangled graph states}, 
\href{https://arxiv.org/abs/2111.08045}{ArXiv: 2111.08045 (2021)}. 
\item[{[13]}]
K. {\.Z}yczkowski, W. Bruzda, G. Rajchel-Mieldzio{\'c}, \textbf{A. Burchardt}, S. Rather, A. Lakshminarayan, 
\textit{9$\times$4 = 6$\times$6: Understanding the quantum solution to the Euler's problem of 36 officers}, 
\href{https://arxiv.org/abs/2204.06800}{ArXiv: 2204.06800 (2022)}. 
\end{enumerate}

\renewcommand{\thesection}{\Alph{section}}

\end{document}